\documentclass[11pt]{article} \usepackage[normal]{optional}
\usepackage{amsmath}
\usepackage{amssymb}
\usepackage{ifpdf}
\usepackage{algorithmic}
\usepackage[caption=false,format=hang]{subfig}
\usepackage{wrapfig}
\usepackage{cite}
\usepackage{graphicx}
\usepackage{enumerate}
\usepackage{comment}
\usepackage{color}

\usepackage[textsize=scriptsize]{todonotes} 

\opt{final}{
    \usepackage{epsfig}
    \usepackage{graphicx}
    \usepackage{makeidx}
    \usepackage{multicol}
    \usepackage{subeqn}

    \usepackage{crop}
    \crop
    \makeindex
}

\opt{normal,submission}{\usepackage{commands-tam,numbering-tam}}

\usepackage{url}
\usepackage{amsfonts}
\opt{normal,submission}{\usepackage{amsthm}\usepackage{fullpage}}

\opt{submission,final}{\newcommand{\figuresize}{\footnotesize}}
\opt{normal}{\newcommand{\figuresize}{\small}}

\newcommand{\dN}{\mathsf{N}}
\newcommand{\dS}{\mathsf{S}}
\newcommand{\dE}{\mathsf{E}}
\newcommand{\dW}{\mathsf{W}}
\newcommand{\dall}{\{\dN,\dS,\dE,\dW\}}

\newcommand{\ha}{\widehat{\alpha}}

\newcommand{\stack}[2]{\genfrac{}{}{0pt}{}{#1}{#2}}

\newcommand{\ta}{\alpha}
\newcommand{\tb}{\beta}
\newcommand{\tg}{\gamma}

\newcommand{\vv}{\vec{v}}

\renewcommand{\exp}[1]{\mathrm{E}\left[ #1 \right]}

\renewcommand{\time}{\mathbf{T}}

\newcommand{\Kddac}{\mathrm{depth}^\mathrm{da}}

\newcommand{\ignore}[1]{}

\newcommand{\w}[1]{\widetilde{#1}}
\newcommand{\wa}{\w{\alpha}}

\ifpdf

  \usepackage[pdftex]{epsfig}
  \usepackage[colorlinks=true,citecolor=blue,linkcolor=blue]{hyperref}

\else
    \usepackage[dvips]{epsfig}
\usepackage[dvips,pdfstartview=FitH,pdfpagemode=None,colorlinks=true,citecolor=blue,linkcolor=blue]{hyperref}

\fi

\begin{document}

\opt{normal,submission}{
    \title{Parallelism and time in hierarchical self-assembly\footnote{A preliminary version of this article appeared as~\cite{CheDot12}.}}
    \date{}

    \author{
        Ho-Lin Chen\thanks{National Taiwan University, Taipei, Taiwan, {\tt holinc@gmail.com}. This author was supported by the Molecular Programming Project under NSF grant 0832824.}
        \and
        David Doty\thanks{University of California, Davis, Davis, CA, USA, {\tt doty@ucdavis.edu}. This author was supported by NSF grants 1219274, 1162589, 1619343, a Computing Innovation Fellowship under NSF grant 1019343, and the Molecular Programming Project under NSF grants 0832824 and 1317694.}
    }
}

\opt{final}{
    \chapter[Sample File for SIAM \LaTeX\ Book Macro Package, Proceedings Version]%
    {Sample File for SIAM \LaTeX\ Book Macro Package, Proceedings Version\thanks{Funding for this paper furnished
    by the American Taxpayers.}}
    \index{Sample!file}

    \begin{authorline}
    H.G.~Wells\thanks{Mr.~Wells.}, S.L.~Clemens\thanks{Mark Twain.}, and H.~Melville\thanks{Call me Ishmael.}
    \end{authorline}
}

\maketitle


\begin{abstract}
We study the role that parallelism plays in time complexity of variants of Winfree's abstract Tile Assembly Model (aTAM), a model of molecular algorithmic self-assembly.
In the ``hierarchical'' aTAM, two assemblies, both consisting of multiple tiles, are allowed to aggregate together, whereas in the ``seeded'' aTAM, tiles attach one at a time to a growing assembly.
Adleman, Cheng, Goel, and Huang (\emph{Running Time and Program Size for Self-Assembled Squares}, STOC 2001) showed how to assemble an $n \times n$ square in $O(n)$ time in the seeded aTAM using $O(\frac{\log n}{\log \log n})$ unique tile types, where both of these parameters are optimal.
They asked whether the hierarchical aTAM could allow a tile system to use the ability to form large assemblies in parallel before they attach to break the $\Omega(n)$ lower bound for assembly time.
We show that there is a tile system with the optimal $O(\frac{\log n}{\log \log n})$ tile types that assembles an $n \times n$ square using $O(\log^2 n)$ parallel ``stages'', which is close to the optimal $\Omega(\log n)$ stages, forming the final $n \times n$ square from four $n/2 \times n/2$ squares, which are themselves recursively formed from $n/4 \times n/4$ squares, etc.
However, despite this nearly maximal parallelism, the system requires superlinear time to assemble the square.
We extend the definition of \emph{partial order tile systems} studied by Adleman et al.\ in a natural way to hierarchical assembly and show that \emph{no} hierarchical partial order tile system can build \emph{any} shape with diameter $D$ in less than time $\Omega(D)$, demonstrating that in this case the hierarchical model affords no speedup whatsoever over the seeded model.
We also strengthen the $\Omega(D)$ time lower bound for deterministic seeded systems of Adleman et al.\ to nondeterministic seeded systems.
Finally, we show that for infinitely many $n$, a tile system can assemble an $n \times n'$ rectangle, with $n > n'$, in time $O(n^{4/5} \log n)$, breaking the linear-time lower bound that applies to all seeded systems and partial order hierarchical systems.
\end{abstract}

\opt{submission,final}{\thispagestyle{empty}\newpage\setcounter{page}{1} \clearpage}
\section{Introduction}
\label{sec-intro}

Tile self-assembly is an algorithmically rich model of ``programmable crystal growth''.
It is possible to design molecules (square-like ``tiles'') with specific binding sites so that, even subject to the chaotic nature of molecules floating randomly in a well-mixed chemical soup, they are guaranteed to bind so as to deterministically form a single target shape.
This is despite the number of different types of tiles possibly being much smaller than the size of the shape and therefore having only ``local information'' to guide their attachment.
The ability to control nanoscale structures and machines to atomic-level precision will rely crucially on sophisticated self-assembling systems that automatically control their own behavior where no top-down externally controlled device could fit.

A practical implementation of self-assembling molecular tiles was proved experimentally feasible in 1982 by Seeman~\cite{Seem82} using DNA complexes formed from artificially synthesized strands.
Experimental advances have delivered increasingly reliable assembly of algorithmic DNA tiles with error rates of 10\% per tile in 2004~\cite{RoPaWi04}, 1.4\% in 2007~\cite{FujHarParWinMur07}, 0.13\% in 2009~\cite{BarSchRotWin09}, and 0.05\% in 2014~\cite{EvansThesis}.
Erik Winfree~\cite{Winf98} introduced the abstract Tile Assembly Model (aTAM) -- based on a constructive version of Wang tiling~\cite{Wang61,Wang63} -- as a simplified mathematical model of self-assembling DNA tiles.
Winfree demonstrated the computational universality of the aTAM by showing how to simulate an arbitrary cellular automaton with a tile assembly system.
Building on these connections to computability, Rothemund and Winfree~\cite{RotWin00} investigated a self-assembly resource bound known as \emph{tile complexity}, the minimum number of tile types needed to assemble a shape.
They showed that for most $n$, the problem of assembling an $n \times n$ square has tile complexity $\Omega(\frac{\log n}{\log \log n})$, and Adleman, Cheng, Goel, and Huang~\cite{AdChGoHu01} exhibited a construction showing that this lower bound is asymptotically tight.
Under natural generalizations of the model~\cite{AGKS05, BeckerRR06, KaoSchS08, KS06, DDFIRSS07, Sum09, Dot10, ChaGopRei09, RNaseSODA2010, DemPatSchSum2010RNase, SolWin07, ManStaSto09ISAAC}, tile complexity can be reduced for tasks such as square-building and assembly of more general shapes.
See~\cite{DotCACM, PatitzSurveyJournal, winslow2016brief} for more background.

The authors of~\cite{AdChGoHu01} also investigated assembly time for the assembly of $n \times n$ squares in addition to tile complexity.
They define a plausible model of assembly time based (implicitly) on the standard stochastic model of well-mixed chemical kinetics~\cite{Gillespie77,GibsonBruck00,Gillespie07} and show that under this model, an $n \times n$ square can be assembled in expected time $O(n)$, which is asymptotically optimal, in addition to having optimal tile complexity $O(\frac{\log n}{\log \log n})$.
Intuitively, the optimality of the $O(n)$ assembly time for an $n \times n$ square results from the following informal description of self-assembly.
The standard ``seeded'' aTAM stipulates that one tile type is designated as the seed from which growth nucleates, and all growth occurs by the accretion of a single tile to the assembly containing the seed.
The set of locations on an assembly $\alpha$ where a tile could attach is called the \emph{frontier}.
An assembly with a frontier of size $k$ could potentially have $\Theta(k)$ attachment events occur in parallel in the next ``unit'' of time, meaning that a speedup due to parallelism is possible in the seeded aTAM.
The geometry of 2D assembly enforces that any assembly with $N$ points has an ``average frontier size'' throughout assembly of size at most $O(\sqrt{N})$.\opt{normal}{\footnote{For intuition, picture the fastest growing assembly: a single tile type able to bind to itself on all sides, filling the plane starting from a single copy at the origin.  After $t$ ``parallel steps'', with high probability it has a circumference, and hence frontier size, of $O(t)$, while occupying area $O(t^2)$.}}
Therefore, the parallelism of the seeded aTAM grows at most linearly with time.
To create an $n \times n$ square of size $n^2$, the best parallel speedup that one could hope for would use an ``average frontier size'' of $O(n)$, which in $O(n)$ ``parallel steps'' of time assembles the entire square.
This is precisely the achievement of~\cite{AdChGoHu01}.

A variant of the aTAM known as the \emph{hierarchical} (a.k.a. \emph{two-handed}, \emph{recursive}, \emph{multiple tile}, \emph{$q$-tile}, \emph{aggregation}, \emph{polyomino}) aTAM allows non-seed tiles to aggregate together into an assembly, allows this assembly to then aggregate to other assemblies, and possibly (depending on the model) dispenses completely with the idea of a seed.
Variants of the hierarchical aTAM have recently received extensive theoretical study~\cite{AGKS05, DDFIRSS07, Winfree06, Luhrs08, AdlCheGoeHuaWas01, Adl00, RNaseSODA2010, DotPatReiSchSum10, PatSumIdentifyingJournal, DemPatSchSum2010RNase, PadLiuSee11FNANO, FuPatSchShe12, TwoHandedNotIntrinsicallyUniversalconf}.
It is intuitively conceivable that by allowing two large assemblies to form in parallel and combine in one step, it may be possible to recursively build an $n \times n$ square in $o(n)$ time, perhaps even $O(\log n)$ or $O(\mathrm{polylog}(n))$ time.
In the terminology of Reif \cite{ReifLocalParallel}, such parallelism is ``distributed'' rather than ``local.''
Determining the optimal time lower bound for uniquely self-assembling an $n \times n$ square in the hierarchical aTAM was stated as an open problem in~\cite{AdChGoHu01}.

We achieve three main results.
We prove that no ``partial order hierarchical system'' (defined below) can break the $\Omega(D)$ lower bound for assembling any shape of diameter $D$.
Next, we show that a hierarchical system violating the ``partial order'' property is able to assemble a rectangle of diameter $D$ in time $o(D)$.
Finally, we show a highly parallel (but surprisingly, slow) assembly of an $n \times n$ square in a hierarchical system.
We now discuss these results in more detail.


Section~\ref{sec-runtime} defines our model of assembly time for hierarchical tile systems.
To obtain a fair comparison between our main result, Theorem~\ref{thm-hier-partial-order-slow}, and the results for assembly time in the seeded model~\cite{AdChGoHu01}, it is necessary to introduce a definition of assembly time applicable to both seeded and hierarchical tile systems.
Defining this concept is nontrivial and constitutes one of the contributions of this paper.
We define such an assembly time model based on chemical kinetics.
When applied to seeded systems, the model results in (nearly) the same definition used in~\cite{AdChGoHu01}, in the limit of low concentration of seed tiles.\footnote{Low seed concentration is required to justify the assumption used in~\cite{AdChGoHu01} of constant concentration of non-seed tiles, so we are not ``cheating'' by using this assumption to argue that the models nearly coincide on seeded systems.  The one sense in which the models are different for seeded systems is that tile concentrations are allowed to deplete in our model. As we argue in Section~\ref{sec-time-complexity-hierarchical}, this difference does not account for our time lower bound.  Furthermore, this difference makes our model strictly more realistic than the model of~\cite{AdChGoHu01}. Tile systems in which this difference would be noticeable are those in which large assemblies not containing the seed can form, which are assumed away in the seeded model.  Such systems are precisely those for which the assumptions of the seeded model are not justified. This is the sense in which our model of assembly time coincides with that of~\cite{AdChGoHu01} when applied to the seeded model: it coincides with a slightly more realistic generalization of the model used in~\cite{AdChGoHu01}.}

Section~\ref{sec-time-complexity-lower-bound} shows our main result, Theorem~\ref{thm-hier-partial-order-slow}, a linear-time lower bound on a certain class of hierarchical tile systems.
In~\cite{AdChGoHu01} the authors define a class of deterministic seeded tile systems known as \emph{partial order systems}, which intuitively are those systems that enforce a precedence relationship (in terms of time of attachment) between any neighboring tiles in the unique terminal assembly that bind with positive strength.
We extend the definition of partial order systems in a natural way to hierarchical systems, and for this special case of systems, we answer the question of~\cite{AdChGoHu01} \emph{negatively}, showing that $\Omega(D)$ time is required to assemble any structure with diameter $D$.
Thus, for the purpose of speeding up self-assembly of partial order systems, the parallelism of the hierarchical assembly model is of no use whatsoever.

Section~\ref{sec-fast-rectangle} shows that the partial order hypothesis is necessary to obtain a linear-time lower bound.
There, we describe a hierarchical tile system that, according to our model of assembly time, can assemble a rectangle in time sublinear in its diameter.
More precisely, we show that for infinitely many $n$, there is a hierarchical tile system that assembles an $n \times n'$ rectangle, where $n > n'$, in time $O(n^{4/5} \log n)$.
The key idea is the use of both ``assembly parallelism'' and ``binding parallelism.''
By ``assembly parallelism,'' we mean the form of parallelism discussed above: the ability of the hierarchical model to form multiple large assemblies independently in parallel.
By ``binding parallelism,'' we mean the (much more modest) parallelism already present in the seeded model: the ability of a single tile or assembly to have multiple potential binding sites to which to attach on the ``main'' growing assembly.
If there are $k$ such binding sites, the \emph{first} such attachment will occur in expected time $\frac{1}{k}$ times that of the expected time for any \emph{fixed} binding site to receive an attachment, a fact exploited in our tile system to achieve a speedup.
We note that Theorem~\ref{thm-seeded-time-lower-bound} implies that ``binding parallelism'' alone --- i.e., the seeded model ---  cannot achieve assembly time sublinear in the diameter of the shape.

Finally, in Section~\ref{sec-hierarchical-square}, we show that in the hierarchical aTAM, it is possible to assemble an $n \times n$ square using nearly maximal ``parallelism,'' so that the full $n\times n$ square is formed from four $n/2 \times n/2$ sub-squares, which are themselves each formed from four $n/4 \times n/4$ sub-squares, etc.\footnote{If one were to assume a constant time for any two producible assemblies to bind once each is produced, this would imply a polylogarithmic time complexity of assembling the final square.
But accounting for the effect of assembly \emph{concentrations} on binding rates in our assembly time model, the construction takes superlinear time.
This is because some sub-square has concentration at most $\widetilde{O}(1/{n}^2)$, so the time for even a single step of hierarchical assembly is at least $\widetilde{\Omega}(n^2)$ by standard models of chemical kinetics.
We note, however, that there are other theoretical advantages to the hierarchical model, for instance, the use of steric hindrance to enable algorithmic fault-tolerance~\cite{DotPatReiSchSum10}.
For this reason, our highly parallel square construction may be of independent interest despite the fact that the parallelism does not confer a speedup.}
Informally, if tile system $\calT$ uniquely self-assembles a shape $S$, define $\Kddac(\calT)$ to be the worst-case ``number of parallel assembly steps'' (depth of the tree that decomposes the final assembly recursively into the subassemblies that combined to create it) required by the tile system to reach its final assembly.
\opt{normal}{(A formal definition is given in Section~\ref{sec-hierarchical-square}.)}
Clearly $\Kddac(\calT) \geq \log |S|$ if $S$ is the shape assembled by $\calT$.
Our construction is quadratically close to this bound in the case of assembling an $n \times n$ square $S_n$, showing that $\Kddac(\calT) \leq O(\log^2 n)$.
Furthermore, this is achievable using $O(\frac{\log n}{\log \log n})$ tile types, which is asymptotically optimal.\opt{normal}{\footnote{Without any bound on tile complexity, the problem would be trivialized by using a unique tile type for each position in the shape, each equipped with specially placed strength-1 bonds, similar to the ``inter-block'' bonds of Figure~\ref{fig:two-handed-square-high-level}, to ensure a logarithmic-depth assembly tree.}}
That is, not only is it the case that every producible assembly can assemble into the unique terminal assembly (by the definition of unique assembly), but in fact every producible assembly is at most $O(\log^2 n)$ attachment events from becoming the terminal assembly.

Section~\ref{sec-runtime} is required to understand Sections~\ref{sec-time-complexity-lower-bound} and~\ref{sec-fast-rectangle}, but Sections~\ref{sec-time-complexity-lower-bound} and~\ref{sec-fast-rectangle} can be read independently of each other.
Section~\ref{sec-hierarchical-square} can be read independently of Sections~\ref{sec-runtime},~\ref{sec-time-complexity-lower-bound}, and~\ref{sec-fast-rectangle}.

\section{Informal description of the abstract tile assembly model}
\label{sec-tam-informal}

This section gives a brief informal sketch of the seeded and hierarchical variants of the abstract Tile Assembly Model (aTAM).
\opt{normal}{See Section \ref{sec-tam-formal} for a formal definition of the aTAM.}

A \emph{tile type} is a unit square with four sides, each consisting of a \emph{glue label} (often represented as a finite string) and a nonnegative integer \emph{strength}.
We assume a finite set $T$ of tile types, but an infinite number of copies of each tile type, each copy referred to as a \emph{tile}. An \emph{assembly}
(a.k.a., \emph{supertile})
is a positioning of tiles on the integer lattice $\Z^2$; i.e., a partial function $\alpha:\Z^2 \dashrightarrow T$. 
Write $\alpha \sqsubseteq \beta$ to denote that $\alpha$ is a \emph{subassembly} of $\beta$, which means that $\dom\alpha \subseteq \dom\beta$ and $\alpha(p)=\beta(p)$ for all points $p\in\dom\alpha$.
In this case, say that $\beta$ is a \emph{superassembly} of $\alpha$.
We abuse notation and take a tile type $t$ to be equivalent to the single-tile assembly containing only $t$ (at the origin if not otherwise specified).
Two adjacent tiles in an assembly \emph{interact} if the glue labels on their abutting sides are equal and have positive strength. 
Each assembly induces a \emph{binding graph}, a grid graph whose vertices are tiles, with an edge between two tiles if they interact.
The assembly is \emph{$\tau$-stable} if every cut of its binding graph has strength at least $\tau$, where the weight of an edge is the strength of the glue it represents.
That is, the assembly is stable if at least energy $\tau$ is required to separate the assembly into two parts.
The \emph{frontier} $\partial\alpha \subseteq \Z^2 \setminus \dom\alpha$ of $\alpha$ is the set of empty locations adjacent to $\alpha$ at which a single tile could bind stably.

A \emph{seeded tile assembly system} (seeded TAS) is a triple $\calT = (T,\sigma,\tau)$, where $T$ is a finite set of tile types, $\sigma:\Z^2 \dashrightarrow T$ is a finite, $\tau$-stable \emph{seed assembly}, 
and $\tau$ is the \emph{temperature}.
An assembly $\alpha$ is \emph{producible} if either $\alpha = \sigma$ or if $\beta$ is a producible assembly and $\alpha$ can be obtained from $\beta$ by the stable binding of a single tile.
In this case write $\beta\to_1 \alpha$ ($\alpha$ is producible from $\beta$ by the attachment of one tile), and write $\beta\to \alpha$ if $\beta \to_1^* \alpha$ ($\alpha$ is producible from $\beta$ by the attachment of zero or more tiles).
An assembly is \emph{terminal} if no tile can be $\tau$-stably attached to it.

A \emph{hierarchical tile assembly system} (hierarchical TAS) is a pair $\calT = (T,\tau)$, where $T$ is a finite set of tile types and $\tau \in \N$ is the temperature.
An assembly is \emph{producible} if either it is a single tile from $T$, or it is the $\tau$-stable result of translating two producible assemblies without overlap.
An assembly $\alpha$ is \emph{terminal} if for every producible assembly $\beta$, $\alpha$ and $\beta$ cannot be $\tau$-stably attached.
The restriction on overlap is a model of a chemical phenomenon known as \emph{steric hindrance}~\cite[Section 5.11]{WadeOrganicChemistry91} or, particularly when employed as a design tool for intentional prevention of unwanted binding in synthesized molecules, \emph{steric protection}~\cite{HellerPugh1,HellerPugh2,GotEtAl00}.

In either the seeded or hierarchical model, let $\prodasm{\calT}$ be the set of producible assemblies of $\calT$, and let $\termasm{\calT} \subseteq \prodasm{\calT}$ be the set of producible, terminal assemblies of $\calT$.
A TAS $\calT$ is \emph{directed} (a.k.a., \emph{deterministic}, \emph{confluent}) if $|\termasm{\calT}| = 1$.


\section{Time complexity in the hierarchical model}
\label{sec-runtime}

In this section we define a formal notion of time complexity for hierarchical tile assembly systems. 
The model we use applies to both the seeded aTAM and the hierarchical aTAM.

\newcommand{\unrealisticHierarchical}{
    For hierarchical systems, our assembly time model may not be completely suitable since we make some potentially unrealistic assumptions.
    In particular, we ignore diffusion rates of molecules based on size and assume that large assemblies diffuse as fast as individual tiles.
    We also assume that the binding energy $\tau$ necessary for a small tile $t$ to attach stably to an assembly $\alpha$ is the same as the binding energy required for a large assembly $\beta$ to attach stably to $\alpha$, even though one would expect such large assemblies to have a higher reverse rate of detachment (slowing the net rate of forward growth) if bound with only strength $\tau$.
    However, from the perspective of our lower bound on assembly time, Theorem~\ref{thm-hier-partial-order-slow}, these assumptions have the effect of making hierarchial self-assembly appear \emph{faster}.
    We show that even with these extra assumptions, the time complexity of hierarchical partial order systems is \emph{still} no better than the seeded aTAM.
    However, caution is warranted in interpreting the upper bound result, Theorem~\ref{thm-fast-rectangle}, of a sublinear time assembly of a shape.
    As we discuss in Section~\ref{sec-conclusion}, a plausible treatment of diffusion rates -- together with our lower bound techniques based on low concentrations of large assemblies -- may yield an absolute linear-time (in terms of diameter) lower bound on assembly time of hierarchical systems, so that Theorem~\ref{thm-fast-rectangle} may owe its truth entirely to the heavily exploited assumption of equally fast diffusion of all assemblies.
    A reasonable interpretation of Theorem~\ref{thm-fast-rectangle} is that the partial order assumption is necessary to prove Theorem~\ref{thm-hier-partial-order-slow} and that concentration arguments alone do not suffice to establish linear-time time lower bounds in general hierarchical systems.
    The techniques that weave together both ``assembly parallelism'' and ``binding parallelism'', as discussed in Section~\ref{sec-intro} and Section~\ref{sec-fast-rectangle}, may prove useful in other contexts, even though their attained speedup is modest.
}

\opt{normal}{\unrealisticHierarchical}

\subsection{Definition of time complexity of seeded tile systems}
\label{sec-seeded-time-defn}

We now review the definition of time complexity of seeded self-assembly proposed in~\cite{AdChGoHu01}.
A \emph{concentrations function} on a tile set $T$ is a subprobability measure $C:T \to [0,1]$ (i.e., $\sum_{r\in T} C(r) \leq 1$).
Each tile type $r$ is assumed to be held at a fixed concentration $C(r)$ throughout the process of assembly.\footnote{For singly-seeded tile systems in which the seed tile $s \in T$ appears only once at the origin, this assumption is valid in the limit of low seed concentration $C(s)$ compared to all other concentrations $C(r)$ for $r \in T \setminus \{s\}$.
\opt{normal}{This is because the number of terminal assemblies (if each is of size at most $K$) will be limited by $C(s)$, implying the percentage change in every other tile type $r$'s concentration is at most $K \cdot C(s) / C(r)$; therefore ``low'' seed concentration means setting $C(s) \ll C(r) / K$ for all $r \in T \setminus\{s\}$.
In fact, to obtain an assembly time asymptotically as fast, one need only ensure that for all $r$, $C(r) \geq 2 \#_{\ha}(r) C(s)$, where $\#_{\ha}(r)$ is the number of times $r$ appears in the terminal assembly $\ha$.
This guarantees that the concentration of $r$ is always at least half of its start value, which means that the assembly time, each step of which is proportional to the concentration of the tile type attaching at that step, is at most doubled compared to the case when the concentrations are held constant.}}
The assembly time for a seeded TAS $\calT=(T,\sigma,\tau)$ is defined by picking a copy of the seed arbitrarily and measuring the expected time before the seed grows into some terminal assembly, when assembly proceeds according to the following stochastic model.
The assembly process is described as a continuous-time Markov process in which each state represents a producible assembly, and the initial state is the seed assembly $\sigma$.
For each pair of producible assemblies $\alpha,\beta$ such that $\alpha \to_1 \beta$ via the addition of tile type $r$, there is a transition in the Markov process from state $\alpha$ to state $\beta$ with transition rate $C(r)$.\opt{normal}{\footnote{That is, the expected time until the next attachment of a tile to $\alpha$ is an exponential random variable with rate $\sum_{r \in T} \sum_{p\in\partial^r\alpha} C(r)$, where $\partial^r\alpha$ is the $r$-frontier of $\alpha$, the set of empty locations at which tile tile $r$ could stably attach to $\alpha$. Note that if $r$ could attach at more than one location then this corresponds to separate terms in the sum; similarly, if one location $p$ can have multiple tile types attach to it, these also correspond to separate terms in the sum.}}
The sink states of the Markov process are precisely the terminal assemblies.
The time to reach some terminal assembly from $\sigma$ is a random variable $\time_{\calT,C}$, and the assembly time complexity of the seeded TAS $\calT$ with concentrations $C$ is defined to be $\mathsf{T}(\calT,C) = \exp{\time_{\calT,C}}$.

The requirement that the tile concentrations function $C$ be a subprobability measure, rather than an arbitrary measure taking values possibly greater than 1, reflects a physical principle known as the \emph{finite density constraint}, which stipulates that a given unit volume of solution may contain only a bounded number of molecules (if for no other reason than to avoid forming a black hole).
By normalizing so that one ``unit'' of volume is the volume required to fit one tile, the total concentration of tiles (concentration defined as  number or mass per unit volume) cannot exceed 1.\opt{normal}{\footnote{When our goal is to obtain only an asymptotic result concerning a family of tile systems assembling a family of assemblies of size/diameter $D$, we may relax the finite density constraint to the requirement that the concentrations sum to a constant $c\in\R_{\geq 0}$ independent of $D$, since these concentrations could be divided by $c$ to sum to 1 while affecting the assembly time results by the same constant $c$, leaving the asymptotic results unaffected.}}

We have the following time complexity lower bound for seeded systems.
This theorem says that even for non-directed systems, a seeded TAS can grow its diameter only linearly with time.
It strengthens and implies Lemma 4.6 of the full version of~\cite{AdChGoHu01}, which applied only to directed systems.
\opt{submission}{
}

\newcommand{\seededLowerBoundFormal}{
    Let $d \in \Z^+$.
    Let $\calT=(T,\sigma,\tau)$ be a singly-seeded TAS (meaning $|\sigma|=1$), and let $C:T\to[0,1]$ be a concentrations function.
    Since it takes only constant time for the assembly to grow to any constant radius, restricting attention to singly-seeded systems does not asymptotically affect the result for tile systems with a finite seed assembly of size larger than 1.
    Assume below that with probability 1, $\calT$ eventually places a tile at distance $d$ (in the $L_1$ norm) from the seed.
    Define $\mathbf{D}(\calT,C,d)$ be the random variable representing the time that any tile is first placed at distance $d$.
}


\opt{normal}{\seededLowerBoundFormal}

\newcommand{\seededLowerBoundFormalStatement}{For each $d \in \Z^+$, each singly-seeded TAS $\calT$, and each concentrations function $C:T\to[0,1]$, $\exp{\mathbf{D}(\calT,C,d)} = \Omega(d)$.}

\newcommand{\seededLowerBoundInformalStatement}{For any seeded TAS $\calT$ and any concentrations function on $\calT$, the expected time at which the assembly formed by $\calT$ grows to diameter $d$ is $\Omega(d)$.}

\begin{theorem}\label{thm-seeded-time-lower-bound}
\opt{normal}{\seededLowerBoundFormalStatement}
\opt{submission}{\seededLowerBoundInformalStatement}
\end{theorem}

\newcommand{\ProofSeededTimeLowerBound}{
    The intuition of the proof is as follows.
    We divide the plane into concentric ``layers'', with layer $i$ being the set of points at $L_1$-distance $i$ from the origin.
    We examine at the rate at which tiles are added to layer $i$, noting that such additions can only happen because of attachment to existing tiles in adjacent layers $i-1$ and $i+1$.
    Therefore, the rate of attachment in layer $i$ is proportional to the number of tiles in layers $i-1$ and $i+1$.
    This turns out to be a process with the property that the time at which layer $d$ gets its first tile is $\Omega(d)$, which we prove by solving some differential equations that bound the attachment process.

    Since we care only about the first time at which a tile is attached at distance $d$ (before which there are no tiles at distance $d'$ for any $d' \geq d$), we can restrict the assembly process to the region of radius $d$ around the seed.
    Therefore we model the assembly process as if it proceeds normally until the first tile attaches at distance $d$ from the seed, at which point all growth immediately halts.

    Define $\R_{\geq 0} = [0,\infty)$.
    Given $i \in \{0,\ldots,d\}$ and $t \in \R_{\geq 0}$, let $\mathbf{X}_i(t)$ be a random variable denoting the number of tiles attached at locations with distance exactly $i$ from the seed at time $t$, under the restriction stated above that all assembly halts the moment that a tile is placed at distance $d$.
    Then for all $t\in\R_{\geq 0}$, the event $\mathbf{X}_d(t) = 0$ (no tile is at distance $d$ by the time $t$) is equivalent to the event $\mathbf{D}(\calT,C,d) > t$ (the time of the first attachment at distance $d$ strictly exceeds $t$).

    In a seeded TAS, tiles can attach at a location only when there is another tile adjacent to the location.
    Locations at $L_1$-distance $i$ to the seed are only adjacent to locations at distance either $i+1$ or $i-1$ to the seed.
    Off the $x$- and $y$-axes, each location at distance $i$ has two neighbors at distance $i-1$ and two neighbors at distance $i+1$, and for the 4 locations at distance $i$ on either axis, every location has one neighbor at distance $i-1$ and three neighbors at distance $i+1$.
    Therefore, at time $t$, tiles are attachable to at most $2 \mathbf{X}_{i-1}(t) + 3 \mathbf{X}_{i+1}(t)$ different locations with distance $i$ to the seed.
    Since the total concentration of any single tile type is at most $1$, the rate at which tiles attach at any given location is at most $1$.
    For all $i\in\{0,\ldots,d\}$, define the function $f_i:\R_{\geq 0} \to \R_{\geq 0}$ for all $t\in\R_{\geq 0}$ by $f_i(t) = \exp{\mathbf{X}_i(t)}$.
    Then for $i\in\{1,\ldots,d-1\}$ and $t \in \R_{\geq 0}$,
    \begin{eqnarray*}
    \frac{df_i(t)}{dt} &\leq& 2f_{i-1}(t)\ +\ 3f_{i+1}(t),
    \\
    \frac{df_0(t)}{dt} &=& 0, \text{ and }
    \\
    \frac{df_d(t)}{dt} &\leq& 2f_{d-1}(t).
    \end{eqnarray*}
    The lack of a $3f_{d+1}(t)$ term in the latter inequality is due to our modification of the assembly process to immediately halt once the first tile attaches at distance $d$, which implies that $f_{d+1}(t) = 0$ for all $t \in \R_{\geq 0}$ since no tile is ever placed at distance $d+1$.
    Since the assembly process always starts with a single seed tile, $f_0(t)=1$ for all $t\in\R_{\geq 0}$, and $f_i(0) = 0$ for all $i \in \{1,\ldots,d\}$.
    For all $t\in\R_{\geq 0}$ and all $i\in\{1,\ldots,d\}$, $f_i(t) \leq 4i$ since there are exactly $4i$ locations at distance exactly $i$ to the seed.

    Let $t_0 \in \R_{\geq 0}$ be the unique time at which $f_d(t_0) = \frac{1}{2}.$
    This time is unique since $f_d$ is monotonically increasing (since tiles cannot detach).
    Since $\exp{\mathbf{X}_d(t_0)} = f_d(t_0) = \frac{1}{2}$, by Markov's inequality, $\prob[\mathbf{X}_d(t_0) \geq 1] \leq \frac{1}{2}$, implying that $\prob[\mathbf{X}_d(t_0) < 1] > \frac{1}{2}$.
    Since $\mathbf{X}_d$ is integer-valued and nonnegative, this is equivalent to stating that $\prob[\mathbf{X}_d(t_0) = 0] > \frac{1}{2}$.
    Recall that $\mathbf{X}_d(t_0) = 0 \iff \mathbf{D}(\calT,C,d) > t_0$,  whence $\prob[\mathbf{D}(\calT,C,d) > t_0] > \frac{1}{2}$.
    By Markov's inequality, $\exp{\mathbf{D}(\calT,C,d)} > \frac{t_0}{2}$.
    Thus it suffices to prove that $t_0 \geq \Omega(d)$.
    To do this, we define a simpler function that is an upper bound for $f_d$ and solve its differential equations.

    For all $i\in\{0,\ldots,d\}$, define the function $g_i:\R_{\geq 0} \to \R_{\geq 0}$ (which will serve as an upper bound for $f_i$) as follows.
    For all $1\in\{1,\ldots,d-1\}$ and $t \in \R_{\geq 0}$,
    \begin{eqnarray*}
    \frac{dg_i(t)}{dt} &=& 2g_{i-1}(t)\ +\ 3g_{i+1}(t),\ \mbox{when $g_i(t) < 4d$},
    \\
    \frac{dg_d(t)}{dt} &=& 2g_{d-1}(t),\ \mbox{when $g_d(t) < 4d$},
    \\
    \frac{dg_0(t)}{dt} &=& 0,
    \end{eqnarray*}
    and for all $i\in\{1,\ldots,d\}$,
    \[
    \frac{dg_i(t)}{dt}\ =\ 0,\ \mbox{when $g_i(t) = 4d$},
    \]
    with the boundary conditions $g_0(t)=1$ for all $t \in \R_{\geq 0}$, $g_i(0) = 0$ for all $i \in \{1,\ldots,d\}$.
    Other than the inequalities governing $f_i$ being changed to equality with $g_i$, the other difference with $f_i$ is that $g_i$ is allowed to grow larger than $4i$ (but no larger than $4d$, which applies also to $f_i$ whenever $i \leq d$).
    As a result, $g_i(t) \geq f_i(t)$ for all $i\in\{0,\ldots,d\}$ and $t\in\R_{\geq 0}$.

    Furthermore, if $g_i(t_0) > g_{i+1}(t_0)$ for all $i\in\{0,\ldots,d-1\}$ at some time $t_0 \in \R_{\geq 0}$, then
    \[
    \frac{dg_i(t)}{dt}\ \geq\ \frac{dg_{i+1}(t)}{dt}\ \mbox{at time $t_0$}.
    \]

    Since $g_i(0) \geq g_{i+1}(0)$ for all $i\in\{0,\ldots,d\}$ by definition, the above inequality implies that $g_i(t) \geq g_{i+1}(t)$ for all $i\in\{0,\ldots,d\}$ and all $t\in\R_{\geq 0}$.
    Using the fact that $g_{i-1}(t) \geq g_{i+1}(t)$, we have $2 g_{i-1}(t) + 3 g_{i+1}(t) \leq 5 g_{i-1}(t)$.
    Thus, we can define a set of functions $h_i(t)$ that are upper bounds for $g_i(t)$ by the following:
    \[
    \frac{dh_i(t)}{dt}\ =\ 5h_{i-1}(t) \mbox{, for all $i\in\{1,\ldots,d\}$}, \text{ and } \frac{dh_0(t)}{dt} = 0,
    \]
    with boundary conditions $h_0(t)=1$ for all $t\in\R_{\geq 0}$, $h_i(0) = 0$ for all $i\in\{1,\ldots,d\}$. Solving these differential equations, we obtain $h_d(t)=\frac{1}{d!}(5t)^d$.  Letting $t'=\frac{d}{10e}$, by Stirling's inequality $d! > \sqrt{2 \pi d} \left( \frac{d}{e} \right)^d e^{1 / (12 d + 1)} > \left( \frac{d}{e} \right)^d$, we have
    \[
    f_d\left(t'\right)
    \leq
    g_d\left(t'\right)
    \leq
    h_d\left(t'\right)
    =
    \frac{1}{d!} \cdot \left(5t'\right)^d
    =
    \frac{1}{d!} \cdot \left( \frac{d}{2 e} \right)^d
    <
    \frac{1}{\left( \frac{d}{e} \right)^d} \cdot \left( \frac{d}{2 e} \right)^d
    =
    \frac{1}{2^d}.
    \]
    Since $f_d$ is monotonically increasing, $f_d(t_0) = \frac{1}{2}$ by definition, and $\frac{1}{2^d} \leq \frac{1}{2}$ for $d \geq 1$, this implies that $t_0 \geq t' = \frac{d}{10e}$.
}

\opt{normal}{
    \begin{proof}
      \ProofSeededTimeLowerBound
    \end{proof}
}

\subsection{Definition of time complexity of hierarchical tile systems}
\label{sec-time-complexity-hierarchical}

\subsubsection{Issues with defining hierarchical time complexity}
To define time complexity for hierarchical systems, we employ more explicitly the chemical kinetics that implicitly underlie the time complexity model for seeded systems stated in Section~\ref{sec-seeded-time-defn}.
We treat each assembly as a single molecule. If two assemblies $\ta$ and $\tb$ can attach to create an assembly $\tg$, then we model this as a chemical reaction $\ta + \tb \to \tg$, in which the rate constant is assumed to be equal for all reactions (and normalized to 1).
In particular, if $\ta$ and $\tb$ can be attached in two different ways, this is modeled as two different reactions, even if both result in the same assembly.\opt{normal}{\footnote{The fact that some directed systems may not require at least one of these attachments to happen in every terminal assembly tree is the reason we impose the partial order requirement when proving our time complexity lower bound.}}

At an intuitive level, the model we define can be explained as follows.
We imagine dumping all tiles into the solution at once, and at the same time, we grab one particular tile and dip it into the solution as well, pulling it out of the solution when it has assembled into a terminal assembly.
Under the seeded model, the tile we grab will be a seed, assumed to be the only copy in solution (thus requiring that it appear only once in any terminal assembly).
In the seeded model, no reactions occur other than the attachment of individual tiles to the assembly we are holding.
In the hierarchical model, other reactions are allowed to occur in the background (we model this using the standard mass-action model of chemical kinetics~\cite{epstein1998introduction}), but only those reactions with the assembly we are ``holding'' move it closer to completion.
The other background reactions merely change concentrations of other assemblies (although these indirectly affect the time it will take our chosen assembly to complete, by changing the rate of reactions with our chosen assembly).

\newcommand{\modelJustification}{
    We now discuss some intuitive justification of our model of assembly time.
    One reason for choosing this model is that we would like to analyze the assembly time in such a way as to facilitate direct comparison with the results of~\cite{AdChGoHu01}.
    In particular, we would like the assembly time model proposed in~\cite{AdChGoHu01} to be derived as a special case of the model we propose, when only single-tile reactions with the seed-containing assembly are allowed.\footnote{As discussed in Section~\ref{sec-intro}, the model of~\cite{AdChGoHu01} is not \emph{exactly} a special case of our model, since we assume tile concentrations deplete.  However, the assumption of constant tile concentrations is itself a simplifying assumption of~\cite{AdChGoHu01} that is approximated by a more realistic model in which tile concentrations deplete, but seed tile types have very low concentration compared to other tile types, implying that non-seed concentrations do not deplete too much.  Under this more realistic assumption, if attachments not involving the seed are disallowed, then our definition of assembly time coincides with that of~\cite{AdChGoHu01}.}
    With a model such as Gillespie's algorithm~\cite{Gillespie77,GibsonBruck00,Gillespie07} using finite molecular counts, it is possible that no copy of the terminal assembly forms, so it is not clear how to sensibly ask how long it takes to form.\footnote{This problem is easily averted in a seeded system by setting the seed count sufficiently low to ensure that the terminal assembly is guaranteed to form at least one copy.  In a hierarchical system it is not clear how to avoid this problem.}
    The mass-action model of kinetics~\cite{epstein1998introduction} describes concentrations as a dynamical system that evolves continuously over time according to ordinary differential equations derived from reaction rates.
    This is an accepted model of kinetics when molecular counts are very large, which is already an implicit assumption in the standard aTAM.
    In the mass-action model, all possible terminal assemblies (assuming there are a finite number of different terminal assemblies) are guaranteed to form, which solves one issue with the purely stochastic model.
    But the solution goes too far: some (infinitesimal) concentration of all terminal assemblies form in the first infinitesimal amount of time, making the first appearance of a terminal assembly a useless measure of the time required to produce it.
    A sensible way to handle this may be to measure the time to half-completion (time required for the concentration of a terminal assembly to exceed half of its steady-state concentration).
    But this model is potentially subject to ``cheats'' such as systems that ``kill'' all but the fastest growing assemblies, so as to artificially inflate the average time to completion of those that successfully assemble into the terminal assembly.
    Furthermore, it would not necessarily be fair to directly compare such a deterministic model with the stochastic model of~\cite{AdChGoHu01}.

    The model of assembly time that we define is a continuous-time, discrete-state stochastic model similar to that of~\cite{AdChGoHu01}.
    However, rather than fixing transition rates at each time $t\in\R_{\geq 0}$ as constant, we use mass-action kinetics to describe the evolution over time of the concentration of producible assemblies, including individual tile types, which in turn determine transition rates.
    To measure the time to complete a terminal assembly, we use the same stochastic model as~\cite{AdChGoHu01}, which fixes attention on one particular tile\footnote{In the seeded model, the seed tile is the only tile to receive this attention. In our model, the tile to choose is a parameter of the definition.} and asks what is the expected time for it to grow into a terminal assembly, where the rate of attachment events that grow it are time-dependent, governed by the continuous mass-action evolution of concentration of assemblies that could attach to it.
    Unlike the seeded model, we allow the tile concentrations to deplete, since it is no longer realistic (or desirable for nontrivial hierarchical constructions) to assume that individual tiles do not react until they encounter an assembly containing the seed.\footnote{However, this depletion of individual tiles is \emph{not} the source of our time lower bound.  Suppose that we used a transition rate of 1 for each attachment of an individual tile (which is an upper bound on the attachment rate even for seeded systems due to the finite density constraint) and dynamic transition rates only for attachment of larger assemblies.  Then the assembly of hierarchical partial order systems \emph{still} would proceed asymptotically no faster than if single tile attachments were the only reactions allowed (as in the seeded assembly case), despite the fact that all the single-tile reactions at the intersection of the seeded and hierarchical model would be at least as fast in the modified hierarchical model as in the seeded model.}
}

\opt{normal}{\modelJustification}

\subsubsection{Formal definition of hierarchical time complexity}
We first formally define the dynamic evolution of concentrations by mass-action kinetics.
Let $\calT=(T,\tau)$ be a hierarchical TAS, and let $C:T\to[0,1]$ be a concentrations function.
Let $\R_{\geq 0} = [0,\infty)$, and let $t\in\R_{\geq 0}.$
For $\alpha\in\prodasm{\calT}$, let $[\alpha]_C(t)$ (abbreviated $[\alpha](t)$ when $C$ is clear from context) denote the concentration of $\alpha$ at time $t$ with respect to initial concentrations $C$, defined as follows.\opt{normal}{\footnote{More precisely, $[\alpha](t)$ denotes the concentration of the equivalence class of assemblies that are equivalent to $\alpha$ up to translation.  We have defined assemblies to have a fixed position only for mathematical convenience in some contexts, but for defining concentration, it makes no sense to allow the concentration of an assembly to be different from one of its translations.}}
We often omit explicit mention of $C$ and use the notation $[r](0)$ to mean $C(r)$, for $r\in T$, to emphasize that the concentration of $r$ is not constant with respect to time.
Given two assemblies $\alpha$ and $\beta$ that can attach to form $\gamma$, we model this event as a chemical reaction $R: \alpha + \beta \to \gamma$.
Say that a reaction $\alpha+\beta\to\gamma$ is \emph{symmetric} if $\alpha=\beta$.
Define the \emph{propensity} (a.k.a., \emph{reaction rate}) of $R$ at time $t\in\R_{\geq 0}$ to be $\rho_R(t) = [\alpha](t) \cdot [\beta](t)$ if $R$ is not symmetric, and $\rho_R(t) = \frac{1}{2} \cdot [\alpha](t)^2$ if $R$ is symmetric.\opt{normal}{\footnote{That is, all reaction rate constants are equal to 1.  To the extent that a rate constant models the ``reactivity'' of two molecules (the probability that a collision between them results in a reaction), it seems reasonable to model the rate constants as being equal.
To the extent that a rate constant also models diffusion rates (and therefore rate of collisions), this assumption may not apply; we discuss the issue in Section~\ref{sec-conclusion}.
Since we are concerned mainly with asymptotic results, if rate constants are assumed equal, it is no harm to normalize them to be 1.}}

If $\alpha$ is consumed in reactions $\alpha + \beta_1 \to \gamma_1, \ldots, \alpha + \beta_n \to \gamma_n$ and produced in asymmetric reactions $\beta'_1 + \gamma'_1 \to \alpha, \ldots, \beta'_m + \gamma'_m \to \alpha$ and symmetric reactions $\beta''_1 + \beta''_1 \to \alpha, \ldots, \beta''_p + \beta''_p \to \alpha$, then the concentration $[\alpha](t)$ of $\alpha$ at time $t$ is described by the differential equation
\begin{equation}\label{eq-conc-diff}
    \frac{d [\alpha](t)}{d t} =
    \sum_{i=1}^m [\beta'_i](t) \cdot [\gamma'_i](t)
    + \sum_{i=1}^p \frac{1}{2} \cdot  [\beta''_i](t)^2
    - \sum_{i=1}^n [\alpha](t) \cdot [\beta_i](t),
\end{equation}
with boundary conditions $[\alpha](0) = C(r)$ if $\alpha$ is an assembly consisting of a single tile $r$, and $[\alpha](0) = 0$ otherwise.
In other words, the propensities of the various reactions involving $\alpha$ determine its rate of change, negatively if $\alpha$ is consumed, and positively if $\alpha$ is produced.
\opt{submission}{
}

\newcommand{\propensityExplanation}{
    The definitions of the propensities of reactions deserve an explanation.
    Each propensity is proportional to the average number of collisions between copies of reactants per unit volume per unit time.
    For a symmetric reaction $\beta'' + \beta'' \to \alpha$, this collision rate is half that of the collision rate compared to the case where the second reactant is a distinct type of assembly, assuming it has the same concentration as the first reactant.\footnote{For intuition, consider finite counts: with $n$ copies of $\gamma$ and $n$ copies of $\beta \neq \gamma$, there are $n^2$ distinct pairs of molecules of respective type $\gamma$ and $\beta$, but with only $n$ copies of $\gamma$, there are $\frac{n(n-1)}{2}$ distinct pairs of molecules of type $\gamma$, which approaches $\frac{1}{2}n^2$ as $n \to \infty$.}
    Therefore the amount of $\alpha$ produced per unit volume per unit time is half that of a corresponding asymmetric reaction.
    The reason that terms of symmetric reactions $\alpha+\alpha\to\gamma$ that \emph{consume} $\alpha$ are not corrected by factor $\frac{1}{2}$ is that, although the number of such reactions per unit volume per unit time is half that of a corresponding asymmetric reaction, each such reaction consumes two copies of $\alpha$ instead of one.
    This constant 2 cancels out the factor $\frac{1}{2}$ that would be added to correct for the symmetry of the reaction.
    Therefore, the term $[\alpha](t) \cdot [\beta_i](t)$ representing the rate of consumption of $\alpha$ is the proper value whether or not $\alpha=\beta_i$.
}
\opt{normal}{\propensityExplanation}


This completes the definition of the dynamic evolution of concentrations of producible assemblies; it remains to define the time complexity of assembling a terminal assembly.
Although we have distinguished between seeded and hierarchical systems, for the purpose of defining a model of time complexity in hierarchical systems and comparing them to the seeded system time complexity model of~\cite{AdChGoHu01}, it is convenient to introduce a seed-like ``timekeeper tile'' into the hierarchical system, in order to stochastically analyze the growth of this tile when it reacts in a solution that is itself evolving according to the continuous model described above.
The seed does not have the purpose of nucleating growth, but is introduced merely to focus attention on a single molecule that has not yet assembled anything, in order to ask how long it will take to assemble into a terminal assembly.\footnote{For our lower bound result, Theorem~\ref{thm-hier-partial-order-slow}, it will not matter which tile type is selected as the timekeeper, except in the following sense.
We define partial order systems, the class of directed hierarchical TAS's to which the bound applies, also with respect to a particular tile type in the unique terminal assembly.
A TAS may be a partial order system with respect to one tile type but not another, but for all tile types $s$ for which the TAS is a partial order system, the time lower bound of Theorem~\ref{thm-hier-partial-order-slow} applies when $s$ is selected as the timekeeper.
The upper bound of Theorem~\ref{thm-fast-rectangle} holds with respect to only a single tile type.}
The choice of which tile type to pick will be a parameter of the definition, so that a system may have different assembly times depending on the choice of timekeeper tile.

Fix a copy of a tile type $s$ to designate as a ``timekeeper seed''.
The assembly of $s$ into some terminal assembly $\ha$ is described as a time-dependent continuous-time Markov process in which each state represents a producible assembly containing $s$, and the initial state is the size-1 assembly with only $s$.
For each state $\alpha$ representing a producible assembly with $s$ at the origin, and for each pair of producible assemblies $\beta,\gamma$ such that $\alpha + \beta \to \gamma$ (with the translation assumed to happen only to $\beta$ so that $\alpha$ stays ``fixed'' in position), there is a transition in the Markov process from state $\alpha$ to state $\gamma$ with transition rate $[\beta](t)$.\opt{normal}{\footnote{That is, for the purpose of determining the continuous dynamic evolution of the concentration of assemblies, including $\alpha$,  in solution at time $t$, the rate of the reaction $\alpha + \beta \to \gamma$ at time $t$ is assumed to be proportional to $[\alpha](t)[\beta](t)$ (or half this value if the reaction is symmetric).
However, for the purpose of determining the stochastic dynamic evolution of one particular copy of $s$, the rate of this reaction at time $t$ is assumed to be proportional only to $[\beta](t)$.
This is because we want to describe the rate at which \emph{this particular copy} of $\alpha$, the one containing the copy of $s$ that we fixed at time 0, encounters assemblies of type $\beta$.
This instantaneous rate is independent of the number of \emph{other} copies of $\alpha$ at time $t$ (although after $\epsilon$ seconds the rate will change to $[\beta](t+\epsilon)$, which of course \emph{will} depend on $[\alpha]$ over that time interval).}}
Unlike the seeded model, the transition rates vary over time since the assemblies (including assemblies that are individual tiles) with which $\alpha$ could interact are themselves being produced and consumed.

We define $\time_{\calT,C,s}$ to be the random variable representing the time taken for the copy of $s$ to assemble into a terminal assembly via some sequence of reactions as defined above.
We define the time complexity of a directed hierarchical TAS $\calT$ with concentrations $C$ and timekeeper $s$ to be $\mathsf{T}(\calT,C,s) = \exp{\time_{\calT,C,s}}$, and the time complexity with respect to $s$ as $\mathsf{T}(\calT,s) = \min_C \mathsf{T}(\calT,C,s)$.\footnote{It is worth noting that this expected value could be infinite.
This would happen if some partial assembly $\alpha$, in order to complete into a terminal assembly, requires the attachment of some assembly $\beta$ whose concentration is depleting quickly.}


We note in particular that our construction of Theorem~\ref{thm-hierarchical-square} is composed of $(\frac{n}{\log n})^2$ different types of $O(\log n) \times O(\log n)$ ``blocks'' that can each grow via only one reaction.
At least one of these blocks $\beta$ must obey $[\beta](t) \leq \frac{\log^2 n}{n^2}$ for all $t\in\R_{\geq 0}$ (this can be seen by applying Lemma~\ref{lem-conserve-mass}, proven in Section~\ref{subsec-lower-bound}).
This implies that the rate of the slowest such reaction has rate at most $\frac{\log^2 n}{n^2}$, hence expected time at least the inverse of that quantity.
Thus our square construction assembles in at least $\Omega(\frac{n^2}{\log^2 n})$ time, slower than the optimal seeded time of $O(n)$ \cite{AdChGoHu01}.
Proving this formally requires more details that we omit.
However, it is simple to modify the system to have a tile type appearing in exactly one position in the terminal assembly, for example, by attaching such a tile type only to the block at coordinate $(0,0)$.
It is routine to check that this would make the system a partial order system with respect to that tile type as defined in Section~\ref{sec-defn-partial-order-systems}.
Then Theorem~\ref{thm-hier-partial-order-slow} implies a time complexity lower bound of $\Omega(n)$, much slower than the polylogarithmic time one might na\"{i}vely expect due to the polylogarithmic depth of the assembly tree.

\section{Time complexity lower bound for hierarchical partial order systems}
\label{sec-time-complexity-lower-bound}

In this section we show that the subset of hierarchical TAS's known as \emph{partial order systems} cannot assemble any shape of diameter $D$ in faster than time $\Omega(D)$.

\subsection{Definition of hierarchical partial order systems}
\label{sec-defn-partial-order-systems}

Seeded partial order systems were first defined by Adleman, Cheng, Goel, and Huang~\cite{AdChGoHu01} for the purpose of analyzing the running time of their optimal square construction.
Intuitively, a seeded directed TAS with unique terminal assembly $\ha$ is a partial order system if every pair of adjacent positions $p_1$ and $p_2$ in $\ha$ that interact with positive strength have the property that either $p_1$ always receives a tile before $p_2$, or vice versa.
We extend the definition of partial order systems to hierarchical systems in the following way.

Let $\calT=(T,\tau)$ be a hierarchical directed TAS with unique terminal assembly $\ha\in\termasm{\calT}$.
A \emph{terminal assembly tree} of $\ha$ is a full binary tree with $|\ha|$ leaves, in which each leaf is labeled with an assembly consisting of a single tile, the root is labeled with $\ha$, and each internal node is labeled with an assembly producible in one step from the $\tau$-stable attachment of its two child assemblies.\footnote{Note that even a directed hierarchical TAS may have more than one assembly tree of the terminal assembly $\ha$.}
Let $\Upsilon$ be any terminal assembly tree of $\calT$.
Let $p\in\dom\ha$ and let $s=\ha(p)$.
The \emph{assembly sequence with respect to $\Upsilon$ starting at $p$} is the sequence of assemblies $\vec{\alpha}_{p,\Upsilon} = (\alpha_1,\ldots,\alpha_k)$ that represent the path from the leaf corresponding to $p$ to the root of $\Upsilon$, so that $\alpha_1$ is the single tile $s$ at position $p$, and $\alpha_k = \ha$.\footnote{That is, $\vec{\alpha}$ is like a seeded assembly sequence in that each $\alpha_i$ is a \emph{subassembly} of $\alpha_{i+1}$ (written $\alpha_i \sqsubseteq \alpha_{i+1}$, meaning $\dom\alpha_i \subseteq \dom\alpha_{i+1}$ and $\alpha_{i}(p) = \alpha_{i+1}(p)$ for all $p \in \dom\alpha_{i}$).
The difference is that $\alpha_i$ and $\alpha_{i+1}$ may differ in size by more than one tile, since $\dom \alpha_{i+1} \setminus \dom\alpha_i$ will consist of all points in the domain of $\alpha_i$'s sibling in $\Upsilon$.}
An \emph{assembly sequence starting at $p$} is an assembly sequence with respect to $\Upsilon$ starting at $p$, for some valid assembly tree $\Upsilon$.

An \emph{attachment quasiorder}  with respect to $p\in\dom\ha$ is a quasiorder (a reflexive, transitive relation) $\preceq$ on $\dom \ha$ such that the following holds:
\begin{enumerate}

  \item For every $p_1,p_2\in\dom\ha$, $p_1 \preceq p_2$ if and only if for every assembly sequence $\vec{\alpha}=(\alpha_1,\ldots,\alpha_k)$ starting at $p$, for all $1 \leq i \leq k$, $\alpha_i(p_2)$ is defined $\implies$ $\alpha_i(p_1)$ is defined.
  In other words, $p_1$ must always have a tile by the time $p_2$ has a tile. (Perhaps they arrive at the same time, if they are both part of some assembly that attaches in a single step.)

  \item For every pair of adjacent positions $p_1,p_2 \in \dom\ha$, if the tiles at positions $p_1$ and $p_2$ interact with positive strength in $\ha$, then $p_1 \preceq p_2$ or $p_2 \preceq p_1$ (or both).
\end{enumerate}


If two tiles always arrive at the same time to the assembly containing $p$, then they will be in the same equivalence class induced by $\preceq$.
Given an attachment quasiorder $\preceq$, we define the \emph{attachment partial order $\prec$ induced by $\preceq$} to be the strict partial order on the quotient set of equivalence classes induced by $\preceq$.
In other words, if some subassembly $\alpha \sqsubseteq \ha$ always attaches to the assembly containing $p$ all at once, then all positions $p' \in \dom\alpha$ will be equivalent under $\preceq$.\footnote{More generally, if there is a subset $X\subset \dom\ha$ such that all assemblies $\alpha$ attaching to the assembly containing $p$ have the property that $X \cap \dom\alpha \neq \emptyset \implies X \subseteq \dom\alpha$; i.e., any position in $X$ attaching implies all positions in $X$ attach with it, then this implies all positions in $X$ are equivalent under $\preceq$.}
It is these equivalence classes of positions that are related under $\prec$.
Each attachment partial order $\prec$ induces a directed acyclic graph $G=(V,E)$, where $V = \{\beta_1,\ldots,\beta_k\}$, each $\beta_i$ represents the subassembly corresponding to some equivalence class (under $\preceq$) of positions in $\dom\ha$, and $(\beta_i,\beta_j) \in E$ if $\dom\beta_i \prec \dom\beta_j$.\footnote{i.e., if the positions are nodes on a directed graph with an edge from $p_1$ to $p_2$ if $p_1 \preceq p_2$, then each equivalence class is a strongly connected component of the graph, and $\prec$ describes the \emph{condensation} directed acyclic graph obtained by contracting each strongly connected component into a single vertex.}
Note that the first assembly $\alpha_1$ of the assembly sequence containing $p$ is always size 1, since by definition $p$ is the only position with a tile at time 0.

We say that a directed hierarchical TAS $\calT$ with unique terminal assembly $\ha$ is a \emph{hierarchical partial order system with respect to $p$} if it has an attachment quasiorder with respect to $p$.
Given a tile type $s\in T$ that appears exactly once in the terminal assembly $\ha$ at position $p$ (i.e., $\ha(p)=s$ and $(\forall q \in \dom \ha \setminus \{p\})\ \ha(q) \neq s$), we say that $\calT$ is a \emph{hierarchical partial order system with respect to $s$} if $\calT$ is a hierarchical partial order system with respect to $p$. 

\begin{remark}
The condition that $s$ appears exactly once in $\ha$ allows us to talk interchangeably of a partial order with respect to a \emph{position} and a partial order with respect to a \emph{tile type}.
If instead we allowed $s$ to appear in multiple positions in $\ha$, then the sequence of attachments would nondeterministically choose one of them.
Since the partial order imposed on $\dom \ha$ is different depending on the position chosen to define the partial order, we would lose the ability to focus on ``the'' partial order imposed by $s$ (since there would be more than one possible).
It is conceivable that a tile type $s$ could appear in multiple positions in $\ha$, and $\calT$ could be a partial order with respect to all of these positions.
It is an open question whether Theorem~\ref{thm-hier-partial-order-slow} would apply to such a system.
Our proof technique relies fundamentally on $s$ appearing in a single fixed position in $\ha$ and defining a single partial order on $\dom \ha$ that is then used to establish the assembly time lower bound.
The definition of seeded partial order systems~\cite{AdChGoHu01,ACGHKMR02} has a similar requirement, that the seed tile type appear exactly once in the terminal assembly.
\end{remark}

In the case of seeded assembly, in which each attachment is of a ``subassembly'' containing a single tile to a subassembly containing the seed, this definition of partial order system is equivalent to the definition of partial order system given in~\cite{AdChGoHu01}.
In the seeded case, since no tiles may attach simultaneously, the attachment quasiorder we have defined induces equivalence classes that are singletons, and the resulting induced strict partial order is the same as the strict partial order used in~\cite{AdChGoHu01}.

\subsection{Repetitious assemblies}

This section shows that hierarchical partial order systems are well-behaved in a certain technical sense that will be useful in the proof of Theorem~\ref{thm-hier-partial-order-slow}.

\begin{definition}
  Two overlapping assemblies $\alpha$ and $\beta$ are \emph{consistent} if $\alpha(p)=\beta(p)$ for every $p\in\dom\alpha\cap\dom\beta$.
  If $\alpha$ and $\beta$ are consistent, define their \emph{union} $\alpha \cup \beta$ to be the assembly with $\dom (\alpha \cup \beta) = \dom \alpha \cup \dom \beta$ defined by $(\alpha \cup \beta)(p) = \alpha(p)$ if $p \in \dom \alpha$ and $(\alpha \cup \beta)(p) = \beta(p)$ if $p \in \dom \beta$.
  Let $\alpha \cup \beta$ be undefined if $\alpha$ and $\beta$ are not consistent.
\end{definition}

\begin{definition}
  Let $\alpha$ be a producible assembly, let $\vv \in \Z^2$ be a vector, and let $\alpha+\vv$ denote the translation of $\alpha$ by $\vv$, i.e., an assembly $\beta$ such that $\dom\beta=\dom\alpha+\vv$ and $\beta(p)=\alpha(p-\vv)$ for all $p\in\dom\beta$.
  We say that assembly $\alpha$ is \emph{repetitious} if there exists a nonzero
  vector $\vv \in \Z^{2}$ such that $\dom \alpha \cap \dom (\alpha  + \vec
  v)\ne \emptyset$ and $\alpha$ and $\alpha  + \vec v$ are consistent.
\end{definition}

The following theorem shows an important property of hierarchical systems that will be useful.

\begin{theorem}[\cite{poirghtsJoCG}]\label{thm-repetitious}
  Let $\calT$ be a hierarchical tile assembly system.
  If $\calT$ has a producible repetitious assembly, then arbitrarily large assemblies are producible in $\calT$.
\end{theorem}

Let $\alpha,\beta$ be producible assemblies that can attach.
Imagine fixing the position of $\alpha$ so that $\beta$ must be translated to attach to $\alpha$.
Let $V_{\alpha,\beta}$ be the set of all vectors $\vv$ such that $\dom \alpha \cap (\dom \beta + \vv) = \emptyset$ and $\alpha \cup (\beta + \vv)$ is a stable assembly; i.e., $V_{\alpha,\beta}$ describes the set of all ways to attach $\beta$ to $\alpha$.
We say that $\alpha$ and $\beta$ are \emph{disjointly attachable} if, for every $\vv_1,\vv_2 \in V_{\alpha,\beta}$ such that $\vv_1 \neq \vv_2$, it holds that $\dom(\beta+\vv_1) \cap \dom(\beta+\vv_2) = \emptyset$; in other words, no two translations of $\beta$ that allow it to attach to $\alpha$ overlap each other.

\begin{corollary}\label{cor-disjointly-attachable}
  Let $\calT=(T,\tau)$ be a hierarchical partial order system with respect to tile type $s \in T$ appearing at position $p \in \dom \ha$ in the terminal assembly $\ha$, and let $(\alpha_1,\ldots,\alpha_k)$ be an assembly sequence starting at $p$.
  Then for every $i \in \{1,\ldots,k\}$, every producible assembly $\beta$ that can attach to $\alpha_i$ is disjointly attachable to $\alpha_i$.
\end{corollary}

\begin{proof}
  Because $s$ appears only at position $p \in \dom\ha$ and $\ha$ is the unique terminal assembly, every position relative to the position of $s$ has a fixed tile type appearing there.

  Suppose for the sake of contradiction that there is some producible assembly $\beta$ that is attachable to $\alpha_i$, but not disjointly, so that there are two vectors $\vv_1 \neq \vv_2$ such that, defining $\beta_1 = \beta + \vv_1$ and $\beta_2 = \beta+\vv_2$, we have that $\alpha \cup \beta_1$ and $\alpha \cup \beta_2$ are both stable assemblies (and $\dom\alpha \cap \dom\beta_1 = \emptyset$ and $\dom\alpha \cap \dom\beta_2 = \emptyset$), but $\dom \beta_1 \cap \dom \beta_2 \neq \emptyset$.
  Therefore, for every position in the overlap $q \in \dom \beta_1 \cap \dom \beta_2$, $\beta_1(q) = \beta_2(q)$, otherwise this would imply one terminal assembly producible from $\alpha_i$ that has one tile type at position $q$ and another that has a different tile type at position $q$.
  But this is precisely what it means for $\beta$ to be a repetitious assembly.
  Theorem~\ref{thm-repetitious} then implies that arbitrarily large assemblies are producible in $\calT$, contradicting the fact that it produces a unique terminal assembly.
\end{proof}

\subsection{Linear time lower bound for partial order systems}
\label{subsec-lower-bound}

Theorem~\ref{thm-hier-partial-order-slow} establishes that hierarchical partial order systems, like their seeded counterparts, cannot assemble a shape of diameter $D$ in less than $\Omega(D)$ time.
This is potentially counterintuitive, since an attaching assembly of size $K$ is able to increase the size of the growing assembly by $K$ tiles in a single attachment step.
Intuitively, the time lower bound is proven by using the fact that such an assembly can have concentration at most $\frac{1}{K}$ by conservation of mass, slowing down its rate of attachment (compared to the rate of a single tile) by factor at least $K$, precisely enough to cancel out the potential speedup over a single tile due to its size.

This simplistic argument is not quite accurate and must be amortized --- using our Conservation of Mass Lemma (Lemma~\ref{lem-conserve-mass}) --- over all assemblies that could extend the growing assembly.
The growing assembly may be extended at a single attachment site by more than one assembly.
However, by Lemma~\ref{lem-conserve-mass}, these assemblies must \emph{collectively} have limited total concentration.
Intuitively, the property of having a partial order on binding subassemblies ensures that the assembly of each path in the partial order graph proceeds by a series of rate-limiting steps.
We prove upper bounds on each of these rates using this concentration argument.\opt{normal}{\footnote{The same assembly $\alpha$ could attach to many locations $p_1,\ldots,p_n$.
In a TAS that is not a partial order system, it could be the case that there is not a fixed attachment location that is necessarily required to complete the assembly.
In this case completion of the assembly might be possible even if only one of $p_1,\ldots,p_n$ receives the attachment of $\alpha$.
Since the minimum of $n$ exponential random variables with rate $1/K$ is itself exponential with rate $n/K$, the very first attachment of $\alpha$ to any of $p_1,\ldots,p_n$ happens in expected time $K/n$, as opposed to expected time $K$ for $\alpha$ to attach to a \emph{particular} $p_i$.
This prevents our technique from applying to such systems, and it is the fundamental speedup technique in our proof of Theorem~\ref{thm-fast-rectangle}.}}
Since the rate-limiting steps must occur in order, we can then use linearity of expectation to bound the total expected time.

The following is a ``conservation of mass lemma'' that will be helpful in the proof of Theorem~\ref{thm-hier-partial-order-slow}.
Note that it applies to any hierarchical system.

\begin{lemma}[{\bf Conservation of Mass Lemma}]\label{lem-conserve-mass}
  Let $\calT=(T,\tau)$ be a hierarchial TAS and let $C:T\to[0,1]$ be a concentrations function. Then for all $t \in \R_{\geq 0}$,
  \[
  \sum_{\alpha \in \prodasm{\calT}} [\alpha](t) \cdot |\alpha| = \sum_{r \in T} C(r) \ \ \ \ (\leq 1).
  \]
\end{lemma}

\begin{proof}
  For all $t\in\R_{\geq 0}$, define $f(t) = \sum_{\alpha \in \prodasm{\calT}} [\alpha](t) \cdot |\alpha|.$
  According to our model, $[\alpha](0) = C(r)$ if $\alpha$ consists of a single tile type $r$ and $[\alpha](0)=0$ otherwise, so $f(0) = \sum_{r \in T} C(r)$.
  Therefore it is sufficient (and necessary) to show that $\frac{df}{dt} = 0$.
  For all $\alpha \in \prodasm{\calT}$ and $t\in\R_{\geq 0}$, define $f_\alpha(t) = [\alpha](t) \cdot |\alpha|$.
  Then by equation~\eqref{eq-conc-diff}, and recalling from that equation the definitions of $m$, $n$, $p$, $\beta'_i$, $\beta''_i$, $\gamma'_i$, and $\beta_i$, annotated as $m(\alpha), n(\alpha)$, etc. to show their dependence on $\alpha$, we have
  \[
    \frac{d f_\alpha}{dt} = |\alpha| \cdot \left(
    \sum_{i=1}^{m(\alpha)} [\beta'_i(\alpha)](t) \cdot [\gamma'_i(\alpha)](t) +
    \sum_{i=1}^{p(\alpha)} \frac{1}{2} \cdot [\beta''_i(\alpha)](t)^2 -
    \sum_{i=1}^{n(\alpha)} [\alpha](t) \cdot [\beta_i(\alpha)](t)
    \right).
  \]
  Then
  \begin{eqnarray*}
    \frac{df}{dt}
    &=&
    \frac{d}{dt} \sum_{\alpha \in \prodasm{\calT}} f_\alpha(t)
    =
    \sum_{\alpha \in \prodasm{\calT}} \frac{d f_\alpha}{dt}
    \\&=&
    \sum_{\alpha \in \prodasm{\calT}} \left(\begin{array}{l}
    \displaystyle\sum_{i=1}^{m(\alpha)} |\alpha| \cdot [\beta'_i(\alpha)](t) \cdot [\gamma'_i(\alpha)](t) +
    \displaystyle\sum_{i=1}^{p(\alpha)} |\alpha| \cdot \frac{1}{2} \cdot [\beta''_i(\alpha)](t)^2
    \\
    - \displaystyle\sum_{i=1}^{n(\alpha)} |\alpha| \cdot [\alpha](t) \cdot [\beta_i(\alpha)](t)
    \end{array}\right).
  \end{eqnarray*}
  Let $\mathcal{R}$ denote the set of all attachment reactions of $\calT$, writing $R(\alpha,\beta,\gamma)$ to denote the reaction $\alpha + \beta \to \gamma$.
  For each such reaction, $|\alpha| + |\beta| = |\gamma|$.
  In particular, if $\alpha=\beta$, then $|\gamma|=2|\alpha|$.
  Each such asymmetric reaction contributes precisely three unique terms in the right hand side above: two negative (of the form $- |\alpha| \cdot [\alpha](t) \cdot [\beta](t)$ and $-|\beta| \cdot [\alpha](t) \cdot [\beta](t)$) and one positive (of the form $|\gamma| \cdot [\alpha](t) \cdot [\beta](t)$).
  Each such symmetric reaction contributes two unique terms: one negative (of the form $- |\alpha| \cdot [\alpha](t)^2$) and one positive (of the form $|\gamma| \cdot \frac{1}{2} \cdot [\alpha](t)^2$).

  \newcommand{\asymrxn}{{\stack{ R(\alpha,\beta,\gamma) \in \mathcal{R} }{ \alpha \neq \beta }}}
  \newcommand{\symrxn}{{R(\alpha,\alpha,\gamma) \in \mathcal{R}}}

  Then we may rewrite the above sum as
  \begin{align*}
    \frac{df}{dt}
    &=
    \sum_\asymrxn
    \left(|\gamma| \cdot [\alpha](t) \cdot [\beta](t) - |\alpha| \cdot [\alpha](t) \cdot [\beta](t) - |\beta| \cdot [\alpha](t) \cdot [\beta](t) \right)
    \\
    & \ \ \ \ + \sum_\symrxn
    \left(\frac{1}{2} |\gamma| \cdot [\alpha](t)^2 - |\alpha| \cdot [\alpha](t)^2 \right)
    \\&=
    \sum_\asymrxn (|\gamma| - |\alpha| - |\beta|) \cdot [\alpha](t) \cdot [\beta](t)
    +
    \sum_\symrxn
    \left( \frac{1}{2} |\gamma| - |\alpha| \right) \cdot [\alpha](t)^2
    \\&=
    \sum_\asymrxn 0 \cdot [\alpha](t) \cdot [\beta](t)
    +
    \sum_\symrxn
    0 \cdot [\alpha](t)^2
    =
    0. \qedhere
  \end{align*}
\end{proof}

The following is the main theorem of this paper, and it shows that hierarchical partial order systems require time linear in the diameter of a shape in order to produce it.

\begin{theorem} \label{thm-hier-partial-order-slow}
  Let $\calT=(T,\tau)$ be a hierarchial partial order system with respect to $s \in T$, with unique terminal assembly $\ha$ of $L_1$ diameter $D$.
  Then $\mathsf{T}(\calT,s) = \Omega(D)$.
\end{theorem}

\newcommand{\ProofHierPartialOrderSlow}{
  Let $C:T\to[0,1]$ be a concentrations function.
  Let $\ha\in\termasm{\calT}$ be the unique terminal assembly of $\calT$, and let $p\in\dom\ha$ be the unique position such that $\ha(p)=s$.
  Let $\preceq$ be the attachment quasiorder testifying to the fact that $\calT$ is a partial order system with respect to $p$.
  Let $\prec$ be the strict partial order induced by $\preceq$.
  Let $G=(V,E)$ be the directed acyclic graph induced by $\prec$.
  Assign weights to the edges of $E$ by $w(\alpha_i,\alpha_j) = |\alpha_j|$.

  If $P' = (\alpha'_1,\ldots,\alpha'_l)$ is any path in $G$, define the \emph{(weighted) length} of $P$ to be $w(P') = \sum_{i=1}^{l-1} w(\alpha'_i,\alpha'_{i+1})$.
  Let $q,r\in\dom\ha$ be two points at $L_1$ distance $D$, which must exist since the diameter of $\dom\ha$ is $D$.
  The distance $d$ from $p$ to one of these points --- without loss of generality, call it $q$ --- in $\dom\ha$ is at least $D/2$ by the triangle inequality.
  Let $\alpha_n \in V$ be such that $q\in\dom\alpha_n$, and let $P=(\alpha_1,\ldots,\alpha_n)$ be any path in $G$ starting with $\alpha_1$ and ending with $\alpha_n$, where $\alpha_1$ is the assembly consisting just of the tile at position $p$.
  The weight of each edge in $P$ is an upper bound on the diameter of the assembly it represents, and by the triangle inequality, the sum of these diameters for all $\alpha_i$ in $P$ is itself an upper bound on the distance from $p$ to $q$.
  Therefore $w(P) \geq D/2$.

  Let $\mathbf{T}_P$ be the random variable representing the time taken for the complete path $P$ to assemble.
  Since the tiles on $P$ represent a subassembly of $\ha$, $\ha$ cannot completely form until the path $P$ forms.
  Therefore $\mathbf{T}_P \leq \time_{\calT,C,p}$.
  Since $\mathsf{T}(\calT,C,p) = \exp{\time_{\calT,C,p}}$, it suffices to show that $\exp{\mathbf{T}_P} = \Omega(w(P))$.

  Because of the precedence relationship described by $\prec$, no portion of the path $P$ can form until its immediate predecessor on $P$ is present.
  After some amount of time, some prefix $P'$ of the path $P$ has assembled (possibly with some other portions of $\ha$ not on the path $P$).
  Given $t\in\R_{\geq 0}$, let $\mathbf{L}(t)$ be the random variable indicating the weighted length of this prefix after $t$ units of time.

  We claim that for all $t\in\R_{\geq 0}$, $\exp{\mathbf{L}(t)} \leq t$; this claim is proven below.
  Assuming the claim, by Markov's inequality, $\prob[\mathbf{L}(t) \geq 2t] \leq \frac{1}{2}$.
  Letting $t=w(P)/2$, the event $\mathbf{L}(w(P)/2) \geq w(P)$ is equivalent to the event $\mathbf{T}_{P} \leq w(P)/2.$
  Thus $\prob[\mathbf{T}_{P} \leq w(P)/2] \leq \frac{1}{2}.$
  By Markov's inequality, $\exp{\mathbf{T}_{P}} \geq w(P) / 4 \geq D/8 = \Omega(D),$ which proves the theorem, assuming that $\exp{\mathbf{L}(t)} \leq t$.

  The remainder of the proof shows the claim that for all $t\in\R_{\geq 0}$, $\exp{\mathbf{L}(t)} \leq t$.
  Define the function $f:\R_{\geq 0} \to \R_{\geq 0}$ for all $t\in\R_{\geq 0}$ by $f(t) = \exp{\mathbf{L}(t)}$, noting that $f(0)=0$.
  Let $f' = \frac{d f}{d t}$.
  Let $P' = (\alpha_1, \ldots, \alpha_{m})$ be the prefix of $P$ formed after $t$ seconds.  Let $\beta_1, \beta_2, \ldots, \beta_{k}$, with $m+k=n$, be the individual subassemblies remaining on the path, in order, so that $P=(\alpha_1, \ldots, \alpha_m,\beta_1,\ldots,\beta_k)$.
  For all $1 \leq i \leq k$, let $\gamma_i = \bigcup_{j=1}^i \beta_j$ be the union of the next $i$ such subassemblies on the path (representing each of the amounts by which $P$ could grow in the next attachment event).
  Let $s_i = |\gamma_i|$ be the size of the $i^\text{th}$ subassembly, and let $c_i(t) = \sum_{\alpha \in A_i(t)} [\alpha](t)$, where $A_i(t)$ is the set of subassemblies (possibly containing tiles not on the path $P$) at time $t$ that contain $\gamma_i$ but do not contain $\gamma_{i+1}$.
  $A_i(t)$ represents the set of all assemblies that could attach to grow $P$ by exactly the tiles in $\gamma_i$.

  However, the set of \emph{reactions} that could grow $P$ is what matters, but Corollary~\ref{cor-disjointly-attachable} implies that no assembly extending $P'$ could attach via two different reactions that both intersect $P$ at the location directly succeeding $P'$.
  Thus, summing over assemblies is equivalent to summing over reactions.
  Our argument uses the conservation of mass property (Lemma~\ref{lem-conserve-mass}) to show that no matter the concentration of assemblies in each $A_i(t)$, the rate of growth is at most one tile per unit of time.

  For each $1 \leq i \leq k$, in the next instant $dt$, with probability $c_i dt$ the prefix will extend by total weighted length $s_i$ by attachment of (a superassembly containing) $\gamma_i$.
  This implies that $f'(t) \leq \sum_{i=1}^k c_i(t) \cdot s_i$.
  Invoking Lemma~\ref{lem-conserve-mass}, it follows that for all $t\in\R_{\geq 0}$,
  \[
    f'(t)
    \leq
    \sum_{i=1}^k c_i(t) \cdot s_i
    \leq
    \sum_{\alpha \in \prodasm{\calT}} [\alpha](t) \cdot |\alpha|
    =
    \sum_{r \in T} C(r)
    \leq 1.
  \]
  Since $f(0)=0$, this implies that $f(t) \leq t$ for all $t\in\R_{\geq 0}$, which completes the proof of the claim that $\exp{\mathbf{L}(t)} \leq t$.
}

\opt{normal}{
    \begin{proof}
    \ProofHierPartialOrderSlow
    \end{proof}

    Although our assembly time model describes concentrations as evolving according to the standard mass-action kinetic differential equations, Lemma~\ref{lem-conserve-mass} is the only property of this model that is required for our proof of Theorem~\ref{thm-hier-partial-order-slow}.
    Even if concentrations of attachable assemblies (and thus their associated attachment rates in the Markov process defining assembly time) could be magically adjusted throughout the assembly process so as to minimize the assembly time, so long as the concentrations obey Lemma~\ref{lem-conserve-mass} at all times, Theorem~\ref{thm-hier-partial-order-slow} still holds.

    For example, staged assembly~\cite{DDFIRSS07} is a relaxation of the mass-action model that obeys Lemma~\ref{lem-conserve-mass}, in which certain assemblies are artificially prevented from interacting by being kept in separate bins before being mixed.
    Theorem~\ref{thm-hier-partial-order-slow} implies that staged assembly gives no time speedup on partial order systems if the completion time in each bin is taken into account in measuring the time complexity.
}

\section{Assembly of a shape in time sublinear in its diameter}
\label{sec-fast-rectangle}

\opt{submission}{This section states and outlines the ideas of the proof of the following theorem, which shows that relaxing the partial order assumption on hierarchical tile systems allows for assembly time sublinear in the diameter of the shape.
}

\opt{normal}{This section is devoted to proving the following theorem, which shows that relaxing the partial order assumption on hierarchical tile systems allows for assembly time sublinear in the diameter of the shape.}

\begin{theorem}\label{thm-fast-rectangle}
  For infinitely many $n \in \N$, there is a (non-directed) hierarchical TAS $\calT=(T,2)$ that strictly self-assembles an $n \times n'$ rectangle, where $n' = o(n)$, such that $|T|=O(\log n)$ and there is a tile type $s \in T$ such that $\mathsf{T}(\calT,s) = O(n^{4/5} \log n)$.
\end{theorem}

\opt{normal}{
In our proof, $n' \approx n^{3/5}$, but we care only that $n' \leq n$ so that the diameter of the shape is $\Theta(n)$.
As discussed in Section~\ref{sec-runtime}, we interpret the upper bound of Theorem~\ref{thm-fast-rectangle} more cautiously than the lower bound of Theorem~\ref{thm-hier-partial-order-slow}, since some of our simplifying assumptions concerning diffusion rates and binding strength thresholds, discussed in Section~\ref{sec-conclusion}, may cause the assembly time to appear artificially faster in our model than in reality.
A reasonable conclusion to be drawn from Theorem~\ref{thm-fast-rectangle} is that concentration arguments alone do not suffice to show a linear-time lower bound on assembly time in the hierarchical model.}

\opt{submission}{
    \begin{figure}[htb]%
    \centering
    \parbox{3.1in}{%
    \includegraphics[width=3.1in]{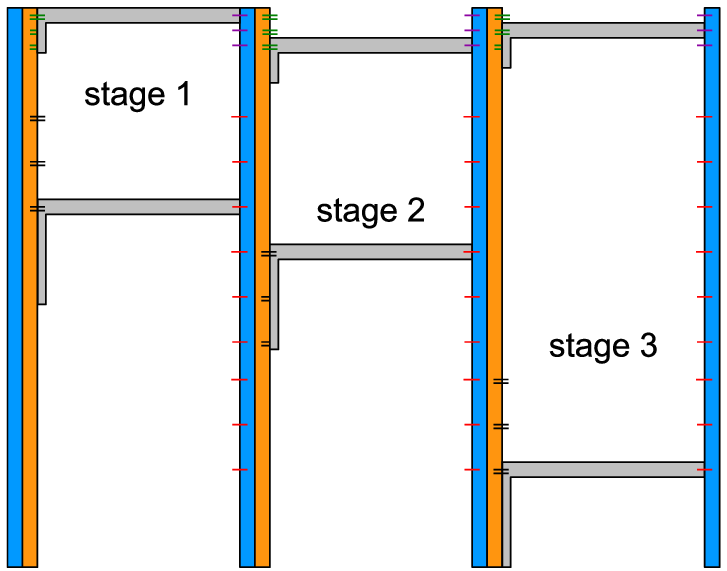}
    \caption{\figuresize
      High-level overview of interaction of ``vertical bars'' and ``horizontal bars'' to create the rectangle of Theorem~\ref{thm-fast-rectangle}.
      Filler tiles fill in the empty regions.
      If glues overlap two regions, then they represent a formed bond.
      If glues overlap one region but not another, then they are glues from the former region but are mismatched (and thus ``covered'') by the latter region.
      This is intended to ensure that only one horizontal bar can attach in each of the top and bottom regions of a vertical bar.}%
    \label{fig:vertical-bars-overview-submission}}%
    \qquad
    \begin{minipage}{3in}%
    \includegraphics[width=3in]{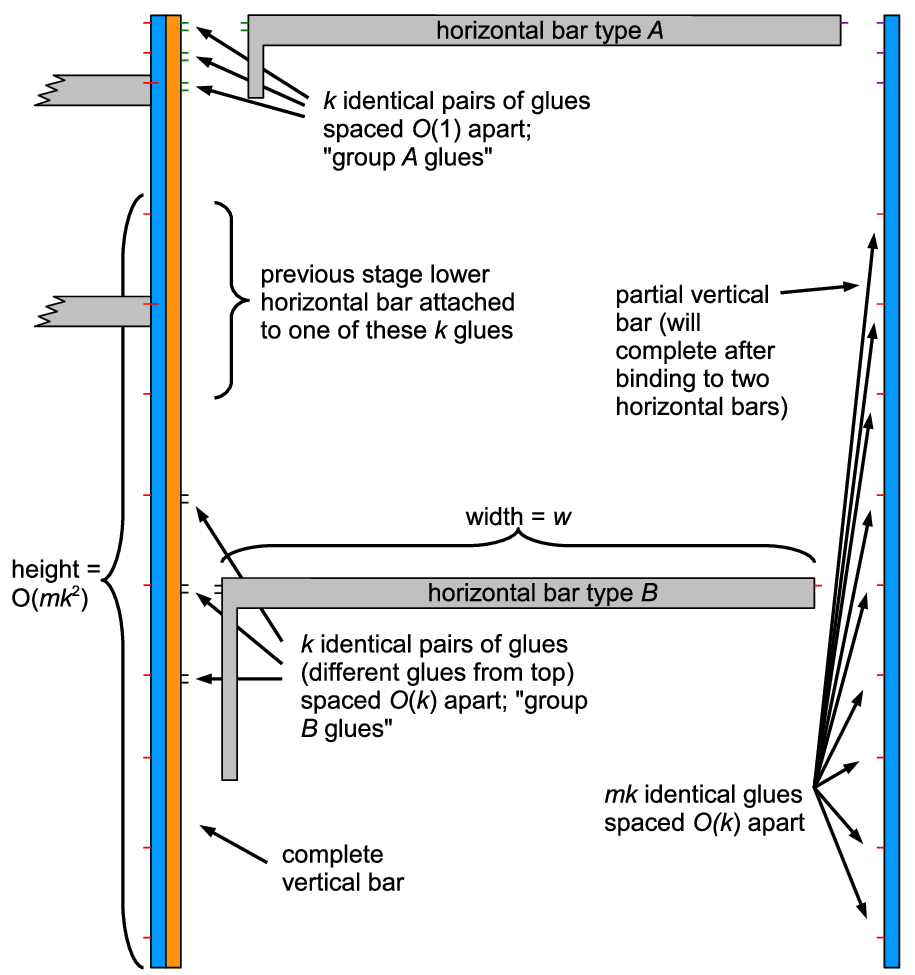}
    \caption{\figuresize
      Vertical bars and horizontal bars for the construction of a fast-assembling square.
      ``Type $B$'' horizontal bars have a longer vertical arm than ``Type $A$'' since the glues they must block are farther apart.}%
    \label{fig:vertical-bars-submission}%
    \end{minipage}%
    \end{figure}%

    A high-level overview of the construction is shown in Figure~\ref{fig:vertical-bars-overview-submission}.
    The interaction of the components of this figure are shown in more detail in Figure~\ref{fig:vertical-bars-submission}.
    Let $k=m=n^{1/5}$, and let $w=n^{4/5}$.
    The rectangle grows rightward in $m$ ``stages'', each stage of width $w$ and height $h = O(m k^2)$.\footnote{In this section, we use the term ``stage'' merely to mean different regions of the final assembly.  This is a different usage of ``stage'' than used in~\cite{DDFIRSS07}.}
    Each stage consists of the attachment of two ``horizontal bars'' to the right, which in turn cooperate to place a single ``vertical bar'', which (after some additional growth not shown) will contain glues for binding of horizontal bars in the next stage.
    The speedup is obtained by using ``binding parallelism'': the ability of a single (large) assembly $\beta$ to bind to multiple sites on another assembly $\alpha$.
    Think of $\alpha$ as the structure built so far, with a vertical bar on its right end, and think of $\beta$ as one of the horizontal bars shown in Figures~\ref{fig:vertical-bars-submission} and \ref{fig:vertical-bars-overview-submission}.
    This ``binding parallelism'' is in addition to ``assembly parallelism'': the ability for $\alpha$ to assemble in parallel with $\beta$ so that (a large concentration of) $\beta$ is ready to bind as soon as $\alpha$ is assembled.
    The number $k$ controls the amount of ``binding parallelism'': $k$ is the number of binding sites on $\alpha$ to which $\beta$ may bind, the \emph{first} of which binds in expected time $\frac{1}{k}$ times that of the expected time before any \emph{fixed} binding site binds (since the minimum of $k$ exponential random variables of expected value $t$ has expected value $\frac{t}{k}$).
    More precisely, two different versions of $\beta$ bind to one of two different regions on $\alpha$, each region having $k$ binding sites.
    The choice of spacing between binding sites is to ensure that each pair of vertical positions where the two horizontal bars could go are separated by a unique distance.
    This means that the vertical bar, despite having all possible binding sites available on its left side, will always bind in the same vertical position relative to the vertical bar to its left, ensuring that the \emph{shape} assembled is always a rectangle.
    These distances are also used to communicate and increment the current stage, encoded in the vertical position of the bottom horizontal bar, since the glues on its right are not specific to the stage.
    Because assembly may proceed as soon as each of the two regions has a $\beta$ bound (so that no individual binding site is required before assembly can proceed), the system is not a partial order system;
    in fact it is not even directed since different filler tiles will fill in the other $k-1$ regions where copies of $\beta$ could have gone but did not.
}

    Although we use mass-action kinetics to model changing concentrations, we occasionally use discrete language to describe the intuition behind reactions -- e.g., ``a copy of $A$ is consumed and two copies of $B$ are produced'' -- despite the fact that concentrations model continuously evolving real-valued concentrations.

    The hierarchical model will permit a speed-up over the seeded model.
    However, when viewed as a \emph{programming language} for tile assembly, the hierarchical model is more unwieldy to program and to analyze.
    Therefore we prove a number of lemmas showing that careful design of hierarchical tiles will cause them to ``behave enough like'' seeded tiles to remain tractable for analysis, and to ensure that the assembly proceeds sufficiently quickly.
    Much like parallel programming, in which critical regions are segregated into a few well-characterized parts of the program, we largely employ ``seeded-like assembly'' for most subcomponents of the construction, combining them using hierarchical parallelism at a small number of well-understood points.

    \subsection{Warm-up: A thin bar}

    We first ``warm up'' by analyzing in detail the assembly time of a simple but non-trivial system.
    The following lemma, Lemma~\ref{lem-thin-bar-mass-action} (more precisely, its corollary, Corollary~\ref{cor-thin-bar-mass-action-delta-1-over-n} that assigns concrete concentrations to the tile types), shows that it is possible to grow a ``substantial'' concentration of a ``hard-coded thin bar'' in quadratic time under the mass-action model.
    This structure will be the first subassembly formed in many other subcomponents of the tile system of Theorem~\ref{thm-fast-rectangle}.
    Furthermore, the proof of Lemma~\ref{lem-thin-bar-mass-action} will illustrate several techniques for analyzing hierarchical assembly time (and ``programming tricks'' to ensure that this time is fast).
    Section~\ref{subsec-techniques-bounding-assembly-time} generalizes these techniques to apply to more complex tile systems used in the full construction.
    However, these techniques are easier to understand by first reading this section with its simple, concrete tile system.

    To achieve quadratic time we use ``polyomino-safe'' tiles that grow a $2 \times n$ bar in a zig-zag fashion to enforce that no substantial growth nucleates except at the ``seed'', similar to the zig-zag tile set described by Schulman and Winfree~\cite{SchWin09,SchWin07} (which prevented spurious nucleation with high probability under the more permissive \emph{kinetic tile assembly model}~\cite{Winfree98simulationsof} that allows reversible attachments and strength-1 attachments).
    This enforces that no large overlapping subassemblies grow that would compete to consume tiles without being able to attach to each other (which happens with the $n$ tile types required to grow a $1 \times n$ bar).\footnote{Assembling a $1 \times n$ bar provably requires $\Theta(n \log n)$ time to reach half of its steady-state concentration~\cite{AdlCheGoeHuaWas01}, even when all $n$ tile types are allowed to have concentration as high as 1, exceeding the bound of the finite density constraint.  Enforcing the finite density constraint and assigning each tile type a concentration of $1/n$ gives a time bound of $\Theta(n^2 \log n)$.}

    \begin{lemma}\label{lem-thin-bar-mass-action}
      Let $\calT=(T,2)$ be the hierarchical TAS shown in Figure~\ref{fig:thin-bar}, and let $\ha$ be its unique terminal assembly of a $2 \times n$ bar.
      Let the initial concentrations be defined by $[s_1](0) = \delta > 0$, $[r_i](0)=2 \delta$ for all $i \in \{1,\ldots,n\}$, and $[s_i](0) = 3 \delta$ for all $i \in \{2,\ldots,n\}$.
      Then for all $t \geq \frac{4n}{\delta}$, $[\ha](t) \geq \frac{\delta}{2}$.
    \end{lemma}

    \begin{figure}[htb]
    \begin{center}
      \includegraphics[width=4in]{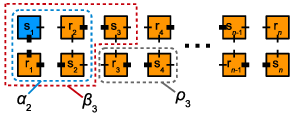}
      \caption{\label{fig:thin-bar} \figuresize
      Tile types that grow a $2 \times n$ bar.  Thick lines represent strength-2 glues, and thin lines represent strength-1 glues.
      Glues are not labeled, but each is hard-coded to represent its position in the final assembly.
      Zig-zag growth enforces that no overlapping subassemblies larger than size 2 can occur.
      All growth other than size-2 dimers (which will attach just as correctly as if they had stayed monomers) must nucleate from the ``seed'' labeled $s_1$.
      Examples are shown of subassemblies $\alpha_i$, $\beta_i$, and $\rho_i$, as defined in the proof of Lemma~\ref{lem-thin-bar-mass-action}.
      }
    \end{center}
    \end{figure}

    Intuitively, the reason for the choice of concentrations is to approximate the speed of seeded single-tile addition assembly, by enforcing that the concentrations of individual tiles (or dimers) other than $s_1$ remain for all time above at least a fixed constant $\delta$, to keep high their rate of reaction with a larger assembly containing ``preceding'' tile types.

    \begin{proof}
        For $i \in \{1,\ldots,n\}$, there are two types of producible assemblies containing $s_1$: the assembly with exactly $2i$ tiles, which we call $\alpha_i$ (its single frontier location is where $s_{i+1}$ binds), and the assembly with exactly $2i-1$ tiles, which we call $\beta_i$ (its single frontier location is where either $r_i$ or $\rho_i$ can bind, where $\rho_i$ is the assembly consisting of $r_i$ and $s_{i+1}$).
        Hence $\beta_1$ contains only $s_1$, and $\alpha_n = \ha$.


        For all $i \in \{1,\ldots,n-1\}$, let $\rho_i$ denote the dimer (2-tile assembly) consisting of just $r_i$ and $s_{i+1}$.
        Since $\beta_i$ contains $s_1$, for all $t \in \R_{\geq 0}$, $[\beta_i](t) \leq [s_1](0) = \delta$.
        Thus at most $\delta$ of the individual $r_i$'s can bind to $\beta_i$.
        The remainder
        must stay unbound or bind to $s_{i+1}$ to form $\rho_i$.
        For all $i \in \{1,\ldots,n-1\}$ and all $t \in \R_{\geq 0}$, by the fact that $[r_i](0)=2\delta$,
        \begin{equation}\label{ineq-dr}
        [\rho_i](t) + [r_i](t) \geq \delta,
        \end{equation}
        Since there is no $\rho_n$, we have
        \begin{equation}\label{ineq-rn}
        [r_n](t) \geq \delta
        \end{equation}
        for all $t\in\R_{\geq 0}$ by the same reasoning.
        By similar reasoning, for $i \in \{2,\ldots,n\}$, since $[s_i](0) = 3 \delta$, no more than $\delta$ of the individual $s_i$'s can bind to $\alpha_{i-1}$ to form $\beta_i$, and no more than $\delta$ of the remaining can bind to copies of $r_{i-1}$ that never attach to $\beta_{i-1}$.
        Thus for all $i \in \{2,\ldots,n\}$ and all $t \in \R_{\geq 0}$,
        \begin{equation}\label{ineq-s}
        [s_i](t) \geq \delta.
        \end{equation}

        Let
        \begin{eqnarray*}
        S(t)
        &=&
        \sum_{i=1}^{n} |\alpha_i| \cdot  [\alpha_i](t)\ +\ \sum_{i=1}^{n} |\beta_i| \cdot [\beta_i](t)
        \\&=&
        \sum_{i=1}^{n} 2i [\alpha_i](t)\ +\ \sum_{i=1}^{n} (2i-1)[\beta_i](t).
        \end{eqnarray*}
        $S(t)$ can be thought of as the total ``mass'' of tiles that belong to an assembly containing $s_1$ at time $t$.
        Observe that
        \begin{equation}\label{ineq-max-S-thin-bar}
            \sup_{t\in\R_{\geq 0}} S(t) = 2 n \delta,
        \end{equation}
        with the supremum $2 n \delta$ attained only in the limit as $t\to\infty$, when all $s_1$ belong to terminal assembly $\alpha_n$.

        The reactions $r_i + s_{i+1} \to \rho_i$ do not change $S(t)$.
        All other reactions increase $S(t)$.
        Each reaction that increases $S(t)$ is of the form $\alpha_i + s_{i+1} \to \beta_{i+1}$, $\beta_i + r_{i} \to \alpha_i$ (each of which increases $S(t)$ by 1 per unit concentration of the product produced), or $\beta_i + \rho_{i} \to \beta_{i+1}$ (which increases $S(t)$ by 2 per unit of product).
        Therefore, summing the propensities of all these reactions, we obtain
        \begin{equation*}
        \frac{dS(t)}{dt}  =  \sum_{i=1}^{n-1} [\alpha_i](t) \cdot [s_{i+1}](t)
        + \sum_{i=1}^{n} [\beta_i](t) \cdot [r_{i}](t)
        + \sum_{i=1}^{n-1} 2 \cdot [\beta_i](t) \cdot [\rho_{i}](t).
        \end{equation*}

        Note that for all $t \in\R_{\geq 0}$,
        \begin{equation} \label{eq-s1-sum-alphai-betai}
         [s_1](0) = \sum_{i=1}^{n} \left( [\alpha_i](t) + [\beta_i](t) \right).
        \end{equation}
        Each right-hand side term represents an assembly that could contain a copy of $s_1$.
        Then
        \begin{eqnarray}
        \frac{dS(t)}{dt}
        & = & \nonumber
        \sum_{i=1}^{n-1} [\alpha_i](t) \cdot [s_{i+1}](t)
        + \sum_{i=1}^{n} [\beta_i](t) \cdot [r_{i}](t)
        + \sum_{i=1}^{n-1} 2 \cdot [\beta_i](t) \cdot [\rho_{i}](t)
        \\& = & \nonumber
        [\beta_n](t) \cdot [r_n](t)
        + \sum_{i=1}^{n-1} [\alpha_i](t) \cdot [s_{i+1}](t)
        + \sum_{i=1}^{n-1} [\beta_i](t) \cdot ([r_{i}](t) + 2 \cdot [\rho_{i}](t))
        \\& \geq & \nonumber
        [\beta_n](t) \cdot [r_n](t)
        + \sum_{i=1}^{n-1} [\alpha_i](t) \cdot [s_{i+1}](t)
        + \sum_{i=1}^{n-1} [\beta_i](t) \cdot ([r_{i}](t) + [\rho_{i}](t))
        \\& \geq & \nonumber
        [\beta_n](t) \cdot \delta
        + \sum_{i=1}^{n-1} [\alpha_i](t) \cdot \delta
        + \sum_{i=1}^{n-1} [\beta_i](t) \cdot \delta
        \ \ \ \ \text{by \eqref{ineq-dr}, \eqref{ineq-rn}, and \eqref{ineq-s}}
        \\& = & \nonumber
        \delta \left( [\beta_n](t) + \sum_{i=1}^{n-1} ( [\alpha_i](t) + [\beta_i](t) ) \right)
        \\& = & \nonumber
        \delta ([s_1](0) - [\alpha_n](t))
        \ \ \ \ \ \ \ \ \ \ \ \ \ \ \ \ \ \ \ \ \ \ \ \ \ \ \ \ \ \ \ \text{by \eqref{eq-s1-sum-alphai-betai}}
        \\& = & \delta^2 - \delta [\alpha_n](t). \label{dS-lower-bound}
        \end{eqnarray}
        Since $\alpha_n$ is not a reactant in any reaction, $[\alpha_n](t)$ is monotonically increasing.
        Thus it suffices to prove that $[\alpha_n](\frac{4n}{\delta}) > \frac{\delta}{2}$.
        Suppose for the sake of contradiction that $[\alpha_n](\frac{4n}{\delta}) \leq \frac{\delta}{2}$.
        Then $[\alpha_n](t) \leq \frac{\delta}{2}$ for all $t \leq \frac{4n}{\delta}$ by the monotonicity of $[\alpha_n](t)$.
        By this bound and \eqref{dS-lower-bound}, $\frac{dS(t)}{dt} \geq \frac{\delta^2}{2}$ for all $t \leq \frac{4n}{\delta}$.
        Since $S(0) = \delta$, this means that $S(\frac{4n}{\delta}) \geq \delta + \frac{\delta^2}{2} \cdot \frac{4n}{\delta} = (2n+1)\delta$, which contradicts \eqref{ineq-max-S-thin-bar}.
    \end{proof}

    The following corollary shows that if we pick the initial concentrations to be maximal subject to the finite density constraint and the constraints of Lemma~\ref{lem-thin-bar-mass-action}, then quadratic time is sufficient to obtain terminal assembly concentration that is at least inversely linear.
    By Lemma~\ref{lem-conserve-mass}, any producible assembly $\alpha$ obeys $[\alpha](t) \leq \frac{1}{|\alpha|}$ for all $t\in\R_{\geq 0}$, so this concentration bound is optimal to within a constant factor.
    The time bound is asymptotically suboptimal\footnote{We assign concentrations of $\Theta(\frac{1}{n})$ to obey the finite density constraint since there are $\Theta(n)$ distinct tile types.  However, only $O(\sqrt{n})$ tile types are required to assemble a $2 \times n$ bar~\cite[Theorem 3.2]{AGKS05}.  The concentrations of these tile types could be set to $\Theta(\frac{1}{\sqrt{n}})$, lowering the half-completion time from $O(n^2)$ to $O(n^{1.5})$, if our goal in this section were to assemble a $2 \times n$ bar as quickly as possible (which it isn't).  However, since we later use the length-$n$ bar to encode $n$ bits, necessitating that each tile type on the top row be unique, we could not use $O(\sqrt{n})$ tile types anyway.} but sufficient for our purposes, since we only use hard-coded thin bars that are logarithmically smaller than the final assembly; hence their contribution to the assembly time is negligible.

    \begin{corollary}\label{cor-thin-bar-mass-action-delta-1-over-n}
      Let $\calT=(T,2)$ be the hierarchical TAS shown in Figure~\ref{fig:thin-bar}, and let $\ha$ be its unique terminal assembly of a $2 \times n$ bar.
      Let the initial concentrations be defined by $[s_1](0) = \delta = \frac{1}{c n}$ for some constant $c \in \R_{\geq 0}$, $[r_i](0)=2 \delta$ for all $i \in \{1,\ldots,n\}$, and $[s_i](0) = 3 \delta$ for all $i \in \{2,\ldots,n\}$.\footnote{We must choose $\delta \leq \frac{1}{5n-1}$ to obey the finite density constraint, hence $c \approx 5$.}
      Then for all $t \geq 4 c n^2$, $[\ha](t) \geq \frac{1}{2 c n}$.
    \end{corollary}

    The next lemma is a discrete version of Corollary~\ref{cor-thin-bar-mass-action-delta-1-over-n}, which shows that selecting $s_1$ as the timekeeper results in quadratic-time assembly of the bar under our stochastic assembly time model.

    \begin{lemma}\label{lem-thin-bar-discrete}
      Let $\calT=(T,2)$ be the hierarchical TAS shown in Figure~\ref{fig:thin-bar}, and define a concentrations function $C:T \to [0,1]$ as in the statement of Corollary~\ref{cor-thin-bar-mass-action-delta-1-over-n}, setting $C(s_1) = \delta = \frac{1}{c n}$ for some constant $c \in\R_{\geq 0}$, $C(r_i)=2\delta$ for $i \in \{1,\ldots,n\}$, and $C(s_i)=3 \delta$ for $i \in \{2,\ldots,n\}$.
      Then $\mathsf{T}(\calT,C,s_1) \leq 2 c n^2$.
    \end{lemma}

    \begin{proof}
      For all $j \in \{1,\ldots,2n-1\}$, let $t_j$ denote the expected time until the assembly containing $s_1$ grows by at least one tile, conditioned on the event that the current assembly is size at least $j$ but less than $2n$.
      By \eqref{ineq-dr}, \eqref{ineq-rn}, and \eqref{ineq-s}
      and the model of Markov process transition rates we employ to determine $\mathsf{T}(\calT,C,s_1)$, it holds that $t_j \leq 1/\delta = cn$ for all $j \in \{1,\ldots,2n-1\}$.
      By linearity of expectation,
      $
        \mathsf{T}(\calT,C,s_1)
        \triangleq
        \exp{\time_{\calT,C,s_1}}
        \leq
        \sum_{j=1}^{2n-1} t_j
        \leq
        (2n-1) c n
        <
        2 c n^2. 
      $
    \end{proof}

    \newcommand{\rs}{{s_1}}

    \subsection{General techniques for bounding assembly time}
    \label{subsec-techniques-bounding-assembly-time}

    Techniques from the proofs of Lemmas~\ref{lem-thin-bar-mass-action} and~\ref{lem-thin-bar-discrete} can be generalized in the following way to ease analysis of assembly time of well-behaved hierarchical systems.
    The results of this section will be our main technical tools used to bound the assembly time of the shape of Theorem~\ref{thm-fast-rectangle}.
    Intuitively, if the tile system is ``polyomino-robust'' (defined below), in the sense that the seeded and hierarchical models result in essentially the same producible assemblies and are well-behaved in other ways, then we can bound the hierarchical assembly time in terms of the size of the structure and the number of tile types needed to assemble it, if concentrations are set appropriately.

    The TAS of Figure~\ref{fig:thin-bar} has the following useful properties:
    \begin{enumerate}[1)]
      \item \label{prop-poly-robust-0} It is directed.

      \item \label{prop-poly-robust-1} There is a constant $q$ ($q=2$ in Figure~\ref{fig:thin-bar}) such that producible assemblies not containing the ``seed'' $s_1$ are of size at most $q$.
          We term such assemblies \emph{polyominos}.
          For mathematical convenience, we treat individual tile types that are not part of any polyomino as if they are polyominos of size 1, and we call larger polyominos \emph{nontrivial polyominos}.

      \item \label{prop-poly-robust-2} The set of producible assemblies containing $s_1$ is precisely the same in the seeded model as in the hierarchical model, and furthermore every terminal producible assembly contains $s_1$. (The tile system is \emph{polyomino-safe}, in the sense defined by Winfree \cite{Winfree06}.)
          In particular this implies that $\termasm{\calT} = \termasm{\calT_{s_1}}$, where $\calT_{s_1}=(T,\sigma,\tau)$ is the seeded version of $\calT=(T,\tau)$ with $\sigma$ containing only $s_1$.

      \item \label{prop-poly-robust-3} The polyominos that attach to (an assembly containing) $s_1$ are a ``total order (sub)system with respect to assemblies containing the seed''.
          More formally, define a \emph{maximal} polyomino $\alpha$ to be a polyomino such that is not attachable to any assembly not containing $s_1$.
          For each maximal polyomino $\alpha$, there is a strict total order $\prec$ on $\dom\alpha$ such that if $p_1 \prec p_2$, then the tile at position $p_1$ always attaches to an assembly containing $s_1$ by at least the time that the tile at position $p_2$ attaches.\footnote{Each polyomino is a ``chain'' with a well-defined tile ``closest'' to $s_1$.  Therefore, while $p_1$ and $p_2$ may attach at the same time, because the order is strict (implying $p_2 \not\prec p_1$) it is always possible for the tile at $p_1$ to attach strictly sooner.  In particular hierarchical growth is not \emph{required} for assembly to proceed.}

      \item \label{prop-poly-robust-4} Every tile type belongs to at most one type of maximal polyomino (which may appear in multiple locations in the terminal assembly), and appears exactly once in the polyomino.
          Here we include maximal polyominos of size 1, which means any tile type in a polyomino does not appear outside of the polyomino.
          More formally, for each maximal polyomino $\alpha$, $|T(\alpha)| = |\alpha|$, where $T(\alpha)$ is the set of tile types in $\alpha$, and for each pair of maximal polyominos $\alpha$ and $\beta$, $\alpha \neq \beta \implies T(\alpha) \cap T(\beta) = \emptyset$.
          Given a tile type $r$, we write $\rho(r)$ to denote the unique maximal polyomino in which $r$ is contained.
          This implies in particular that any tile type in a non-trivial polyomino appears in the terminal assembly equally often as any other tile type in the same polyomino.

    \end{enumerate}

    Say that a tile system (possibly a subset of a larger tile system) that satisfies these properties is \emph{polyomino-robust}.\footnote{Our full tile system assembling a rectangle is not directed, hence it does not satisfy Property~\ref{prop-poly-robust-0}. However, we will apply the lemmas proven in this section to subsets of the full tile system that \emph{are} directed, and in fact that satisfy all of the properties of polyomino-robustness.}
    Most useful seeded tile systems, when analyzed in the hierarchical model, tend to have these properties or are easily modified to have them.
    The two main tile subsystems that we analyze, shown in Figures~\ref{fig:horizontal-bar-with-arms} and~\ref{fig:vertical-bar}, can be verified by inspection to obey these constraints.

    The property of polyomino-robustness allows us to reason about the system, in certain senses, as if it were a seeded system.
    Properties \eqref{prop-poly-robust-1} and \eqref{prop-poly-robust-3}, in particular, allow us to set concentrations in such a way that we may assume that the concentration of individual tiles or polyominos that can extend an intermediate assembly are always at least a certain value bounded away from 0 ($\delta$ in Lemma~\ref{lem-thin-bar-mass-action}, and $\delta_1$ in Lemma~\ref{lem-polyonimo-fast-as-seeded-mass-action}).
    The trick is that tiles ``further from the seed'' (under the ordering $\prec$) are always at least $\delta$ greater concentration than tiles ``closer to the seed'', so that there will always be at least a $\delta$ excess of them in solution, no matter what combinations of partial polyominos form before attaching to the seed.
    Property~\eqref{prop-poly-robust-1} implies that we may use a bounded interval of concentrations (from $\delta$ to $3 \delta$ in Figure~\ref{fig:thin-bar}) to achieve this.
    Given a maximal polyomino $\alpha$ and a tile type $r$ in $\alpha$, define $\mathrm{dist}_\alpha(r)$ to be the distance of $r$ from the minimal position (under $\prec$) in the polyomino. (Since the polyomino is a linear chain, this number is well-defined).
    In Figure~\ref{fig:thin-bar}, for example, $\mathrm{dist}_{\rho_i}(0,0)=0$ and $\mathrm{dist}_{\rho_i}(1,0)=1$, where the polyomino $\rho_i$ is defined as in the proof of Lemma~\ref{lem-thin-bar-mass-action}: a size-2 polyomino containing $r_i$ at position $(0,0)$ (within the polyomino, assuming it is translated to the lower-left corner of the first quadrant) and $s_i$ at position $(1,0)$.


    Given an assembly $\beta$ and a tile type $r\in T$, define $\#_\beta(r) = |\setr{p \in \dom\beta}{\beta(p)=r}|$ as the number of times $r$ appears in $\beta$.
    If $\beta,\zeta$ are assemblies such that $\beta\sqsubseteq \zeta$, define $\zeta\setminus\beta$ to be the unique assembly $\gamma$ such that $\dom\gamma=\dom\zeta \setminus \dom\beta$ and $\gamma\sqsubseteq\zeta$.
    If $\beta$ is an assembly and $T$ is a tile set, define $T(\beta) = \mathrm{range}\ \beta = \setr{r \in T}{(\exists p\in\dom\beta)\ \beta(p)=r}$ to be the set of tile types in $\beta$.
    For any polyomino-robust TAS $\calT=(T,\tau)$ and $s_1 \in T$, let $\calT_\rs=(T,\sigma,\tau)$ denote the seeded version of the hierarchical system $\calT$, where $\sigma$ is the single-tile initial assembly consisting of only the tile $s_1$.
    For any producible (in the seeded model) assembly $\beta\in\prodasm{\calT_\rs}$, let $\prodasm{\beta} = \setr{\alpha\in\prodasm{\calT_\rs}}{\beta \sqsubseteq \alpha}$ denote the set of producible (in the seeded model) superassemblies of $\beta$, and let $[\prodasm{\beta}](t) = \sum_{\alpha \in \prodasm{\beta}} [\alpha](t) \cdot \#_\alpha(\beta) $, where $\#_\alpha(\beta)$ denotes the number of times that $\beta$ appears as a subassembly of $\alpha$.
    That is, $[\prodasm{\beta}](t)$ is the total concentration of $\beta$ or of assemblies containing $\beta$, where each duplicate appearance of $\beta$ in a single superassembly contributes to the concentration separately.\footnote{For this definition to make sense, we must weight the sum by the number of times $\beta$ appears in $\alpha$ because each time an assembly containing $\beta$ binds to another assembly containing $\beta$, the number of assemblies containing $\beta$ decreases by one, even though the total concentration of ``completed $\beta$'s'' has stayed the same.  However, whenever we actually apply this definition, it will be the case that $\#_\alpha(\beta)=1$ for any producible $\alpha$ such that $\beta \sqsubseteq \alpha$.}
    Note that, since assemblies can attach but not detach and $[\prodasm{\beta}](t)$ takes into account not only the assembly $\beta$ but any superassembly of it, $[\prodasm{\beta}](t)$ is monotonically nondecreasing with $t$: assemblies can attach to create new copies of $\beta$, but once formed $\beta$ cannot be broken apart.

    The next lemma shows conditions under which a partial assembly $\beta$ grows into a superassembly $\zeta \sqsupseteq \beta$.
    Informally, the lemma says that if we have ``substantial'' (at least $\delta_0$) concentration of $\beta$ (and its superassemblies), but the total concentration of seeds is ``small'' (at most $\delta_{s_1}$) compared to tile types that assemble an extension $\gamma$ of $\beta$ (to create a superassembly of $\beta$ called $\zeta$), and if the concentrations of those tile types are set to ensure that all individual tile types have ``excess'' concentration (at least $\delta_1$), then a ``substantial'' concentration of $\zeta$ will assemble in time linear in $|\gamma|$.



    \begin{lemma}\label{lem-polyonimo-fast-as-seeded-mass-action}
        Let $\calT=(T,\tau)$ be a polyomino-robust hierarchical TAS with unique terminal assembly $\ha$.
        Let $\rs \in T$.
        Let $\beta,\zeta\in\prodasm{\calT_\rs}$ such that $\beta\sqsubseteq \zeta$.
        Let $\gamma=\zeta \setminus \beta$, and suppose that $T(\gamma) \cap T(\ha \setminus \gamma) = \emptyset$. (i.e., tile types within $\gamma$ appear only within $\gamma$), and that all polyominos contained in $\zeta$ are completely contained in $\beta$ or completely contained in $\gamma$.
        Suppose also that for all $\alpha \in \prodasm{\beta}$, $\#_\alpha(\beta) = 1$, and for all $\alpha \in \prodasm{\zeta}$, $\#_\alpha(\zeta) = 1$.
        Suppose that there exist $t_0,c,\delta_0,\delta_{s_1} \in \R_{\geq 0}$ such that the following hold.

        \begin{itemize}
          \item $[s_1](0) \leq \delta_{s_1}$.

          \item $[\prodasm{\beta}](t_0) \geq \delta_0$.

          \item $\delta_{s_1} \leq \delta_0 c$.
        \end{itemize}
        Set initial concentrations of all $r \in T(\gamma)$ as follows.
        Let $\delta_1 > 0$.\footnote{Think of $\delta_1 = \frac{1}{|T(\gamma)|}$ for common usage of the lemma.  Intuitively, by our choice to concentrations, $\delta_1$ excess of each tile type is ensured even after they have been maximally consumed in attachment events, meaning we can think of $\delta_1$ as a lower bound on the concentration of tile types in the seeded model.}
        Set $[r](0) = \delta_{s_1} \cdot \#_\gamma(r) + (\mathrm{dist}_{\rho(r)}(r) + 1) \cdot \delta_1$.

        Then for all $t \geq \frac{2 |\gamma| c}{\delta_1} + t_0$, $[\prodasm{\zeta}](t) \geq \frac{\delta_0}{2}$.
    \end{lemma}

    \renewcommand{\poly}[2]{\rho_{#2}^{#1}}

    \begin{proof}
        To prove the lemma, we will first argue that for every producible assembly containing $\beta$ and containing a frontier location within $\gamma$, the total concentration of producible assemblies that could attach to this location is at least $\delta_1$.
        This provides a lower bound on the rate of reactions that grow the assembly into $\zeta$ (or a superassembly of $\zeta$).

        Let $\rho \sqsubseteq \gamma$ be a maximal polyomino consisting of tile types $r_1,r_2,\ldots,r_{k}$ at positions $p_1,\ldots,p_{k}$, with $p_1 \prec p_2 \prec \ldots \prec p_{k}$.
        Then we have an increasing sequence of polyominos $\rho_1 \sqsubseteq \rho_2 \sqsubseteq \ldots \sqsubseteq \rho_k = \rho$, where each $\rho_i$ is $\rho_{i-1}$ with the tile $r_i$ added at position $p_i$.

        Let $j \in \{1,\ldots,k\}$, and consider the polyominoes $\poly{j}{1} \sqsubseteq \poly{j}{2} \sqsubseteq \ldots \sqsubseteq \poly{j}{k-j+1} \in \prodasm{\calT}$, with $|\poly{j}{i}| = |\poly{j}{i-1}|+1$, that all have $r_j$ at their minimal position under $\prec$.
        In other words, $\poly{j}{i}$ contains exactly the tiles $r_j,\ldots,r_{j+i-1}$.
        In particular, for all $i \in \{1,\ldots,k\}$, $\rho_i = \poly{1}{i}$.
        All of these are attachable to some producible assembly containing $s_1$ via the tile type $r_j$.
        We allow $k=1$ so that $\poly{1}{1}$ can also represent any individual tile type $r \in T(\gamma)$ that is not part of a nontrivial polyomino.

        Call such a sequence $(\poly{j}{1},\ldots,\poly{j}{k-j+1})$ a \emph{polyomino attachment class}.
        Each such polyomino $\poly{j}{i}$ is consumed in a reaction in one of only two ways:
        \begin{enumerate}

        \item in a reaction that binds $r_j$ to $r_{j-1}$ (which have concentration at most $\delta_{s_1} \cdot \#_\gamma(r_{j-1}) + \mathrm{dist}_{\rho}(r_{j-1}) \cdot \delta_1 = \delta_{s_1} \cdot \#_\gamma(r_{j-1}) + (j-1) \cdot \delta_1$, which equals $\delta_{s_1} \cdot \#_\gamma(r_{j}) + (j-1) \cdot \delta_1$ by our assumption of equal counts of tiles that are part of the same polyomino, which is less than $[r_j](0)$ by $\delta_1$), or

        \item in a reaction producing another member of the same polyomino attachment class (thus not altering the sum of \eqref{ineq-r-hier}, just shifting its terms).  This corresponds to the attachment of tiles ``further from $\beta$ under $\prec$''.
        \end{enumerate}
        Therefore, there will always be at least an excess of $\delta_1$ concentration of polyominos in the attachment class $(\poly{j}{1},\ldots,\poly{j}{k-j+1})$, i.e., we have the following for all $t \in \R_{\geq 0}$:

        \begin{equation}\label{ineq-r-hier}
            \sum_{i=1}^{k-j+1} [\poly{j}{i}](t) \geq \delta_1.
        \end{equation}

        In particular, for any producible assembly $\alpha$ containing $\beta$, with frontier location $p \in \dom \gamma$ at which some tile type $r_j$ can attach in the seeded model, \eqref{ineq-r-hier} ensures that the total concentration of polyominos that can attach to position $p$ in the hierarchical model is at least $\delta_1$.

        Now we must argue that this lower bound on polyomino concentration implies that the claimed lower bound on concentration of assemblies containing $\zeta$ (i.e., $[\prodasm{\zeta}]$).
        The next inequality we derive (inequality~\eqref{eq-sum-assemblies-with-beta}) places a lower bound on the concentration of the \emph{other reactant} that reacts with the polyomino (the other reactant could be any assembly in the sum on the left side of~\eqref{eq-sum-assemblies-with-beta}) of the reactions that increase $[\prodasm{\zeta}]$.

        Recall that for all $\alpha \in \prodasm{\beta}$, $\#_\alpha(\beta) = 1$ and for all $\alpha \in \prodasm{\zeta}$, $\#_\alpha(\zeta) = 1$.
        Then for all $t\in\R_{\geq 0}$, $[\prodasm{\beta}](t) = \sum_{\alpha \in \prodasm{\beta}} [\alpha](t)$ and $[\prodasm{\zeta}](t) = \sum_{\alpha \in \prodasm{\zeta}} [\alpha](t)$.
        Since $[\prodasm{\beta}](t)$ is monotonically increasing, for all $t \geq t_0$,
        \begin{eqnarray*}
            [\prodasm{\beta}](t_0)
            &=& \sum_{\alpha \in \prodasm{\beta}} [\alpha](t_0)
            \nonumber
            \\&\leq& \sum_{\alpha \in \prodasm{\beta}} [\alpha](t)
            \nonumber
            \\&=&
            \sum_{\alpha \in \prodasm{\zeta}} [\alpha](t)
            + \sum_{\alpha \in \prodasm{\beta} \setminus \prodasm{\zeta}} [\alpha](t)
            \nonumber
            \\&=& [\prodasm{\zeta}](t) + \sum_{\alpha \in \prodasm{\beta} \setminus \prodasm{\zeta}} [\alpha](t).
        \end{eqnarray*}
        Recall that the hypothesis of the lemma supposes that $[\prodasm{\beta}](t_0) \geq \delta_0$.
        Combined with the previous inequality, this gives us for all $t \geq t_0$,
        \begin{equation} \label{eq-sum-assemblies-with-beta}
          \sum_{\alpha \in \prodasm{\beta} \setminus \prodasm{\zeta}} [\alpha](t)  \geq \delta_0 - [\prodasm{\zeta}](t).
        \end{equation}

        For all $\alpha\in\prodasm{\calT_\rs}$, define $|\alpha|_\gamma = |\dom\alpha \cap \dom\gamma|$ to be the number of tiles in $\alpha$ that are part of $\gamma$.
        For all $t\in\R_{\geq 0}$, define
        \[
            S_\beta^\gamma(t) = \sum_{\alpha \in \prodasm{\beta}} |\alpha|_\gamma \cdot [\alpha](t).
        \]
        Think of $S_\beta^\gamma(t)$ as the total ``mass within $\gamma$'' (concentration $\cdot$ (size within $\dom\gamma$)) of assemblies that contain $\beta$, noting that $|\alpha|_\gamma \leq |\gamma|$ for all $\alpha\in\prodasm{\calT_\rs}$.
        Since $\beta$ contains $s_1$,
        this implies that $[\prodasm{\beta}](t) \leq \delta_{s_1}$ for all $t\in\R_{\geq 0}$ by the bound on $[s_1](0)$.
        $S_\beta^\gamma(t)$ is maximized when all seed tiles have been incorporated into $\zeta$ or one of its superassemblies, which implies that
        \begin{equation} \label{ineq-max-S}
          \max_{t\in\R_{\geq 0}} S_\beta^\gamma(t)
          = \max_{t\in\R_{\geq 0}} \sum_{\alpha \in \prodasm{\beta}} |\alpha|_\gamma \cdot [\alpha](t)
          \leq |\gamma| \max_{t\in\R_{\geq 0}} \sum_{\alpha \in \prodasm{\beta}} [\alpha](t)
          = |\gamma| \max_{t\in\R_{\geq 0}} [\prodasm{\beta}](t)
          \leq \delta_{s_1} |\gamma|.
        \end{equation}

        Since every incomplete (with respect to completing $\zeta$) assembly $\alpha \in \prodasm{\beta} \setminus \prodasm{\zeta}$ that contains $\beta$ can react with at least one polyomino attachment class as defined in \eqref{ineq-r-hier} (in the inequality below, let $(\poly{j,\alpha}{1},\ldots,\poly{j,\alpha}{k-j+1})$ be the polyomino attachment class associated to $\alpha$), increasing the number of tiles occupying $\dom\gamma$ by at least 1, it follows that for all $t \geq t_0$,
        \begin{eqnarray}
            \frac{d S_\beta^\gamma(t)}{dt}
            &\geq& \nonumber
            \sum_{\alpha \in \prodasm{\beta} \setminus \prodasm{\zeta}} \left( \sum_{i=1}^{k-j+1} [\poly{j,\alpha}{i}](t) \right) \cdot [\alpha](t)
            \\&\geq& \nonumber
            \delta_1 \left( \sum_{\alpha \in \prodasm{\beta} \setminus \prodasm{\zeta}} [\alpha](t) \right) \ \ \ \ \text{by \eqref{ineq-r-hier}}
            \\&\geq& \label{ineq-dSdt-size-lemma}
            \delta_1 \delta_0  - \delta_1 [\prodasm{\zeta}](t).  \ \ \ \ \ \ \ \ \ \ \text{by \eqref{eq-sum-assemblies-with-beta}}
        \end{eqnarray}

	The first inequality above is an equality if there is only a single frontier location of $\alpha$ within $\dom \gamma$, to which any polyonimo in $(\poly{j,\alpha}{1},\ldots,\poly{j,\alpha}{k-j+1})$ could attach, but if there are multiple frontier locations in $\dom \alpha$, then $\frac{d S_\beta^\gamma(t)}{dt}$ would be even larger.

        It suffices to show that for $t^* = \frac{2 |\gamma| c}{\delta_1}+t_0$, $[\prodasm{\zeta}](t^*) \geq \frac{\delta_0}{2}$.
        (By the monotonicity of $[\prodasm{\zeta}](t)$, this will imply $[\prodasm{\zeta}](t) \geq \frac{\delta_0}{2}$ for all $t \geq t^*$ as well.)
        Suppose for the sake of contradiction that $[\prodasm{\zeta}](t^*) < \frac{\delta_0}{2}$.
        Then by \eqref{ineq-dSdt-size-lemma}, for all $t \in [t_0,t^*]$,
        $
            \frac{d S_\beta^\gamma(t)}{dt}
            \geq \delta_1 \delta_0  - \delta_1 [\prodasm{\zeta}](t)
            > \frac{\delta_1 \delta_0}{2}.
        $
        This implies that
        $
            S_\beta^\gamma(t^*)
            >  \frac{\delta_1 \delta_0}{2} (t^* - t_0)
            = \frac{\delta_1 \delta_0}{2} \frac{2 |\gamma| c}{\delta_1}
            = \delta_0 c |\gamma|
            \geq \delta_{s_1} |\gamma|,
        $
        which contradicts \eqref{ineq-max-S}.
    \end{proof}

    Intuitively, Lemma~\ref{lem-polyonimo-fast-as-seeded-mass-action} shows that one ``stage'' of assembly (the stage that goes from $\beta$ to $\zeta$) is ``fast.''
    The following lemma extends the analysis of Lemma~\ref{lem-polyonimo-fast-as-seeded-mass-action} to the case where we want to analyze multiple stages of assembly ($p$ is the number of stages in the statement of Lemma~\ref{lem-polyonimo-fast-multi-stage-mass-action}), where each stage may involve the addition of tile types of asymptotically different concentrations than the other stages.
    Intuitively, this is required to prove fast assembly time because some stages (such as hard-coding the seed row of a counter with $\omega(1)$ tile types) proceed slowly relative to their size since they require many tile types (hence each tile type has low concentration).
    This is not a problem since the size of such stages is small, but it implies that we cannot apply the ``average attachment time per tile'' of such a slow stage uniformly across the entire assembly.
    Some stages (such as completing a counter with a complete seed row with $O(1)$ tile types) proceed very quickly, and such stages account for ``most'' of the size of the final assembly, so that the total assembly speed is fast despite a few ``bottleneck stages'' in which the assembly process slows down for a short time.
    The statement of the lemma is quite intricate, but the proof is simple, setting up each stage of growth to match the hypothesis of Lemma~\ref{lem-polyonimo-fast-as-seeded-mass-action} and then applying that lemma iteratively.

    \begin{lemma}\label{lem-polyonimo-fast-multi-stage-mass-action}
        Let $\calT=(T,\tau)$ be a polyomino-robust hierarchical TAS with unique terminal assembly $\ha$.
        Let $\beta_1, \beta_2, \dots, \beta_p \in\prodasm{\calT_\rs}$ such that $\beta_{i-1}\sqsubseteq \beta_i$ for all $i\leq p$, where $\beta_1$ is the assembly with a single seed $s_1$.
        Let $\gamma_i=\beta_{i+1} \setminus \beta_i$ for each $i < p$, and supposed that $T(\gamma_i) \cap T(\gamma_j) = \emptyset$ for all $i, j < p$, and that all polyominos contained in $\beta_p$ are completely contained in $\gamma_i$ for some $i$. (i.e., suppose that each $\beta_p$ and $\beta_{p+1}$ satisfy the hypothesis of Lemma~\ref{lem-polyonimo-fast-as-seeded-mass-action}, where $\beta_{p+1}$ in this lemma is interpreted to be $\zeta$ in Lemma~\ref{lem-polyonimo-fast-as-seeded-mass-action}.)
        Suppose also that $\delta_{s_1} \leq \frac{1}{|\beta_p|}$.
        Set initial concentrations of all $r \in T(\gamma_i)$, where $i \leq p$ as $[r](0) = \delta_{s_1} \cdot \#_{\gamma_i}(r) + (\mathrm{dist}_{\rho(r)}(r) + 1) \cdot \frac{1}{|T(\gamma_i)|}$.

        Then for all $t \geq \sum_{i=1}^{p-1}
        2^i |\gamma_i| |T(\gamma_i)|$, $[\prodasm{\beta_p}](t) \geq \frac{\delta_{s_1}}{2^{p-1}}$.
        Furthermore, $\sum_{r \in T} [r](0) \leq 1 + p(q+1)$, where $q$ is the maximum size of any polyomino as in the definition of polyomino-robustness.
     \end{lemma}

    \begin{proof}
     We will show the statement $[\prodasm{\beta_m}](t) \geq \frac{\delta_{s_1}}{2^{m-1}}$ for all $t \geq \sum_{i=1}^{m-1}
     2^i |\gamma_i| |T(\gamma_i)|$ is true for all $m \leq p$ by induction on $m$.

    This statement is trivially true for $m=1$ since $[\prodasm{\beta_1}](t) = [\prodasm{s_1}](t) = \delta_{s_1}$ for all $t$.

    Assume that the statement is true for $m = m_0$. In other words, at time $t = \sum_{i=1}^{m_0-1} 2^i |\gamma_i| |T(\gamma_i)|$, $[\prodasm{\beta_{m_0}}](t) \geq \frac{\delta_{s_1}}{2^{m_0-1}}$. Invoking Lemma~\ref{lem-polyonimo-fast-as-seeded-mass-action} with $\beta = \beta_{m_0}$, $\gamma = \gamma_{m_0}$, $t_0 =  \sum_{i=1}^{m_0-1} 2^i |\gamma_i| |T(\gamma_i)|$, $c = 2^{m_0-1}$, and $\delta_1 = \frac{1}{|T(\gamma_i)|}$, we know that $[\prodasm{\beta_{m_0+1}}](t) \geq \frac{\delta_{s_1}}{2^{m_0}}$ at $ t\ =\ t_0 +  2^{m_0} |\gamma_{m_0}| |T(\gamma_{m_0})|\ =\ \sum_{i=1}^{m_0} 2^i |\gamma_i| |T(\gamma_i)|$. Therefore, the statement is true for $m = m_0 +1$.

    The total concentration of tiles whose type is in $T(\beta_p)$ is
    \begin{eqnarray*}
    \sum_{r\in T(\beta_p)} [r](0) & = & \sum_{i=1}^p \sum_{r\in T(\gamma_i)} [r](0) \\ \nonumber
    & = &  \sum_{i=1}^p \left( \delta_{s_1} \sum_{r\in T(\gamma_i)} \#_{\gamma_i}(r) + \sum_{r\in T(\gamma_i)} (\mathrm{dist}_{\rho(r)}(r) + 1) \cdot \frac{1}{|T(\gamma_i)|}   \right)  \\ \nonumber
    & \leq &  \sum_{i=1}^p \delta_{s_1} |\gamma_i| +\ \sum_{i=1}^p \max_{r \in T(\gamma_i)} [\mathrm{dist}_{\rho(r)}(r) + 1] \\ \nonumber
    & \leq &  \delta_{s_1} |\beta_p|\ +\ p \max_{r \in T(\gamma_i)} [\mathrm{dist}_{\rho(r)}(r)+1].
    \end{eqnarray*}
    Since $\delta_{s_1} \leq \frac{1}{|\beta_p|}$, and all polyominos have size at most $q$, the total initial concentration of all tile types is at most $1 + p (q+1)$.
    \end{proof}

    In particular, if the number of stages $p$ and polyomino size bound $q$ are constant with respect to the size of the terminal assembly (call this parameter $n$; as in our diameter $\Theta(n)$ rectangle of Theorem~\ref{thm-fast-rectangle}), then the total initial concentration of all tile types is constant with respect to $n$.  In particular, we can scale these concentrations to obey the finite density constraint without affecting the asymptotic time and concentration bounds derived in Lemma~\ref{lem-polyonimo-fast-multi-stage-mass-action}.

    The following is a discrete version of Lemma~\ref{lem-polyonimo-fast-as-seeded-mass-action} that can be used to analyze polyomino-robust systems in the discrete assembly time model.

    \begin{lemma}\label{lem-polyonimo-fast-as-seeded-discrete}
        Let $\calT=(T,\tau)$ be as in Lemma~\ref{lem-polyonimo-fast-as-seeded-mass-action}, and let $\alpha\in\prodasm{\beta}$ be the current state of the Markov process defining the assembly time of $\calT$, and suppose that the current time is $t_0$ and that $\beta \sqsubseteq \alpha$.
        Define $\time_{\alpha,\zeta}$ to be the random variable representing the first time at which $\alpha$ grows into a superassembly of $\zeta$.
        Define $\mathsf{T}(\alpha,\zeta) = \exp{\time_{\alpha,\zeta}}$.
        Then $\mathsf{T}(\alpha,\zeta) \leq \frac{|\gamma|}{\delta_1} + t_0$.
    \end{lemma}

    \begin{proof}
        For all $j \in \{0,\ldots,|\gamma|\}$, let $t_j$ denote the expected time until the assembly adds at least one more tile to $\dom\gamma$, conditioned on the event that current size of the assembly within $\dom\gamma$ is $j$.
        By \eqref{ineq-r-hier}
        and the model of Markov process transition rates we employ to determine $\mathsf{T}(\calT,C,s_1)$, it holds that $t_j \leq 1/\delta_1$ for all $j \in \{0,\ldots,|\gamma|-1\}$.
        By linearity of expectation,
        $
        \mathsf{T}(\alpha,\zeta)
        \triangleq
        \exp{\time_{\alpha,\zeta}}
        \leq
        t_0 + \sum_{j=0}^{|\gamma|-1} t_j
        \leq
        t_0 + \frac{|\gamma|}{\delta_1}. 
        $
    \end{proof}

    \subsection{Construction of a fast-assembling shape}

    This section describes the main components of the construction of Theorem~\ref{thm-fast-rectangle}.

    \begin{figure}[htb]
    \begin{center}
      \includegraphics[width=6in]{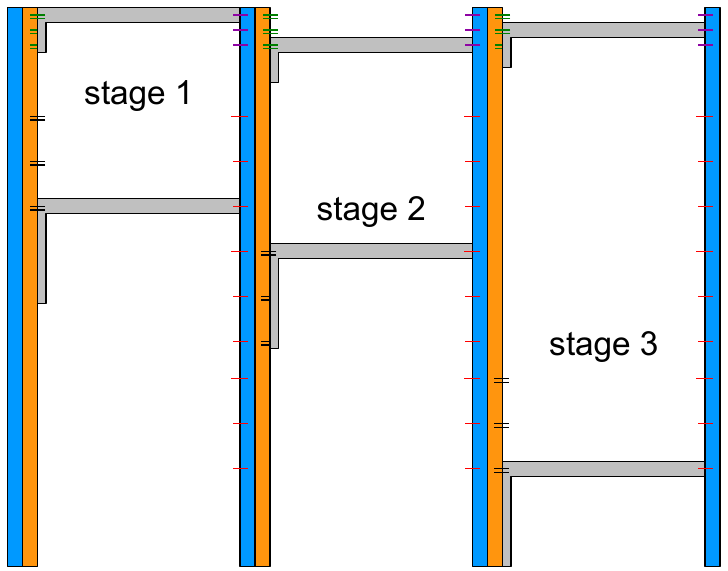}
      \caption{\label{fig:vertical-bars-overview} \figuresize
      High-level overview of interaction of ``vertical bars'' and ``horizontal bars'' to create the rectangle of Theorem~\ref{thm-fast-rectangle}.
      Filler tiles fill in the empty regions as shown in Figure~\ref{fig:vertical-bar-detailed}.
      If glues overlap two regions then represent a formed bond.
      If glues overlap one region but not another, they are glues from the former region but are mismatched (and thus ``covered and protected'') by the latter region.
      }
    \end{center}
    \end{figure}

    Figure~\ref{fig:vertical-bars-overview} shows an overview of the assembly described in Theorem~\ref{thm-fast-rectangle}.
    It consists of an initial (blue) ``vertical bar'', and $m$ copies of another type of (blue) vertical bar, each adjacent pair connected by a pair of (two different types of gray) ``horizontal bars''.
    The leftmost vertical bar forms, then the two horizontal bars attach, after which their right-side single strength glues cooperate to attach a new vertical bar to the right.
    This continues until the entire $m$ ``stages'' are complete.
    In the meantime, filler tiles fill in the gaps to complete the rectangle.
    Figures~\ref{fig:horizontal-bar-with-arms} and~\ref{fig:vertical-bar} show some more detail of the tile types that create the horizontal and vertical bars, and Lemmas~\ref{lem-horizontal-bar-mass-action} and~\ref{lem-vertical-bar-mass-action} respectively show that these subassemblies have ``substantial'' concentration ``quickly enough'' to be useful to prove the time bound of Theorem~\ref{thm-fast-rectangle}.


    Figures~\ref{fig:horizontal-bar-with-arms} and~\ref{fig:vertical-bar} show details of the tile types that assemble the ``horizontal bars'' and ``vertical bars'' of Figure~\ref{fig:vertical-bars-overview}, and Lemmas~\ref{lem-horizontal-bar-mass-action} and~\ref{lem-vertical-bar-mass-action} bound their assembly time.
    \opt{normal}{It may be beneficial for the reader first to skim the proof of Theorem~\ref{thm-fast-rectangle}, prior to examining Figures~\ref{fig:horizontal-bar-with-arms} and~\ref{fig:vertical-bar} in detail, in order to understand the intuitive purpose of the shape and outer glue placement of the horizontal and vertical bars.}

    For all $n\in\Z^+$, define $\log' n = \floor{\log n}+1$, the number of bits required to represent $n$ in binary, so that $2^{\log' n}$ is the next power of 2 greater than $n$.

    \begin{figure}[htb]
    \begin{center}
      \includegraphics[width=6.5in]{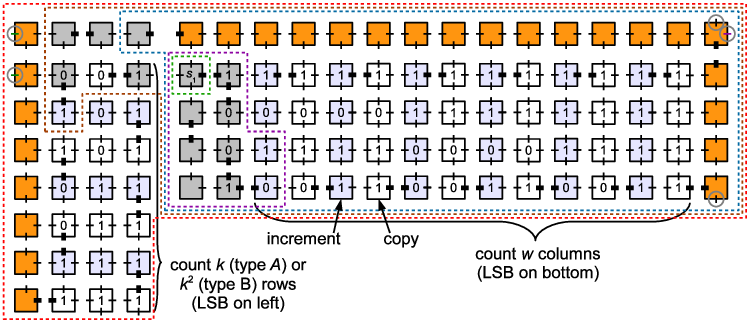}
      \caption{\label{fig:horizontal-bar-with-arms} \figuresize
      ``Horizontal bar with a vertical arm'' for the construction of a fast-assembling square.
      The arm (assembled by a downward-growing counter similar to the horizontal counter, and set up to grow only after the horizontal counter is complete) is intended to block other horizontal bars from binding after this one has bound to a vertical bar.
      The arm has height either $k$ (for type ``$A$'' horizontal bars of Figure~\ref{fig:vertical-bars}) or $k^2$ (for type ``$B$'' horizontal bars).
      Since we choose $w = k^4$, in either case the tile complexity and assembly time are dominated by the horizontal bar.
      The circles on the left and right indicate single-strength glues that are used when and after the bar binds hierarchically.
      The two west-facing single-strength glues cannot cooperate until the entire assembly is complete, so in particular the right glue must already be present.
      The east-facing single-strength glue is designed to cooperate with a glue from a different horizontal bar to control placement of a vertical bar, as in Figure~\ref{fig:vertical-bars}.
      The north and south facing single-strength glues are designed to help with filler tiles or stage-counting as in Figure~\ref{fig:vertical-bar-detailed}.
      }
    \end{center}
    \end{figure}


    \begin{lemma}\label{lem-horizontal-bar-mass-action}
      Let $\calT=(T,2)$ be the ``horizontal bar with an arm'' hierarchical TAS shown (by example) in Figure~\ref{fig:horizontal-bar-with-arms}, and let $\ha$ be its unique terminal assembly of an $O(w) \times O(\log w)$ horizontal bar with an $O(\log k) \times O(k)$ ``vertical arm'' on its left, where $w$ and $k$ are as in the proof of Theorem~\ref{thm-fast-rectangle} (so that $w \geq k$ in particular).
      Then there is an assignment of concentrations obeying $\sum_{r \in T} [r](0) = O(1)$ (with respect to $k$ and $w$) such that
      for all $t \geq O(w \log w)$, $[\ha](t) \geq \Omega(\frac{1}{w \log w})$.
    \end{lemma}

    \begin{proof}
        Let $\beta_1 \sqsubseteq \beta_2 \sqsubseteq \beta_3 \sqsubseteq \beta_4 \sqsubseteq \beta_5$ respectively represent the assemblies encircled by dotted lines of Figure~\ref{fig:horizontal-bar-with-arms}, so that $\beta_1$ is just the tile $s_1$, and $\beta_5 = \ha$.
        Set the initial concentrations of tile types in $T$ as in the statement of Lemma~\ref{lem-polyonimo-fast-multi-stage-mass-action}.
        Defining $\gamma_1,\ldots,\gamma_5$ as in Lemma~\ref{lem-polyonimo-fast-multi-stage-mass-action}, note that $|T(\gamma_1)| = O(\log w)$,
        $|T(\gamma_2)| = O(1)$,
        $|T(\gamma_3)| = O(\log k)$,
        $|T(\gamma_4)| = O(1)$,
        $|\gamma_1| = O(\log w)$,
        $|\gamma_2| = O(w \log w)$,
        $|\gamma_3| = O(\log k)$, and
        $|\gamma_4| = O(k \log k)$.
        Therefore Lemma~\ref{lem-polyonimo-fast-multi-stage-mass-action} tells us that for all
        $t \geq 2 \cdot O(\log^2 w) + 4 \cdot O(w \log w) + 8 \cdot O(\log^2 k) + 16 \cdot O(k \log k) = O(w \log w)$
        (since $w \geq k$), $[\prodasm{\ha}](t) \geq \delta_{s_1} / 16$.  Setting $\delta_{s_1} = \frac{1}{2 w \log' w}$ satisfies $\delta_{s_1} \leq \frac{1}{|\beta_5|}$, so that the total concentration of tile types is at most $1 + p(q+1)$, where $p=5$ and $q=4$ in the tile system of Figure~\ref{fig:horizontal-bar-with-arms}, where $q$ is defined to be the maximum size of any polyomino as in the definition of polyomino-robustness, giving the required constant concentration bound.
    \end{proof}


    \begin{figure}[htb]
    \begin{center}
      \includegraphics[height=5.3in]{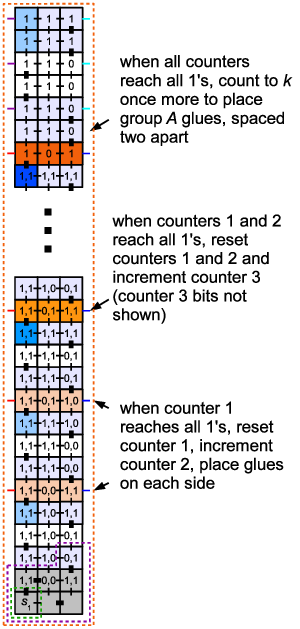}
      \caption{\label{fig:vertical-bar} \figuresize
      Tiles to assemble a ``vertical bar'' as in Figures~\ref{fig:vertical-bars-overview}, \ref{fig:vertical-bars}, and \ref{fig:vertical-bar-detailed} for the fast self-assembly of a rectangle.
      There are two types of vertical bars, the leftmost one with no glues on its west side (not shown), and the other vertical bars with glues each side (shown).
      }
    \end{center}
    \end{figure}


    \begin{lemma}\label{lem-vertical-bar-mass-action}
      Let $\calT=(T,2)$ be the vertical bar hierarchical TAS shown (by example) in Figure~\ref{fig:vertical-bar},
      and let $\ha$ be its unique terminal assembly of an $O(m k^2) \times O(\log k)$ rectangle, where $m$ and $k$ are as in the proof of Theorem~\ref{thm-fast-rectangle}.
      $\ha$ places $m k^2$ single strength ``type $B$'' (as in Figure~\ref{fig:vertical-bars}) glues on its left and right (all glues on the left identical, and all glues on the right identical to each other but different from the left glues), spaced every $k$ vertical rows, and another $k$ ``type $A$'' glues (as in Figure~\ref{fig:vertical-bars}) spaced 2 rows apart on the top left and right.
      Also, there is an assignment of concentrations obeying $\sum_{r \in T} [r](0) = O(1)$ (with respect to $k$ and $w$), such that for all
      $t \geq \Omega(m k^2 \log(mk))$, $[\ha](t) \geq \Omega(\frac{1}{m k^2 \log(mk)})$.
    \end{lemma}

    \begin{proof}
        The tiles are essentially zig-zag counters as described in~\cite{RotWin00}.
        Since single-strength glues must be placed in precise locations as required in Figures~\ref{fig:vertical-bars-overview}, \ref{fig:vertical-bars}, and \ref{fig:vertical-bar-detailed}, some modifications of the counter are necessary.
        There are three embedded counters 1, 2, and 3, counting to $k$, $k$, and $m$, respectively ($m=k$ in our construction), respectively.
        Counter 1 bits are shown on the left side of each tile, counter 2 bits on the right, and counter 3 bits are omitted.
        Counter 1 increments each row, and when it rolls over, it resets and counter 2 increments once.
        Similarly, when counter 2 rolls over, it and counter 1 reset and counter 3 increments.
        The values $m$ and $k$ are embedded in the first two rows and carried through each subsequent row to enable the resets.
        These place single strength glues on each side (the left glues to allow the vertical bar to bind to two horizontal bars as show in Figure~\ref{fig:vertical-bars}, and the right glues to help the orange counter in Figure~\ref{fig:vertical-bar-detailed} to correctly place glues on the right side of the vertical bar once it has attached.
        These single-strength glues are the ``group $B$'' glues of Figure~\ref{fig:vertical-bars}.
        Finally, when counter 3 rolls over, a new counter (using new tile types) is initiated to count to $2k$, placing a glue every other row, which are the ``group $A$'' glues of Figure~\ref{fig:vertical-bars}.
        To ensure that the vertical arm of the bottommost horizontal bar does not protrude below the bottom of the vertical bar, it is necessary to first count $k^2$ rows without placing side glues, but for space reasons this is not shown.
        Also, the first vertical bar must count an additional $\log m$ rows, since each subsequent vertical bar will add this many rows to the bottom, as shown in Figure~\ref{fig:vertical-bar-detailed}, when the stage computation counter must ``crawl'' below the bottom of the vertical bar in order to place the east-facing glues needed for the next-stage horizontal bars to bind to the vertical bar.

        Let $\beta_1 \sqsubseteq \beta_2 \sqsubseteq \beta_3$ respectively represent the assemblies encircled by dotted lines of Figure~\ref{fig:vertical-bar}, so that $\beta_1$ is just the tile $s_1$, and $\beta_3 = \ha$.
        Set the initial concentrations of tile types in $T$ as in the statement of Lemma~\ref{lem-polyonimo-fast-multi-stage-mass-action}.
        Defining $\gamma_1,\ldots,\gamma_3$ as in Lemma~\ref{lem-polyonimo-fast-multi-stage-mass-action}, note that $|T(\gamma_1)| = O(\log k + \log m)$,
        $|T(\gamma_2)| = O(1)$,
        $|\gamma_1| = O(\log k + \log m)$, and
        $|\gamma_2| = O(m k^2 (\log k + \log m))$.
        Therefore Lemma~\ref{lem-polyonimo-fast-multi-stage-mass-action} tells us that for all $t \geq 2 \cdot O((\log k + \log m)^2) + 4 \cdot O(m k^2 (\log k + \log m)) = O(m k^2 \log(mk))$, $[\prodasm{\ha}](t) \geq \delta_{s_1} / 4$.  Setting $\delta_{s_1} = \frac{1}{2 m k^2 (\log m + 2 \log k)}$ satisfies $\delta_{s_1} \leq \frac{1}{|\beta_3|}$ and gives a total concentration of tile types at most $1 + p(q+1)$, where $p=3$ and $q=3$ in the tile system of Figure~\ref{fig:horizontal-bar-with-arms}, where $q$ is defined to be the maximum size of any polyomino as in the definition of polyomino-robustness, giving the required constant concentration bound.
    \end{proof}

    We now prove \opt{normal}{the main theorem of this section,} Theorem~\ref{thm-fast-rectangle}.

    \begin{proof}[Proof of Theorem~\ref{thm-fast-rectangle}]
        Based on the primitives introduced in Figures~\ref{fig:thin-bar}, \ref{fig:horizontal-bar-with-arms}, and \ref{fig:vertical-bar} and their assembly time analysis, Figures~\ref{fig:vertical-bars-overview}, \ref{fig:vertical-bars}, and \ref{fig:vertical-bar-detailed} outline the remainder of the construction.

        Intuitively, the construction proceeds as follows.
        Fix positive integers $k$, $m$, and $w$ to be defined later.
        The rectangle grows rightward in $m$ ``stages'', each stage of width $w$ and height $h = O(m k^2)$.
        The speedup is obtained by using ``binding parallelism'': the ability of a single (large) assembly $\beta$ to bind to multiple sites on another assembly $\alpha$.
        Think of $\alpha$ as the structure built so far, with a vertical bar on its right end, and think of $\beta$ as one of the horizontal bars shown in Figures~\ref{fig:vertical-bars} and \ref{fig:vertical-bars-overview}.
        This ``binding parallelism'' is in addition to the ``assembly parallelism'' described in Section~\ref{sec-hierarchical-square}: the ability for $\alpha$ to assemble in parallel with $\beta$ so that (a large concentration of) $\beta$ is ready to bind as soon as $\alpha$ is assembled.
        The number $k$ controls the amount of ``binding parallelism'': it is the number of binding sites on $\alpha$ to which $\beta$ may bind, the \emph{first} of which binds in expected time $\frac{1}{k}$ times that of the expected time before any \emph{fixed} binding site binds (since the minimum of $k$ exponential random variables of expected value $t$ has expected value $\frac{t}{k}$).
        Actually two different versions of $\beta$ bind to one of two different regions on $\alpha$, each with $k$ binding sites.
        Because assembly may proceed as soon as each of the two regions has a $\beta$ bound (so that no individual binding site is required before assembly can proceed), the system is not a partial order system;
        in fact it is not even directed since different filler tiles will fill in the other $k-1$ regions where copies of $\beta$ could have gone but did not.

    \begin{figure}[htb]
    \begin{center}
      \opt{normal}{\includegraphics{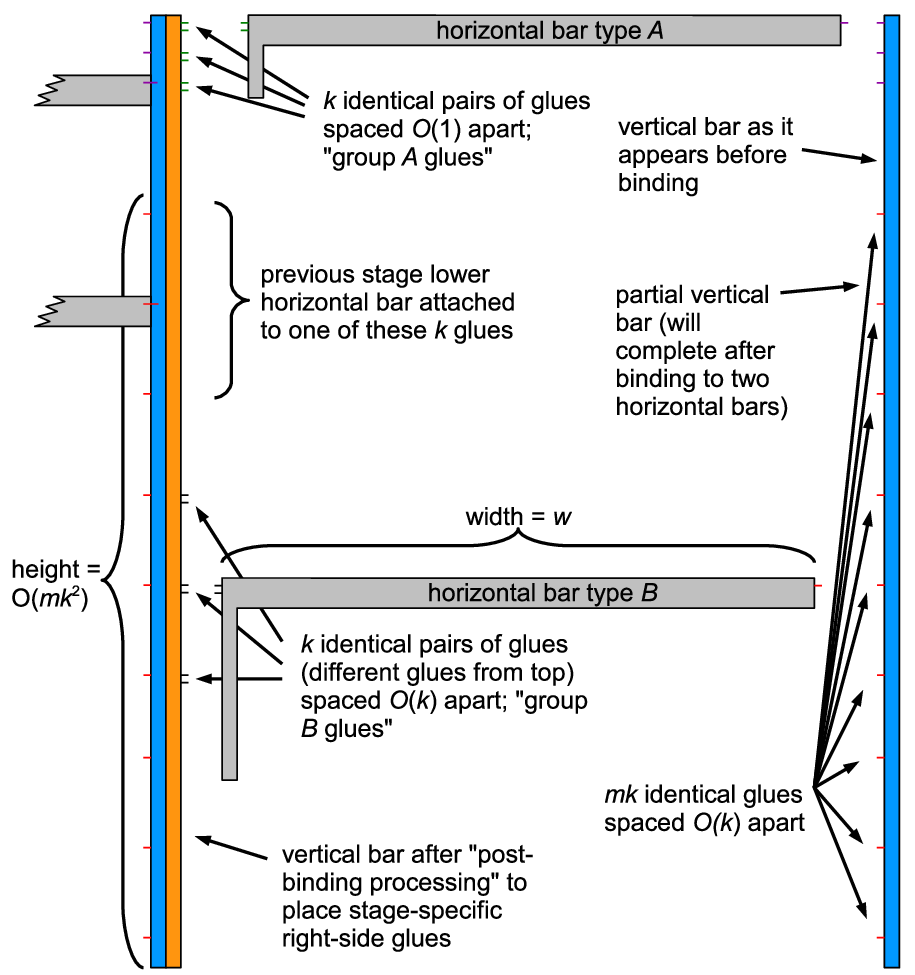}}
      \opt{submission}{\includegraphics{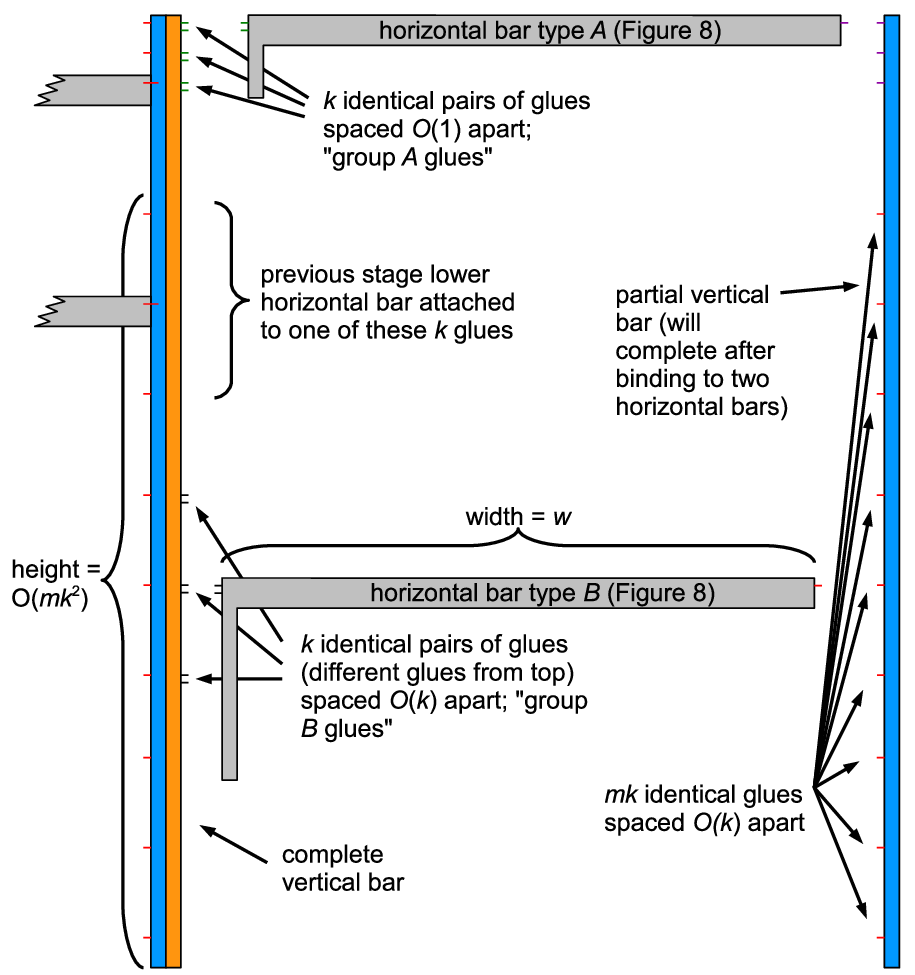}}
      \caption{\label{fig:vertical-bars} \figuresize
      ``Vertical bars'' for the construction of a fast-assembling square, and their interaction with horizontal bars of Figure~\ref{fig:horizontal-bar-with-arms}.
      ``Type $B$'' horizontal bars have a longer vertical arm than ``Type $A$'' since the glues they must block are farther apart.
      }
    \end{center}
    \end{figure}

    \begin{figure}[htb]
    \begin{center}
      \includegraphics[height=6in]{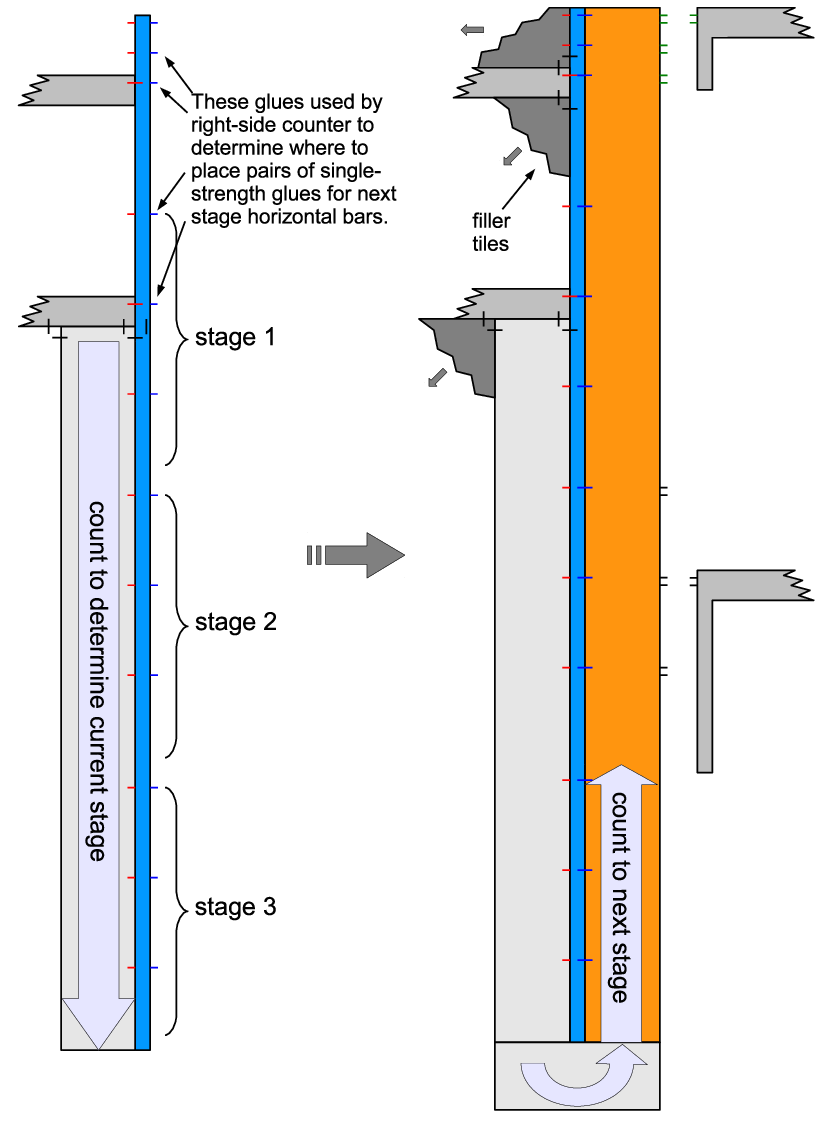}
      \caption{\label{fig:vertical-bar-detailed} \figuresize
      Detailed view of how a single partial vertical bar attaches completes into a full vertical bar after attaching to the previous stage's two horizontal bars.
      This enables the correct vertical placement of the two horizontal bars in the next stage.
      This is needed to communicate which stage is the current one based on the position of the previous lower horizontal bar, since there is only one type of vertical bar and it doesn't ``know'' the current stage.
      }
    \end{center}
    \end{figure}

        The timekeeper seed is contained in a height-$h$ ``vertical bar'' with a region of $k$ pairs of single-strength glues on its top right, and another region of $k$ (more widely spaced) pairs of single-strength glues on its right below the first region.
        Within each pair of glues, the two glues are different (despite being the same color in Figures~\ref{fig:vertical-bars} and \ref{fig:vertical-bars-overview}).
        However, within the first region, each pair is the same.
        Within the second region, each pair of glues is also identical to all other pairs in the second region, although each is distinct from the pair of glues of the first region.
        This vertical bar assembles as in Figure~\ref{fig:vertical-bar}, but rotated 90 degrees.
        The placement of the glues on the right is simple to calculate from the numbers $k$ and $m$, which are embedded into the tile types.
        The width of the vertical bar is therefore $O(\log k + \log m)$ (although most figures show the vertical bar as ``thin'', it is more than one tile wide) and requires at most $O(\log k + \log m)$ tile types to encode.
        The leftmost vertical bar, having no west-facing glues, has distinct tile types from the other vertical bars, but all other vertical bars are identical to each other.

        The horizontal bar of Figure~\ref{fig:horizontal-bar-with-arms} also has two types, but unlike the vertical bar, both types appear in each stage: type $A$ to bind in the top region and type $B$ to bind in the bottom region of the vertical bar, each with different single-strength glues on each end from the other type.
        The vertical ``arm'' on the left side below the bar is intended to prevent other horizontal bars from binding to another pair of glues in the same region of a vertical bar.
        An order of growth is chosen that enforces that the entire horizontal bar and the arm must be assembled before the two strength-1 glues on the left are present.
        This is required for speed; if a horizontal bar could bind to the right of a vertical bar before the horizontal bar is complete, then it would have to complete (taking at least time $w$ in the worst case) before growth of the larger assembly could continue.
        By ensuring that only complete horizontal bars can attach to vertical bars (and ensuring that many horizontal bars will be complete by the time ``most'' vertical bars require them), the assembly grows rightward quickly.
        Since the right side of a horizontal bar has only a single strength-1 glue, two horizontal bars, one in each of regions $A$ and $B$, are required to bind before the next stage's vertical bar may attach by using cooperative binding with each of these strength-1 glues.

        The following is a key idea in the construction: since there are only two kinds of horizontal bars, the glues on the right side cannot ``communicate'' the current stage.
        The natural solution to this, letting single-tile addition propagate the stage information from left to right along the horizontal bar, defeats the purpose of letting the horizontal bar attach hierarchically: such a solution would slow down the assembly process to be no faster than the seeded model.
        To enforce that stage $i+1$ properly follows stage $i$ (so that we deterministically stop after $m$ stages), the relative \emph{position} of the horizontal bars with respect to each other communicates the current stage.
        This is the reason that each stage ``staggers'' the vertical position of the group $B$ region of $k$ east-facing glues on the vertical bar, depending on the stage.

        In more detail, examine the single-strength glues on the left side of the vertical bar in Figure~\ref{fig:vertical-bar-detailed}.
        The group $A$ glues on the top left side of the vertical bar are each 2 spaces apart (to make room for their counterpart glues that will be placed on the \emph{right} side, which are twice as numerous since they cooperate in adjacent pairs).
        The group $B$ glues are each $k$ spaces apart.
        This ensures that every vertical distance between a glue from group $A$ and a glue from group $B$ is unique.
        Note that all vertical bars except the initial vertical bar have more single-strength glues on the left than on the right.
        These allow the vertical bars to bind at \emph{any} stage.
        However, by the distance-uniqueness property just explained, there is only one vertical position in which the vertical bar can bind, so that the vertical bar will be evenly lined up with the vertical bar from all the previous stages.
        Furthermore, it is possible, as shown in Figure~\ref{fig:vertical-bar-detailed}, to use a counter to \emph{measure} the relative height of the bottom horizontal bar in order to determine the current stage.
        The counter increments once for each group of $k$ single-strength glues that it passes (this can be implemented by marking boundaries between groups of $k$ glues with a special glue).

        As Figure~\ref{fig:vertical-bar-detailed} shows, once this counter reaches the bottom of the vertical bar and has value $h$ to indicate that the previous stage was stage $m-h$, it uses $h$ to determine where to place pairs of glues on the right side of the vertical bar for the next stage: they should placed at relative height $h-1$.
        This construction processing of the left side of the vertical bar to assemble the right side of the vertical bar can be done with a constant number of tile types.
        When the value $h$ is calculated to be 0, no right side of the vertical bar is constructed, since we have reached the final stage.
        As Figure~\ref{fig:vertical-bar-detailed} shows, filler tiles fill in the gaps above and below the horizontal bars after the vertical bar binds to the right.

        We now analyze the assembly time of this construction.
        Although the TAS we describe is not directed, the subcomponents are deterministic, so we may apply Lemmas~\ref{lem-polyonimo-fast-as-seeded-mass-action},~\ref{lem-polyonimo-fast-multi-stage-mass-action} and \ref{lem-polyonimo-fast-as-seeded-discrete} to them.

        We choose the timekeeper seed to be the same $s_1$ as shown in Figure~\ref{fig:vertical-bar}, for the alternate (not shown) version of the tiles that create the stage 1 vertical bar, with no glues on the left side.
        By our design in Figure~\ref{fig:vertical-bar-detailed}, the only way a horizontal bar can attach to any other
        assembly is to attach to a complete vertical bar to its left that part of an assembly containing $s_1$.

        Let $T_\text{h-A}$, $T_\text{h-B}$, $T_\text{v}$, and $T_{\text{v}_0}$ be the tile types to create horizontal type $A$ bars, horizontal type $B$ bars, vertical bars (partial, those colored blue in Figures~\ref{fig:vertical-bars} and~\ref{fig:vertical-bar-detailed}), and leftmost vertical bar, respectively.
        By the construction, we have $|T_\text{h-A}| = O(\log w + \log k)$, $|T_\text{h-B}| = O(\log w + \log k)$, and $|T_\text{v}| = |T_{\text{v}_0}| = O(\log w + \log k + \log m)$.

        For the general system, we add all four systems
        $T_{\text{h-A}}$, $T_{\text{h-B}}$, $T_{\text{v}}$, and $T_{\text{v}_0}$. We also add a constant number of counter
        tile types that wrap around the vertical bar as in Figure~\ref{fig:vertical-bar-detailed}, with each tile type having
        concentration $\Omega(1)$
        and a constant number of filler tile types also with total concentration $\Omega(1)$.
        In particular, assign the tile types concentrations as in the statement of Lemma~\ref{lem-polyonimo-fast-as-seeded-mass-action}, so that each tile type in the counter tiles and filler tiles (or polyomino attachment class) is guaranteed to have $\Omega(1)$ \emph{excess} concentration by assigning extra concentration to tile types in polyominos.
        (Although the details are not shown, it is trivial to implement the stage counter tiles with a polyomino-robust system, and to implement the filler tiles with no polyominos at all.)
        This allows us to apply Lemma~\ref{lem-polyonimo-fast-as-seeded-discrete} to these subsystems with $\delta_1 = \Omega(1).$

        Let $s_1$ denote the seed of the whole tile system.
        Let $\delta_{s_1} = [s_1](0) = \frac{1}{2 n n'}$.
        For our choice of $k$ and $m$, this will mean $\delta_{s_1} = O(\frac{1}{k^8})$.
        Let $\alpha_\text{h-A}$ denote the type $A$ ``horizontal bar with an arm'' assembly.
        Let $\alpha_\text{h-B}$ denote the type $B$ ``horizontal bar with an arm'' assembly.
        Let $\alpha_{\text{v}_0}$ denote the leftmost vertical bar assembly.
        Let $\alpha_{\text{v}}$ denote the other vertical bar assembly.
        By Lemma~\ref{lem-horizontal-bar-mass-action}, for all $t \geq O(w \log w)$, the total concentration of type $A$ horizontal bars that have been produced is as least $\Omega(\frac{1}{w \log w})$ (although some may be incorporated into assemblies containing the seed, and similarly for type $B$ vertical bars).
        Since at most $\delta_{s_1} m$ concentration of each horizontal bar can be attached to an assembly containing $s_1$, this implies that for all $t \geq O(w \log w)$, $[\alpha_\text{h-A}](t) \geq \Omega(\frac{1}{w \log w}) - \delta_{s_1} m$, which is $\Omega(\frac{1}{w \log w})$ by our choice of $\delta_{s_1}$.
        By the same reasoning, for all $t \geq O(w \log w)$, $[\alpha_\text{h-B}](t) \geq \Omega(\frac{1}{w \log w}).$
        Similarly, by Lemma~\ref{lem-vertical-bar-mass-action} and similar reasoning regarding the scarcity of $s_1$-containing assemblies that could attach to $\alpha_\text{v}$, for all $t \geq \Omega(m k^2 \log (m k))$, $[\alpha_\text{v}](t) \geq \Omega(\frac{1}{m k^2 \log(m k)})$.

        Having derived these concentration bounds in the mass-action model, we now analyze the stochastic assembly time of $s_1$ into a terminal assembly.

        By Lemma~\ref{lem-polyonimo-fast-as-seeded-discrete}, the leftmost vertical bar's bottom two rows assemble in expected time $O(\log k + \log m)$, using the fact that each tile type in the bottom two rows can be assigned concentration $\Omega(\frac{1}{\log k + \log m})$, to ensure that $\delta_1 = \Omega(\frac{1}{\log k + \log m})$ in Lemma~\ref{lem-polyonimo-fast-as-seeded-discrete}.
        Again applying Lemma~\ref{lem-polyonimo-fast-as-seeded-discrete}, the expected time to assemble the complete leftmost vertical bar from its seed row is at most $O(m k^2 \log (m k))$, using the fact that there are $O(1)$ tile types needed to complete the remaining rows and can be therefore be assigned concentration $\Omega(1)$ and in particular ensuring excess $\delta_1$ at least $\Omega(1)$ in Lemma~\ref{lem-polyonimo-fast-as-seeded-discrete}.

        Once the leftmost vertical bar completes, and if the current time $t = \Omega(w \log w)$, then by our above-derived bound on the concentration of $\alpha_\text{h-A}$, the expected time before a type $A$ horizontal bar attaches to some binding site in the group $A$ glues on the right side is at most $O(w \log w) / k$.
        This holds similarly for $\alpha_{h-B}$, so the expected time before both attachments happen is also at most $O(w \log w) / k$.
        Once two horizontal bars have attached, and if the current time $t = \Omega(m k^2 (\log k + \log m))$, by our above-derived bound on the concentration of $\alpha_\text{v}$, the expected time before the next vertical bar attaches is at most $O(m k^2 (\log k + \log m))$.
        Once this occurs, by Lemma~\ref{lem-polyonimo-fast-as-seeded-discrete}, the expected time before the constant-size tile set shown in Figure~\ref{fig:vertical-bar-detailed} to complete the placement of glues on the east side of the just-attached vertical bar is at most $O(m k^2 (\log k + \log m))$.
        At this point the first stage is complete, requiring time at most $O(w \log w) + O(m k^2 (\log k + \log m))$ (to wait for sufficient concentration of horizontal and vertical bars) $+ O(m k^2 (\log k + \log m))$ (to grow the leftmost vertical bar) $+ O(w \log w) / k$ (to attach two horizontal bars) $+ O(m k^2 (\log k + \log m))$ (to attach the second vertical bar) by linearity of expectation.
        Simplified, this is $O(w \log w) + O(m k^2 (\log k + \log m)) + O(w \log w) / k$.

        Repeating this analysis for each of the remaining stages, the total time for the complete ``skeleton'' of Figure~\ref{fig:vertical-bars} to complete is at most $m$ times the previous bound, $O(w \log w) + O(m^2 k^2 (\log k + \log m)) + O(m w \log w) / k$, by linearity of expectation.

        Finally, the filler tiles must tile the $3m$ empty regions left in the skeleton.
        Although this assembly process likely begins before the full skeleton is complete, we analyze it as if no filler tiles attach until the full skeleton is complete, as an upper bound for the actual assembly time.
        Each of these regions is a rectangle (minus the vertical arms of the horizontal bars) of diameter at most $O(w + m k^2)$, which is tiled by a constant-size rectilinear tile set (tile set in which each tile cooperates using north and east glues to grow towards the southwest), where each tile type has concentration $\Omega(1)$ at all times.
        By Theorem 4.4 of~\cite{AdChGoHu01}, the expected time for one of these regions to be completely tiled is at most $O(w + m k^2)$, with an exponentially decaying tail on the time distribution.
        Since there are $3m$ such regions assembling in parallel, and each has an exponentially decaying tail, the time for all regions to completely fill is at most $O((w + m k^2) \log m)$.

        Therefore the entire expected assembly time is at most
        $$
            O(w \log w) + O(m^2 k^2 (\log k + \log m)) + O(m w \log w) / k + O((w + m k^2) \log m).
        $$
        We choose $k = m = n^{1/5}$ and $w = k^4$.
        Simplifying the above expression, this gives an expected assembly time of
        $
            O(n^{4/5} \log n).
        $
    \end{proof}

    By using the base-conversion technique of~\cite{AdChGoHu01} for all counters, the number of tile types required could be reduced from $O(\log n)$ to the information-theoretically optimal $O(\frac{\log n}{\log \log n})$.
    However, this is a now-standard technique for obtaining optimal tile complexity of structures that ``encode'' a natural number $n$.
    Since the primary contribution of our construction is the bound on assembly time, we have presented a simpler (but larger than optimal) tile system for illustrative purposes.
    \opt{normal}{Unlike Theorem~\ref{thm-hierarchical-square}, where the problem of obtaining small assembly depth for a shape $S$ is trivialized by allowing tile complexity $|S|$, obtaining fast assembly \emph{time} is nontrivial whether tile complexity is large or not.
    In fact, small tile complexity \emph{helps} to obtain fast assembly time, since with fewer tile types, one can distribute to each tile type a greater share of the $O(1)$ concentration allowed by the finite density constraint, which in turn reduces the expected time for each tile to attach.}

    \opt{normal}{Also, it is possible to shave log factors from the assembly time analysis by using the ``optimal counter'' tiles of~\cite{AdChGoHu01}, which grow an $n \times \log n$ counter in the seeded model in time $O(n)$, compared to the suboptimal $O(n \log n)$ time required by the zig-zag counters we use.
    However, our Lemma~\ref{lem-polyonimo-fast-as-seeded-mass-action} does not take the ``binding parallelism'' of the seeded model into account, but instead implicitly assumes in the worst case that the frontier is always size 1.
    A more careful analysis could remove some of these log factors, but we have allowed the log factors in order to simplify the analysis, since our main goal is to obtain a sublinear time bound.}

\section{Nearly maximally parallel hierarchical assembly of a square with optimal tile complexity}
\label{sec-hierarchical-square}

In this section we show that under the hierarchical model of tile assembly, it is possible to self-assemble an $n \times n$ square, for arbitrary $n \in \Z^+$, using the asymptotically optimal $O(\frac{\log n}{\log \log n})$ number of tile types.
Furthermore, the square assembles using nearly the maximum possible parallelism in the hierarchical model, building the final square out of four assembled sub-squares of size $n/2 \times n/2$, which are themselves each assembled from four sub-squares of size $n/4 \times n/4$, etc.
The sub-optimality stems from the need for us to construct the smallest sub-squares of size $O(\log n) \times O(\log n) = O(\log^2 n)$ without parallelism.

We formalize the notion of ``parallelism through hierarchical assembly'' as follows.

\newcommand{\FigSquareHighLevelCaption}{
  Overview of the hierarchical TAS that assembles an $n \times n$ square with $O(\log^2 n)$ assembly depth and $O(\frac{\log n}{\log \log n})$ tile complexity.
  Each square in the figure represents a block of width $O(\log n)$ with each side of each block encoding its $(x,y)$-address in the square.
  \opt{normal}{(The encoding scheme is shown in more detail in Figure~\ref{fig:two-handed-square-block}.)}
  Each of the thin solid lines is a strength-1 glue intended to connect the block to other blocks.
  Dotted lines are drawn between those glues that are intended to bind to each other.
  The circled subassemblies show the order of growth of one particular block (at coordinates $(3,2)$) into the final square.
}

\opt{submission}{
    \newcommand{\twoHandedSquareFigWidth}{3.0in}
    \begin{wrapfigure}{r}{\twoHandedSquareFigWidth}
    \vspace{-25pt}
      \begin{center}
      \includegraphics[width=\twoHandedSquareFigWidth]{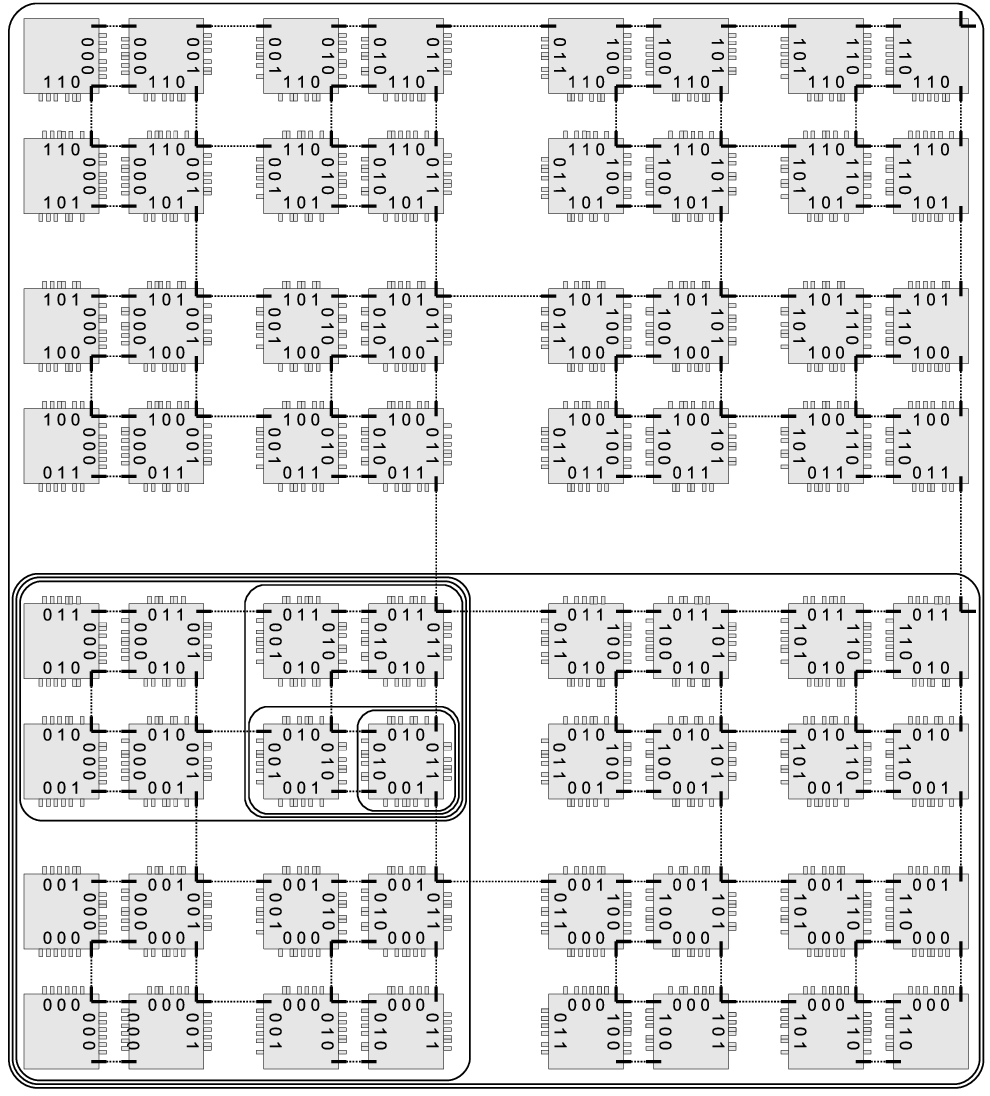}
      \caption{\label{fig:two-handed-square-high-level} \figuresize
      \FigSquareHighLevelCaption}
      \end{center}
    \vspace{-50pt}
    \end{wrapfigure}
}

Let $\calT=(T,\tau)$ be a directed hierarchical TAS.
Let $\alpha \in \prodasm{\calT}$ be a producible assembly.
An \emph{assembly tree} $\Upsilon$ of $\alpha$ is a full binary tree whose nodes are labeled by producible assemblies, with $\alpha$ labeling the root, individual tile types labeling the $| \alpha|$ leaves, and node $v$ having children $u_1$ and $u_2$ with the requirement that $u_1$ and $u_2$ can attach to assemble $v$.
That is, $\Upsilon$ represents one possible pathway through which $\alpha$ could be produced from individual tile types in $\calT$.
Let $\Upsilon(\calT)$ denote the set of all assembly trees of $\calT$.
Say that an assembly tree is \emph{terminal} if its root is a terminal assembly.
Let $\Upsilon_\Box(\calT)$ denote the set of all terminal assembly trees of $\calT$.
Note that even a directed hierarchical TAS can have multiple terminal assembly trees that all have the same root terminal assembly.
The \emph{assembly depth} of $\calT$ is
$
    \Kddac(\calT) = \max_{\Upsilon \in \Upsilon_\Box(\calT)} \mathrm{depth}(\Upsilon),
$
where $\mathrm{depth}(\Upsilon)$ denotes the standard depth of the tree $\Upsilon$, the length of the longest path from any leaf to the root.


It is clear by the definition that for any shape $S$ with $N$ points strictly self-assembled by a tile system $\calT$, $\Kddac(\calT) \geq \log N$.
Our construction achieves $\Kddac(\calT) \leq O(\log^2 n)$ in the case of assembling an $n \times n$ square $S_n$, while simultaneously obtaining optimal tile complexity $O(\frac{\log n}{\log \log n})$.\footnote{In~\cite{AGKS05}, the authors prove that whenever $n \in \N$ is algorithmically random, at least $\Omega(\log n / \log \log n)$ tile types are required to strictly self-assemble an $n \times n$ square in the hierarchical model.
Actually, that paper states only that this holds for the \emph{$q$-tile model}, in which some constant $q$ exists that limits the size of attachable assemblies other than those containing a special seed tile, and the authors claim that the proof requires the bound $q$, but in fact their proof does not use the bound $q$ and works for the general hierarchical model~\cite{SchwellerPersonalComm2016}.
Thus the tile complexity obtained in Theorem~\ref{thm-hierarchical-square} is asymptotically optimal.}
In other words, not only is it the case that every producible assembly can assemble into the terminal assembly (by the definition of directed), but in fact every producible assembly is at most $O(\log^2 n)$ attachment events from becoming the terminal assembly.

Demaine, Demaine, Fekete, Ishaque, Rafalin, Schweller, and Souvaine~\cite{DDFIRSS07} studied a complexity measure similar to assembly depth called \emph{stage complexity} for another variant of the aTAM known as \emph{staged assembly}.
In the staged assembly model, a hierarchical model of attachment is used, with the added ability to prepare different assemblies in separate test tubes.
The separate test tubes are allowed to reach a terminal state, after which any produced nonterminal assemblies (including individual tile) are assumed to be washed away, before combining the tubes.
The stage complexity of a tile system is similarly defined to be the depth of the ``mixing tree'' describing the order of test tube mixing steps.
Our model is more restrictive by permitting only one test tube (``bin complexity 1'' in the language of~\cite{DDFIRSS07}).
In a sense, Theorem~\ref{thm-hierarchical-square} ``automates'' the highly selective mixing that is assumed to be externally controlled in the staged assembly model, while paying only a quadratic price in the number of parallel assembly stages required.
It naturally pays a price in tile complexity as well, since unlike the staged model in which both the tile types and the mixing order can encode information, the construction of Theorem~\ref{thm-hierarchical-square} must encode the size $n$ of the square entirely in the tile types.
The primary challenge in achieving a highly parallel square construction in the hierarchical model --- a challenge not present in the staged assembly model --- is the prevention of overlapping subassemblies.\footnote{Adleman~\cite{Adl00} showed a $\Omega(n)$ lower bound (in a much different and more permissive model of assembly time than in the present paper; later improved to $\Omega(n \log n)$ by Adleman, Cheng, Goel, Huang, and Wasserman~\cite{AdlCheGoeHuaWas01}) for the problem of assembling a $1 \times n$ line from $n$ distinct tile types $t_1,\ldots,t_n$.
The intuitive reason that the time is not $O(\log n)$ is that if assemblies $\alpha_1 = t_i \ldots t_j$ and $\alpha_2 = t_{i'} \ldots t_{j'}$ form, with $i < i' < j < j'$, then $\alpha_1$ can never attach to $\alpha_2$ because they overlap.
Staged assembly can be used to control the overlap directly by permitting only the growth of lines covering dyadic intervals.}

\opt{normal}{
    \begin{figure}[htb]
    \begin{center}
      \includegraphics[width=5.7in]{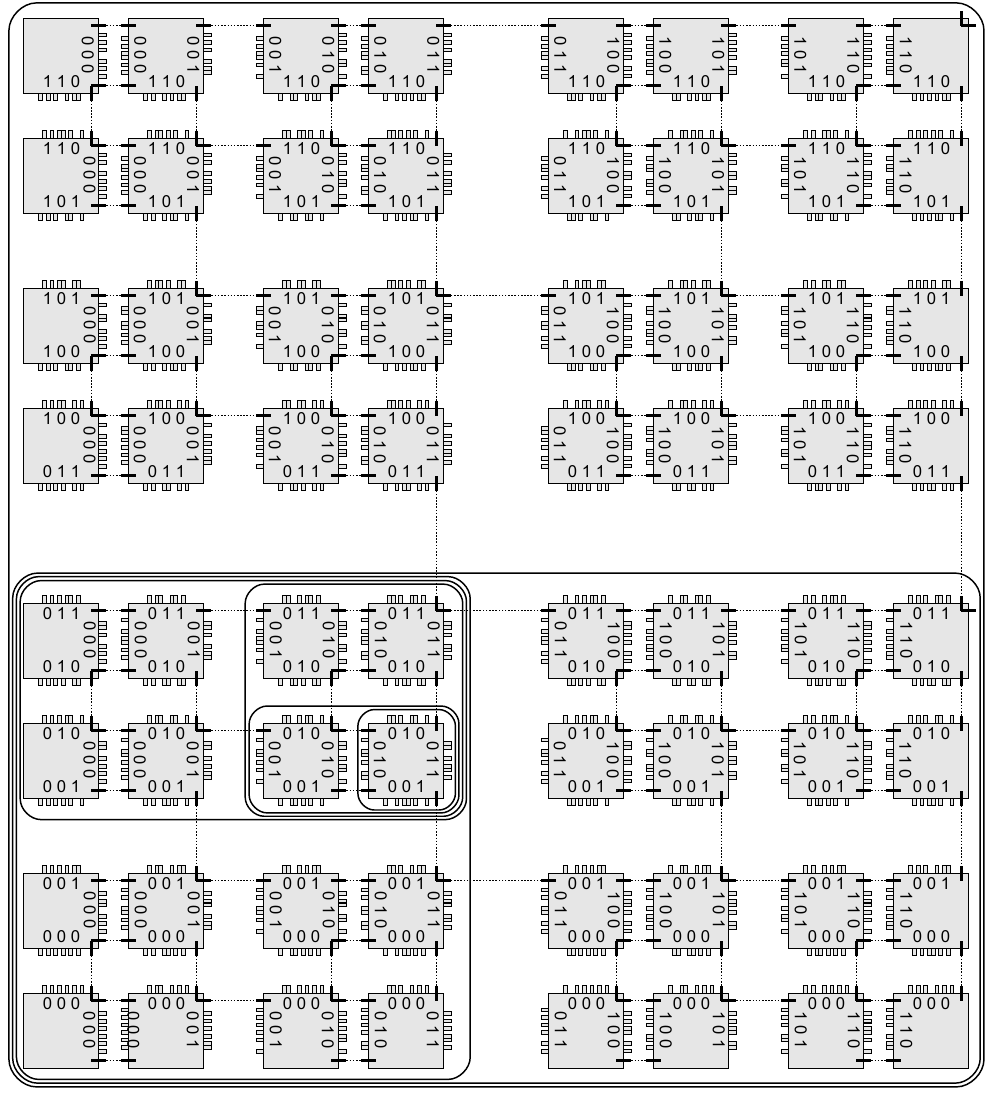}
      \caption{\label{fig:two-handed-square-high-level} \figuresize
      \FigSquareHighLevelCaption
      }
    \end{center}
    \end{figure}
}

\newcommand{\FigNonPowTwo}{
    \begin{figure}[htb]
    \begin{center}
      \includegraphics[width=4in]{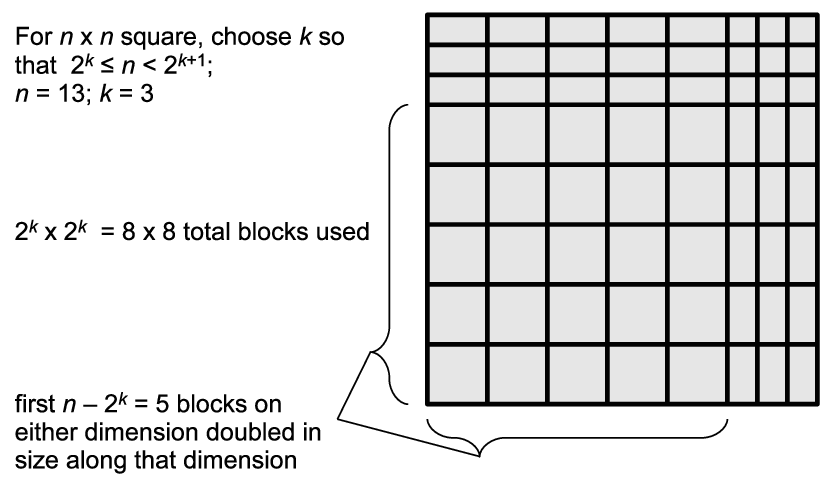}
      \caption{\label{fig:two-handed-square-non-power-of-2} \figuresize
      Design of block sizes to handle values of $n$ that are not a power of two. There are always exactly $2^k \times 2^k$ blocks, where $2^k \leq n < 2^{k+1}$.
      Each block doubles its length along the $x$-axis (resp. $y$-axis) if $n - 2^k$ exceeds its $x$-coordinate (resp. $y$-coordinate).
      }
    \end{center}
    \end{figure}
}

\newcommand{\FigBlock}{
    \begin{figure}[htb]
    \begin{center}
      \includegraphics[width=6.5in]{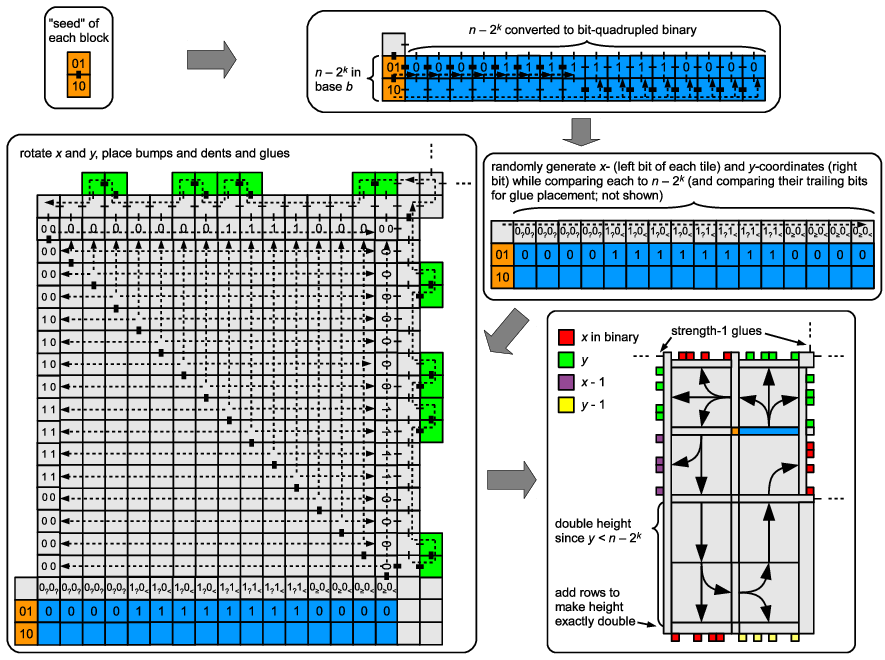}
      \caption{\label{fig:two-handed-square-block} \figuresize
      Assembly of $O(\log n) \times O(\log n)$ size block from $O(\frac{\log n}{\log \log n})$ tile types.
      Every block starts from the same tile types that encode $n-2^k$, using exactly $k$ bits; in this example, $n=22$ so $k=4$ and $n-2^k$ is 0110 in binary using 4 bits.
      Thick solid lines represent strength-2 glues.
      For clarity, strength-1 glues are shown selectively to help verify that a certain order of growth is possible to enforce.
      The tiles encode $n-2^k$ in base $b$ chosen to be a power of two such that $\frac{\log n}{\log \log n} \leq b < 2  \frac{\log n}{\log \log n}$ (labeled ``seed'' for intuition, although those tiles start unattached).
      $n-2^k$ is first converted to binary and each of its bits quadrupled to make room for the bumps and dents.
      A constant-size set of tile types does the rest.
      Then $x$ and $y$ coordinates are randomly guessed and simultaneously compared to $n-2^k$; if either is smaller, that dimension is doubled in length (in this example the height is doubled but not the width).
      At the same time, the values of $x$ and $y$ are compared as described in the proof of Theorem~\ref{thm-hierarchical-square} to determine where to place strength-1 glues.}
    \end{center}
    \end{figure}
}

\opt{normal}{
    \FigNonPowTwo

    \FigBlock
}

\begin{theorem} \label{thm-hierarchical-square}
    For all $n\in\N$, there is a hierarchical TAS $\calT=(T,2)$ such that $\calT$ strictly self-assembles an $n \times n$ square, $|T| = O(\frac{\log n}{\log \log n})$, and $\Kddac(\calT) = O(\log^2 n)$.
\end{theorem}

\opt{submission}{
    Figure~\ref{fig:two-handed-square-high-level} shows the high-level idea.
    The figure shows how to create a square of $n \times n$ blocks, each of width $O(\log n)$, for $n$ a power of two.
    The full proof explains how to handle $n$ that are not a power of two and how to handle $n$ that are not perfectly divisible by the block length.
    Intuitively, tiles aggregate into blocks while randomly guessing their ``$(x,y)$ address'' in the final square.
    The blocks encode this address into the bumps and dents on the side of the block, which ensures that fully formed blocks can only bind to their neighbors (the strength-1 glues shown in Figure~\ref{fig:two-handed-square-high-level} are all identical so the bumps and dents are required to encode this specificity.)
    The growth order of the blocks enforces that the strength-1 glues needed for binding are not present until enough bumps and dents are present to enforce this binding specificity.
    The $(x,y)$ address is also used to determine where to put strength-1 glues, since not all sides of all blocks have them.
    This enforces the ``recursive doubling'' of Figure~\ref{fig:two-handed-square-high-level}, namely that all blocks bind without overlap, so that blocks at level $i$ of the recursive hierarchy cannot bind to their level $i$ ``twin'' set of blocks until both have completely finished all level $i-1$ binding.
}

\newcommand{\ProofParallelSquare}{
   A high-level outline of the construction is shown in Figure~\ref{fig:two-handed-square-high-level}.
   We assemble a number of blocks of width $O(\log n) \times O(\log n)$, each of which represents in its tile types an address indicating its position in the square, and the block binds only to (some of) its neighboring blocks.
   The blocks assemble using standard single-tile accretion (actually we cannot directly enforce this in the model, but each block will nonetheless assemble the same structure in either model).
   Since each block is $O(\log^2 n)$ total tiles, this is the source of the suboptimal $O(\log^2 n)$ assembly depth.
   Once the blocks are assembled, however, they assemble into the full square using $O(\log n)$ assembly depth.
   All blocks $(x,y)$ with $x$ even bind to $(x+1,y)$ to create the two-block assembly $(x,y):(x+1,y)$, then all blocks $(x,y):(x+1,y)$ with $y$ even bind to $(x,y+1):(x+1,y+1)$ to create the four-block assembly $(x,y):(x+1,y):(x,y+1):(x+1,y+1)$, etc.

   The construction will actually control the width of the square only to within an additive logarithmic factor by bring together blocks of width and length $\Theta(\log n)$; standard techniques can be used to make the square precisely $n \times n$.
   For instance, we could add $O(\log n)$ total filler tiles to the leftmost and bottommost blocks, while adding only $O(\log^2 n)$ to the assembly depth and $O(\log \log n)$ to the tile complexity since such filler tiles could be assembled from a counter that counts to $\log n$ using $O(\log \log n)$ tile types.
   For simplicity we describe the desired width $n$ as the number of blocks instead of the desired dimensions of the square and omit the details of this last step of filling in the logarithmic gap.

   Figure~\ref{fig:two-handed-square-high-level} outlines the construction of a square when the number of blocks $n$ is a power of two.
   Figure~\ref{fig:two-handed-square-non-power-of-2} shows how to modify the blocks so that some of them are double in width, double in height, or both, to achieve a total square width that is an arbitrary positive integer.
   Each block contains the same $O(\frac{\log n}{\log \log n})$ tile types that encode $n$, and as the block assembles it randomly chooses $x$ and $y$-coordinates, which represent an index in the square.
   This random choice is implemented through competition between tile types that share the same ``input'' glues but represent different bits of $x$ or $y$. 
   These are used to determine the block's own size and to determine what series of bumps and dents to place on its perimeter to enforce that the only blocks that can bind are adjacent in Figure~\ref{fig:two-handed-square-high-level}.
   The coordinates are also used to determine where to place strength-1 glues. The same strength-1 glue is used uniformly throughout Figure~\ref{fig:two-handed-square-high-level}.
   The bumps and dents ensure that no two blocks can attach unless they are adjacent in the figure.

   The growth of an individual block is shown in Figure~\ref{fig:two-handed-square-block}.
   We describe the assembly as if it grows only by single-tile accretion.
   There are some strength-2 glues so this is not completely accurate, but the growth of the block is ``polyomino-safe'', to borrow a term of Winfree~\cite{Winfree06}.
   By design, no assembly larger than four can form except by attachment to the growing block, and even if these assemblies attach at once to the block rather than by single-tile accretion, the correct operation of the block growth is unaffected.
   This is due to the fact that all strength-2 glues are ``one-to-one''; no strength-2 glue is shared as an ``output'' (in the direction of growth in the seeded model) by two different tile types.
   This implies that no partial assembly occurring away from the main ``seeded'' assembly can grow ``backward'' and place an incorrect tile.

   To form a square of size $n \times n$ ``units'' (where a unit is $O(\log n)$, the width and height of a small block), we choose the largest power of two $2^k \leq n$ and assemble exactly $2^k \times 2^k$ different types of blocks, doubling the width (resp. height) of the first $n-2^k$ of them in the $x$-direction (resp. $y$-direction), as in Figure~\ref{fig:two-handed-square-non-power-of-2}.
   The orange (medium darkness in gray-scale) tile types and the base-conversion tile types that attach to them in Figure~\ref{fig:two-handed-square-block} are the only non-constant set of tile types.
   Borrowing a technique from~\cite{AdChGoHu01}, we will represent $n$ in base $b$, where $b \approx \log n / \log \log n$, using $\approx \log n / \log \log n$ unique tile types, and we use $O(\log n / \log \log n)$ tile types to convert $n-2^k$ to binary and $O(1)$ tile types to accomplish all the other tasks needed to assemble the block.

   Choose $b=2^m$ to be a power of two such that $\log n / \log \log n \leq b < 2 \log n / \log \log n.$
   Each digit in base $b$ can represent $m$ bits of $n-2^k$.
   $n-2^k$ is encoded in exactly $m \cdot \ceil{\frac{k}{m}} = O(\frac{\log n}{\log \log n})$ base-$b$ digits.
   The blue (dark in gray-scale) tile types in Figure~\ref{fig:two-handed-square-block} convert $n-2^k$ from base $b$ to binary and at the same time represent $n-2^k$ with its ``bit-quadrupled'' version (e.g., $0110 \mapsto 0000111111110000$), since each bit along the edge will eventually require width four to make room to place the bumps properly.\footnote{The bumps cannot simply be placed with strength-2 glues above a width-1 or even width-2 representation of a bit in the obvious way, otherwise there would be nothing to force that the bumps are present before the inter-block strength-1 glues.  If the bumps are allowed to grow in parallel with the rest of the assembly then they may not complete fast enough.  Width four is required to create a ``linear assembly path'' for the bumps and dents tiles to follow, ensuring that growth of the block continues only once the path is complete.}
   The set of base-conversion tile types from~\cite{AdChGoHu01} of cardinality $O(\log n / \log \log n)$ can be easily modified to achieve this ``bit-quadrupling'' without increasing the asymptotic tile complexity.
   The potential doubling of height and/or width can be achieved with a constant set of tile types since the unit width is implicitly encoded in the width of the block and a constant set of ``rotator'' tiles similar to those in Figure~\ref{fig:two-handed-square-block} can be used to add extra unit-width blocks when needed.
   The counterclockwise order of growth ensures that if not all of the bumps and dents are formed, then at least one of the four strength-1 glues necessary for an attachment event to occur is not yet present in one of the blocks.
   To ensure that the TAS is directed, we do not include base-conversion tiles for any digit $d \in \{0,1,\ldots,b-1\}$ that does not appear in the base-$b$ expansion of $n-2^k$, otherwise those tiles will form unused terminal assemblies.
   Each glue in a tile type representing a bit is ``marked'' indicating whether it is the most significant bit, least significant bit, or interior bit, as well as being marked with its relative position among the four copies of the bit.

   Once $n-2^k$ is converted to binary, we use nondeterministic attachment of tiles to the north of this value to randomly guess $2k$ bits that represent the $x$- and $y$-coordinates of the block, meaning the binary numbers represented on the top and right, respectively, of each block in Figure~\ref{fig:two-handed-square-high-level}.
   To be precise, we must actually choose each of $x$ and $y$ to be a random bit string that is not all 1's, since each represents a connection between two blocks, of which there are $2^k-1$ along each dimension.
   It is straightforward to encode into the tile types the logic that if the first $k$ bits were 1, then the final bit must be 0.
   A number of additional computations are done on these values (some computations are possible to do as the values are guessed).
   The results of these computations will be stored in the rightmost tile type and propagated to all subsequent tile types.
   First, each of $x$ and $y$ is compared to $n-2^k$ to determine how large to make each dimension of the block.
   In the example of Figure~\ref{fig:two-handed-square-block}, $y < n-2^k$ and $x \geq n-2^k$, so the block is one ``unit'' wide and two ``units'' high.
   Also, the binary expansions of $x$ and $y$ are themselves compared to determine where to place strength-1 glues.
   After $x$ and $y$ are determined, to place bumps and dents on the left and bottom of the block, the values $x-1$ and $y-1$ must be computed, which requires assembling from least significant to most significant, so this is delayed until after the first 90-degree rotation shown in Figure~\ref{fig:two-handed-square-block}.
   Once these values are computed, they are also used to determine placement of glues.
   The entire block is created by rotating either counter-clockwise (in the case of $x$ even, as shown in the bottom right of Figure~\ref{fig:two-handed-square-block}) or clockwise (in the case of $x$ odd, not shown but the exact mirror image of the bottom right of Figure~\ref{fig:two-handed-square-block}), placing bumps and dents and single-strength glues.
   The glues are placed in the order shown by the rotations, so that the last glue to be (potentially) placed is the top east-facing glue in the case of $x$ even, or the top west-facing glue in the case of $x$ odd.

   By inspection of Figure~\ref{fig:two-handed-square-high-level}, it is routine to verify that the following rules can be used to determine placement of strength-1 glues.
   If $x$ is even, then place two single-strength glues on the right edge.
   If $x$ is odd, then place two single-strength glues on the left edge.
   For a natural number $n$, define $t(n)$ to be the number of trailing 1's in $n$'s binary expansion.
   If $t(x) \geq t(y)$, then place exactly one strength-1 glue on the top edge.
   If $t(x) \geq t(y-1)$, then place exactly one strength-1 glue on the bottom edge.
   If $t(x-1) \leq t(y)+1$ and $x$ is even, then place exactly one strength-1 glue on the left edge.
   If $t(x) \leq t(y)+1$ and $x$ is odd, then place exactly one strength-1 glue on the right edge.

   Each of these computations (for placement of glues and for determining block dimensions) can be computed by a deterministic finite automaton whose input symbols represent tuples of bits from $n-2^k$, $x$, $x-1$, $y-1$, and $y$.
   These automata can then be combined in a product construction and embedded into the tile types that accrete in the row above $n-2^k$ if only $x$ and $y$ are needed, and embedded into tile types that are placed after the first rotation if $x-1$ or $y-1$ is needed.
   Since the decision for placing glue on the top edge requires only $x$ and $y$, this ensures that the decision for each glue placement can be made before the region containing the potential glue site is assembled.

   As shown in Figure~\ref{fig:two-handed-square-block}, some padding with filler tiles is necessary to make the block a perfect rectangle.
   Also, some padding is needed in the case of a doubling of height or width, to ensure that the resulting assembly has height or width precisely twice that of the non-doubled version.
}

\opt{normal}{
    \begin{proof}
    \ProofParallelSquare
    \end{proof}
}

\section{Open Questions}
\label{sec-conclusion}

There are some interesting questions that remain open.
Say that a tile system $\calT$ \emph{strictly self-assembles} a shape $S \subseteq \Z^2$ if all terminal assemblies $\ha$ of $\calT$, appropriately translated, satisfy $\dom \ha = S$.

\begin{enumerate}
  \item What upper or lower bound can be placed on the quantity $\Kddac(\calT)$ for $\calT$ a hierarchical TAS that strictly self-assembles an $n \times n$ square with optimal tile complexity $O(\frac{\log n}{\log \log n})$ (or even with nearly-optimal tile complexity $O(\log n)$)?
      It is not obvious how to show either $\Kddac(\calT) = o(\log^2 n)$ for some such $\calT$ or $\Kddac(\calT) = \omega(\log n)$ for all such $\calT$.
      Obtaining bounds for more general shapes would also be interesting.

  \item What is the complexity of the following decision problems?
  $$
    \textsc{HierDirectedAssembly} = \setr{\pair{\alpha}{\calT}}{
    \begin{array}{l}
    \calT \text{ is a directed hierarchical TAS with}
    \\
    \text{unique producible terminal assembly } \alpha
    \end{array}},
  $$
  $$
    \textsc{HierDirectedShape} = \setr{\pair{S}{\calT}}{
    \begin{array}{l}
    \calT \text{ is a directed hierarchical TAS that}
    \\
    \text{strictly self-assembles finite shape } S
    \end{array}},
  $$
  $$
    \textsc{HierUniqueShape} = \setr{\pair{S}{\calT}}{
    \begin{array}{l}
    \calT \text{ is a hierarchical TAS that}
    \\
    \text{strictly self-assembles finite shape } S
    \end{array}}.
  $$
  In the case of the seeded aTAM, the seeded variants of these problems are known to be in $\P$~\cite{ACGHKMR02} for the first two,
  and $\coNP$-complete~\cite{AGKS05} for the last.

  For the case of \emph{3D} hierarchical tile systems, $\textsc{HierDirectedAssembly}$ was shown to be $\coNP$-complete by Cannon, Demaine, Demaine, Eisenstat, Patitz, Schweller, Summers, and Winslow~\cite{TwoHandsBetterThanOne}.\footnote{In that paper, the problem is called the \emph{Unique Assembly Verification} problem.}
  Furthermore, their proof shows that the 3D version of  $\textsc{HierDirectedShape}$ is $\coNP$-hard.\footnote{Their technique to reduce the complement of $\mathsf{SAT}$ to the problem is such that, if the formula is unsatisfiable, then $\calT$ has a unique terminal assembly $\alpha$ (hence strictly self-assembles the shape $\dom \alpha$), and if the formula is satisfiable, then $\calT$ produces multiple terminal assemblies, and at least two of them are guaranteed to have different shapes.
  Therefore their proof also shows that the 3D version of $\textsc{HierDirectedShape}$ is $\coNP$-hard.
  The proof that $\textsc{HierDirectedAssembly} \in \coNP$ (which applies to any number of dimensions)
  does not so easily apply to $\textsc{HierDirectedShape}$, so the computational complexity of the 2D shape version of the problem is still open.}

  See~\cite{arora2009computational} for definitions of complexity classes $\Sigma^\mathsf{P}_i$ and $\Pi^\mathsf{P}_i$, where $\NP = \Sigma^\mathsf{P}_1$ and $\coNP = \Pi^\mathsf{P}_1$.
  The ``obvious'' containments are
  $\textsc{HierDirectedAssembly} \in \coNP$ (proven in~\cite[Lemma 4.3]{TwoHandsBetterThanOne}),
  $\textsc{HierUniqueShape} \in \Pi^\mathsf{P}_2$\footnote{
      The producibility of $\alpha$ is decidable in polynomial time~\cite{DotyProducibilityNaCo}.
      For $k\in\N$, Let $\prodasm{\calT}_{\leq k} = \{ \alpha \in \prodasm{\calT} \mid |\alpha| \leq k\}$.
      Then $\pair{S}{\calT} \in \textsc{HierUniqueShape}$​ if and only if
      for all $\alpha \in \prodasm{\calT}_{\leq 2 |S|}$:
      \begin{itemize}
        \item $|\alpha| \leq |S|$ (if this is verified for all $\alpha\in\prodasm{\calT}_{\leq 2|S|}$ then no assembly larger than $2|S|$ is producible either),
        \item if $|\alpha| < |S|$​ then there exists $\gamma \in \prodasm{\calT}_{\leq |S|}$​ attachable to $\alpha$ (so $\alpha \not\in \termasm{\calT}$), and
        \item if $|\alpha| = |S|$​ (so $\alpha \in \termasm{\calT}$ since nothing larger than $S$​ is producible) then $S = \dom \alpha$.
      \end{itemize}
      The second condition is a $\forall \exists$​ quantifier that makes the problem in $\Pi^\mathsf{P}_2$; the other conditions have only one $\exists$ or $\forall$.
      Note that the second check guarantees no assembly strictly smaller than $|S|$ is terminal, 
      and the first check guarantees that no assembly strictly larger than $|S|$ is producible.
      Therefore there must be at least one $\alpha \in \termasm{\calT}$ with $|\alpha| = |S|$ 
      (and the third check guarantees that it has shape $S$).
  }
  and
  $\textsc{HierDirectedShape} \in \Pi^\mathsf{P}_2$.\footnote{
      Using similar reasoning as above​, we have $\pair{S}{\calT} \in \textsc{HierDirectedShape}$ if and only if
      for all $\alpha \in \prodasm{\calT}_{\leq 2|S|}$:
      \begin{itemize}
        \item $|\alpha| \leq |S|$,
        \item if $|\alpha| < |S|$​ then there exists $\gamma \in \prodasm{\calT}_{\leq |S|}$​ attachable to $\alpha$ (so $\alpha \not\in \termasm{\calT}$), and
        \item for all $\beta \in \prodasm{\calT}_{\leq |S|}$, if $|\alpha| = |\beta| = |S|$​ (so $\alpha,\beta \in \termasm{\calT}$) then $\alpha = \beta$ and $S = \dom \alpha$.
        \end{itemize}
  }
  It is open whether $\textsc{HierDirectedAssembly}$ is $\coNP$-hard (in 2D),
  and whether $\textsc{HierUniqueShape}$ or $\textsc{HierDirectedShape}$ are $\Pi^\mathsf{P}_2$-hard.
  The proof of $\textsc{HierUniqueShape}$ in the ``multiple-tile'' model of~\cite{AGKS05} can be used to show that $\textsc{HierUniqueShape}$ (in the hierarchical aTAM as defined in this paper) is $\coNP$-hard~\cite{AGKS05, SchwellerPersonalComm2016}.

  \item What is the complexity of the following decision problems?
  $$\begin{array}{l}
    \textsc{HierMinTileSet} =
    \\
    \setr{\pair{S}{c}}{
    \begin{array}{l}
    (\exists \calT=(T,\tau))\ \calT \text{ is a hierarchical TAS with}
    \\
    \text{$|T| \leq c$ and $\calT$ strictly self-assembles finite shape $S$}
    \end{array}}
    \end{array},
  $$
  $$\begin{array}{l}
    \textsc{HierDirectedMinTileSet} =
    \\
    \setr{\pair{S}{c}}{
    \begin{array}{l}
    (\exists \calT=(T,\tau))\ \calT \text{ is a directed hierarchical TAS with}
    \\
    \text{$|T| \leq c$ and $\calT$ strictly self-assembles finite shape $S$}
    \end{array}}
    \end{array}.
  $$
  In the case of the seeded aTAM, the seeded variants of these problems are known to be $\Sigma^\mathsf{P}_2$-complete~\cite{BryChiDotKarSekJournal} and $\NP$-complete~\cite{ACGHKMR02}, respectively.

  \item What is the optimal time complexity of strictly self-assembling an $n \times n$ square with a hierarchical TAS? 
      Any shape with diameter $n$?
      What if we require the TAS to be directed?

  \item Two asymptotically unrealistic aspects of the model are the assumption of a constant rate of diffusion of assemblies and a constant binding strength threshold required to bind two assemblies together.
      Large assemblies will diffuse more slowly in a well-mixed solution; some simple models predict that the diffusion rate of a molecule is inversely proportional to its diameter~\cite{riseman1950intrinsic,berg1985diffusion}.
      It is conceivable that an assembly model properly accounting for diffusion rates could enforce an absolute lower bound of $\Omega(D)$ on the assembly time required to assemble any shape of diameter $D$.

      The binding strength threshold of the seeded aTAM is a simplified model of a more complicated approximation in the \emph{kinetic} tile assembly model (kTAM, \cite{Winf98}).
      Tiles in reality will occasionally detach, but so long as their concentration  is sufficient, another tile will reattach after not too much time.
      While our model accounts directly for concentrations of large assemblies, it only accounts for this concentration up to the moment of first binding.
      A more realistic model might require a larger binding strength threshold to balance the fact that, if a large assembly detaches, it may take a long time to reattach.
      In particular, the seeded aTAM is justifiable as a model, despite its lack of reverse reactions or modeling of strength-1 attachments (which happen in reality but have a higher reverse rate than higher strength attachments), in part due to Winfree's proof~\cite{Winfree98simulationsof,Winf98} that under suitable conditions (in particular setting concentrations and binding energies such that the rate of forward attachments is just barely larger than the rate of backward detachments of strength-2-bound tiles), the kTAM ``simulates the aTAM with high probability.''
      It is an open question whether there is any similar theorem that can be proven in the hierarchical aTAM, showing that detachment reactions may be safely ignored under certain conditions.

      Incorporating these and other physical phenomena into the hierarchical assembly model would be an interesting challenge.

\end{enumerate}

\paragraph{Acknowledgement.}
The authors are especially grateful to David Soloveichik for contributions to this paper, including the proof of Theorem~\ref{thm-seeded-time-lower-bound}, discussion and insights on other proofs, and generally for indispensable help with discovering and solidifying the results.
The authors also thank Adam Marblestone, Robbie Schweller, Matt Patitz, and the members of Erik Winfree's group, particularly Joe Schaeffer, Erik Winfree, Damien Woods, and Seung Woo Shin, for insightful discussion and comments, and to Bernie Yurke and Rizal Hariadi for discussing diffusion rates and pointing the authors to~\cite{riseman1950intrinsic,berg1985diffusion}.
We are grateful to anonymous reviewers for identifying important flaws in the original proofs of some theorems, one of which spurred significant followup work~\cite{poirghtsJoCG}.

\opt{submission}{
}

\bibliographystyle{plain}
\bibliography{tam}

\opt{normal}{
    \appendix
    \section{Appendix: Formal Definition of Abstract Tile Assembly Model}
    \label{sec-tam-formal}
    
\newcommand{\fullgridgraph}{G^\mathrm{f}}
\newcommand{\bindinggraph}{G^\mathrm{b}}

This section gives a terse definition of the abstract Tile Assembly Model (aTAM,~\cite{Winf98}). This is not a tutorial; for readers unfamiliar with the aTAM,  \cite{RotWin00} gives an excellent introduction to the model.

Fix an alphabet $\Sigma$. $\Sigma^*$ is the set of finite strings over $\Sigma$.
Given a discrete object $O$, $\langle O \rangle$ denotes a standard encoding of $O$ as an element of $\Sigma^*$.
$\Z$, $\Z^+$, $\N$, $\R^+$ denote the set of integers, positive integers, nonnegative integers, and nonnegative real numbers, respectively.
For a set $A$, $\calP(A)$ denotes the power set of $A$.
Given $A \subseteq \Z^2$, the \emph{full grid graph} of $A$ is the undirected graph $\fullgridgraph_A=(V,E)$, where $V=A$, and for all $u,v\in V$, $\{u,v\} \in E \iff \| u-v\|_2 = 1$; i.e., if and only if $u$ and $v$ are adjacent on the integer Cartesian plane.
A \emph{shape} is a set $S \subseteq \Z^2$ such that $\fullgridgraph_S$ is connected.

A \emph{tile type} is a tuple $t \in (\Sigma^* \times \N)^4$; i.e., a unit square with four sides listed in some standardized order, each side having a \emph{glue label} (a.k.a. \emph{glue}) $\ell \in \Sigma^*$ and a nonnegative integer \emph{strength}.
For a set of tile types $T$, let $\Lambda(T) \subset \Sigma^*$ denote the set of all glue labels of tile types in $T$.
Let $\dall$ denote the \emph{directions} consisting of unit vectors $\{(0,1), (0,-1), (1,0), (-1,0)\}$.
Given a tile type $t$ and a direction $d \in \dall$, $t(d) \in \Lambda(T)$ denotes the glue label on $t$ in direction $d$.
We assume a finite set $T$ of tile types, but an infinite number of copies of each tile type, each copy referred to as a \emph{tile}. An \emph{assembly}
is a nonempty connected arrangement of tiles on the integer lattice $\Z^2$, i.e., a partial function $\alpha:\Z^2 \dashrightarrow T$ such that $\fullgridgraph_{\dom \alpha}$ is connected and $\dom \alpha \neq \emptyset$.
The \emph{shape of $\alpha$} is $\dom \alpha$.
Write $|\alpha|$ to denote $|\dom\alpha|$.
Given two assemblies $\alpha,\beta:\Z^2 \dashrightarrow T$, we say $\alpha$ is a \emph{subassembly} of $\beta$, and we write $\alpha \sqsubseteq \beta$, if $\dom \alpha \subseteq \dom \beta$ and, for all points $p \in \dom \alpha$, $\alpha(p) = \beta(p)$.

Given two assemblies $\alpha$ and $\beta$, we say $\alpha$ and $\beta$ are \emph{equivalent up to translation}, written $\alpha \simeq \beta$, if there is a vector $\vec{x} \in \Z^2$ such that $\dom\alpha = \dom\beta + \vec{z}$ (where for $A \subseteq \Z^2$, $A + \vec{z}$ is defined to be $\setr{p + \vec{z}}{p \in A}$) and for all $p \in \dom\beta$, $\alpha(p + \vec{z}) = \beta(p)$.
In this case we say that $\beta$ is a \emph{translation} of $\alpha$.
Given a shape $S \subseteq \Z^2$, we say that $S$ is \emph{canonical} if $S \subseteq \N^2$, $(x,0) \in S$ for some $x\in\N$, and $(0,y) \in S$ for some $y\in\N$.
In other words, $S$ is located entirely in the first quadrant, but at far to the left and down as possible.
We say an assembly $\alpha$ is \emph{canonical} if $\dom\alpha$ is canonical.
For each finite assembly $\alpha$, there is exactly one canonical assembly $\wa$ such that $\alpha \simeq \wa$.
Given such a finite $\alpha$, we say $\wa$ is the \emph{canonical assembly of $\alpha$}.
For brevity and clarity, we will tend to abuse notation and speak of assemblies equivalent up to translation as if they are the same object, often taking $\alpha$ to mean $\wa$, particularly when discussing concentrations.
We have fixed assemblies at certain positions on $\Z^2$ only for mathematical convenience in some contexts, but of course real assemblies float freely in solution and do not have a fixed position.

Let $\alpha$ be an assembly and let $p\in\dom\alpha$ and $d\in\dall$ such that $p + d \in \dom\alpha$.
Let $t=\alpha(p)$ and $t' = \alpha(p+d)$.
We say that the tiles $t$ and $t'$ at positions $p$ and $p+d$ \emph{interact} if $t(d) = t'(-d)$ and $g(t(d)) > 0$, i.e., if the glue labels on their abutting sides are equal and have positive strength.
Each assembly $\alpha$ induces a \emph{binding graph} $\bindinggraph_\alpha$, a grid graph $G=(V_\alpha,E_\alpha)$, where $V_\alpha=\dom\alpha$, and $\{p_1,p_2\} \in E_\alpha \iff \alpha(p_1) \text{ interacts with } \alpha(p_2)$.\footnote{For $\fullgridgraph_{\dom \alpha}=(V_{\dom \alpha},E_{\dom \alpha})$ and $\bindinggraph_\alpha=(V_\alpha,E_\alpha)$, $\bindinggraph_\alpha$ is a spanning subgraph of $\fullgridgraph_{\dom \alpha}$: $V_\alpha = V_{\dom \alpha}$ and $E_\alpha \subseteq E_{\dom \alpha}$.}
Given $\tau\in\Z^+$, $\alpha$ is \emph{$\tau$-stable} if every cut of $\bindinggraph_\alpha$ has weight at least $\tau$, where the weight of an edge is the strength of the glue it represents.
That is, $\alpha$ is $\tau$-stable if at least energy $\tau$ is required to separate $\alpha$ into two parts.
When $\tau$ is clear from context, we say $\alpha$ is \emph{stable}.

\opt{submission}{\paragraph{Seeded aTAM.}}
\opt{normal}{\subsection{Seeded aTAM}}
A \emph{seeded tile assembly system} (seeded TAS) is a triple $\calT = (T,\sigma,\tau)$, where $T$ is a finite set of tile types, $\sigma:\Z^2 \dashrightarrow T$ is the finite, $\tau$-stable \emph{seed assembly},
and $\tau\in\Z^+$ is the \emph{temperature}.
Given two $\tau$-stable assemblies $\alpha,\beta:\Z^2 \dashrightarrow T$, we write $\alpha \to_1^{\calT} \beta$ if $\alpha \sqsubseteq \beta$ and $|\dom \beta \setminus \dom \alpha| = 1$. In this case we say $\alpha$ \emph{$\calT$-produces $\beta$ in one step}.\footnote{Intuitively $\alpha \to_1^\calT \beta$ means that $\alpha$ can grow into $\beta$ by the addition of a single tile; the fact that we require both $\alpha$ and $\beta$ to be $\tau$-stable implies in particular that the new tile is able to bind to $\alpha$ with strength at least $\tau$. It is easy to check that had we instead required only $\alpha$ to be $\tau$-stable, and required that the cut of $\beta$ separating $\alpha$ from the new tile has strength at least $\tau$, then this implies that $\beta$ is also $\tau$-stable.}
If $\alpha \to_1^{\calT} \beta$, $ \dom \beta \setminus \dom \alpha=\{p\}$, and $t=\beta(p)$, we write $\beta = \alpha + (p \mapsto t)$.
The \emph{$\calT$-frontier} of $\alpha$ is the set $\partial^\calT \alpha = \bigcup_{\alpha \to_1^\calT \beta} \dom \beta \setminus \dom \alpha$, the set of empty locations at which a tile could stably attach to $\alpha$.

A sequence of $k\in\Z^+ \cup \{\infty\}$ assemblies $\vec{\alpha} = (\alpha_0,\alpha_1,\ldots)$ is a \emph{$\calT$-assembly sequence} if, for all $1 \leq i < k$, $\alpha_{i-1} \to_1^\calT \alpha_{i}$.
We write $\alpha \to^\calT \beta$, and we say $\alpha$ \emph{$\calT$-produces} $\beta$ (in 0 or more steps) if there is a $\calT$-assembly sequence $\vec{\alpha}=(\alpha_0,\alpha_1,\ldots)$ of length $k = |\dom \beta \setminus \dom \alpha| + 1$ such that
1) $\alpha = \alpha_0$,
2) $\dom \beta = \bigcup_{0 \leq i < k} \dom {\alpha_i}$, and
3) for all $0 \leq i < k$, $\alpha_{i} \sqsubseteq \beta$.
In this case, we say that $\beta$ is the \emph{result} of $\vec{\alpha}$, written $\beta=\res{\vec{\alpha}}$.
If $k$ is finite then it is routine to verify that $\res{\vec{\alpha}} = \alpha_{k-1}$.\footnote{If we had defined the relation $\to^\calT$ based on only finite assembly sequences, then $\to^\calT$ would be simply the reflexive, transitive closure $(\to_1^\calT)^*$ of $\to_1^\calT$. But this would mean that no infinite assembly could be produced from a finite assembly, even though there is a well-defined, unique ``limit assembly'' of every infinite assembly sequence.}
We say $\alpha$ is \emph{$\calT$-producible} if $\sigma \to^\calT \alpha$, and we write $\prodasm{\calT}$ to denote the set of $\calT$-producible canonical assemblies.
The relation $\to^\calT$ is a partial order on $\prodasm{\calT}$ \cite{Roth01,jSSADST}.\footnote{In fact it is a partial order on the set of $\tau$-stable assemblies, including even those that are not $\calT$-producible.}
A $\calT$-assembly sequence $\alpha_0,\alpha_1,\ldots$ is \emph{fair} if, for all $i$ and all $p\in\partial^\calT\alpha_i$, there exists $j$ such that $\alpha_j(p)$ is defined; i.e., no frontier location is ``starved''.

An assembly $\alpha$ is \emph{$\calT$-terminal} if $\alpha$ is $\tau$-stable and $\partial^\calT \alpha=\emptyset$.
It is easy to check that an assembly sequence $\vec{\alpha}$ is fair if and only $\res{\vec{\alpha}}$ is terminal.
We write $\termasm{\calT} \subseteq \prodasm{\calT}$ to denote the set of $\calT$-producible, $\calT$-terminal canonical assemblies.

A seeded TAS $\calT$ is \emph{directed (a.k.a., deterministic, confluent)} if the poset $(\prodasm{\calT}, \to^\calT)$ is directed; i.e., if for each $\alpha,\beta \in \prodasm{\calT}$, there exists $\gamma\in\prodasm{\calT}$ such that $\alpha \to^\calT \gamma$ and $\beta \to^\calT \gamma$.\footnote{The following two convenient characterizations of ``directed'' are routine to verify.
$\calT$ is directed if and only if $|\termasm{\calT}| = 1$.
$\calT$ is \emph{not} directed if and only if there exist $\alpha,\beta\in\prodasm{\calT}$ and $p \in \dom \alpha \cap \dom \beta$ such that $\alpha(p) \neq \beta(p)$.}
We say that a TAS $\calT$ \emph{strictly self-assembles} a shape $S \subseteq \Z^2$ if, for all $\alpha \in \termasm{\calT}$, $\dom \alpha = S$; i.e., if every terminal assembly produced by $\calT$ has shape $S$.
If $\calT$ strictly self-assembles some shape $S$, we say that $\calT$ is \emph{strict}.
Note that the implication ``$\calT$ is directed $\implies$ $\calT$ is strict'' holds, but the converse does not hold.

\opt{submission}{\paragraph{Hierarchical aTAM.}}
\opt{normal}{\subsection{Hierarchical aTAM}}
A \emph{hierarchical tile assembly system} (hierarchical TAS) is a pair $\calT = (T,\tau)$, where $T$ is a finite set of tile types,
and $\tau\in\Z^+$ is the \emph{temperature}.
Let $\alpha,\beta:\Z^2 \dashrightarrow T$ be two (possibly non-canonical) assemblies.
Say that $\alpha$ and $\beta$ are \emph{nonoverlapping} if $\dom\alpha \cap \dom\beta = \emptyset$.
If $\alpha$ and $\beta$ are nonoverlapping assemblies, define $\alpha \cup \beta$ to be the assembly $\gamma$ defined by $\gamma(p) = \alpha(p)$ for all $p\in\dom\alpha$, $\gamma(p) = \beta(p)$ for all $p \in \dom \beta$, and $\gamma(p)$ is undefined for all $p \in \Z^2 \setminus (\dom\alpha \cup \dom\beta)$.
An assembly $\gamma$ is \emph{singular} if $\gamma(p)=t$ for some $p\in\Z^2$ and some $t \in T$ and $\gamma(p')$ is undefined for all $p' \in \Z^2 \setminus \{p\}$.
Given a hierarchical TAS $\calT=(T,\tau)$, an assembly $\gamma$ is \emph{$\calT$-producible} if either 1) $\gamma$ is singular, or 2) there exist producible nonoverlapping assemblies $\alpha$ and $\beta$ such that $\gamma = \alpha\cup\beta$ and $\gamma$ is $\tau$-stable. 
In the latter case, write $\alpha + \beta \to \gamma$.
An assembly $\alpha$ is \emph{$\calT$-terminal} if for every producible assembly $\beta$ such that $\alpha$ and $\beta$ are nonoverlapping, $\alpha \cup \beta$ is not $\tau$-stable.\footnote{The restriction on overlap is a model of a chemical phenomenon known as \emph{steric hindrance}~\cite[Section 5.11]{WadeOrganicChemistry91} or, particularly when employed as a design tool for intentional prevention of unwanted binding in synthesized molecules, \emph{steric protection}~\cite{HellerPugh1,HellerPugh2,GotEtAl00}.}
Define $\prodasm{\calT}$ to be the set of all $\calT$-producible canonical assemblies.
Define $\termasm{\calT} \subseteq \prodasm{\calT}$ to be the set of all $\calT$-producible, $\calT$-terminal canonical assemblies.
A hierarchical TAS $\calT$ is \emph{directed (a.k.a., deterministic, confluent)} if $|\termasm{\calT}|=1$. 
We say that a TAS $\calT$ \emph{strictly self-assembles} a shape $S \subseteq \Z^2$ if, for all $\alpha \in \termasm{\calT}$, $\dom \alpha = S$; i.e., if every terminal assembly produced by $\calT$ has shape $S$.

Let $\calT$ be a hierarchical TAS, and let $\alpha \in \prodasm{\calT}$ be a $\calT$-producible assembly.
An \emph{assembly tree} $\Upsilon$ of $\widehat{\alpha}$ is a full binary tree with $|\widehat{\alpha}|$ leaves, whose nodes are labeled by $\calT$-producible assemblies, with $\widehat{\alpha}$ labeling the root, singular assemblies labeling the leaves, and node $u$ labeled with $\gamma$ having children $u_1$ labeled with $\alpha$ and $u_2$ labeled with $\beta$, with the requirement that $\alpha + \beta \to \gamma$.
That is, $\Upsilon$ represents one possible pathway through which $\widehat{\alpha}$ could be produced from individual tile types in $\calT$.
Let $\Upsilon(\calT)$ denote the set of all assembly trees of $\calT$.
Say that an assembly tree is \emph{$\calT$-terminal} if its root is a $\calT$-terminal assembly.
Let $\Upsilon_\Box(\calT)$ denote the set of all $\calT$-terminal assembly trees of $\calT$.
Note that even a directed hierarchical TAS can have multiple terminal assembly trees that all have the same root terminal assembly.

When $\calT$ is clear from context, we may omit $\calT$ from the notation above and instead write
$\to_1$,
$\to$,
$\partial \alpha$,
\emph{frontier},
\emph{assembly sequence},
\emph{produces},
\emph{producible}, and
\emph{terminal}.
We also assume without loss of generality that every positive-strength glue occurring in some tile type in some direction also occurs in some tile type in the opposite direction, i.e., there are no ``effectively null'' positive-strength glues.

}

\end{document}